%% file: traversal_main.tex
\documentclass[3p]{elsarticle}
\pdfoutput=1
 
\usepackage{amsmath}
\usepackage{amssymb}
\usepackage{amsthm}
\usepackage{url}

\usepackage[dvipsnames]{xcolor}
\usepackage{tikz}
\usepackage{algorithmicx,algorithm}
\usepackage[noend]{algpseudocode}
\usepackage{xspace}

\newlength{\figurewidth}
\newlength{\smallfigurewidth}

\newcounter{thmcounter}
\numberwithin{thmcounter}{section}

\newtheorem{theorem}[thmcounter]{Theorem}

\newtheorem{lemma}[thmcounter]{Lemma}

\theoremstyle{definition}
\newtheorem{definition}[thmcounter]{Definition}

\theoremstyle{remark}
\newtheorem*{remark}{Remark}

\usetikzlibrary{arrows,arrows.meta,shapes,positioning,calc,patterns,spy,decorations.pathmorphing,decorations,decorations.pathreplacing,fit,fadings,backgrounds}
 
\newcommand{\rank}{\mathsf{rank}}
\newcommand{\att}{\mathsf{att}}
\newcommand{\lbl}{\mathsf{lab}}
\newcommand{\tlbl}{\lambda}
\newcommand{\rt}{\mathsf{root}}
\newcommand{\ext}{\mathsf{ext}}
\newcommand{\N}{\mathbb{N}}
\newcommand{\HGR}{\mathsf{HGR}}
\newcommand{\rhs}{\mathsf{rhs}}
\newcommand{\val}{\mathsf{val}}
\newcommand{\nodes}{\mathsf{i\text{-}nodes}}
\newcommand{\dt}{\mathsf{dt}}
\newcommand{\internal}{\mathsf{internal}}
\newcommand{\outEdge}{\mathsf{traverse}}
\newcommand{\NTSet}{E^{\nt}}
\newcommand{\sib}{\ensuremath{\mathsf{sib}}}
\newcommand{\firstID}{\ensuremath{\mathsf{first}}}
\newcommand{\cur}{\ensuremath{\mathsf{cNode}}}
\newcommand{\tab}{\ensuremath{t}}
\newcommand{\curTab}{\ensuremath{\mathsf{cTab}}\xspace}
\newcommand{\maxRank}{\ensuremath{\kappa}\xspace}
\newcommand{\pos}{\ensuremath{\mathsf{pos}}\xspace}

\newcommand{\reg}{\ensuremath{\mathsf{nxt}}\xspace}
\newcommand{\nde}{\ensuremath{\mathsf{vtx}}\xspace}

\newcommand{\id}{\ensuremath{\mathsf{vertex}}\xspace}
\newcommand{\nt}{\ensuremath{\mathsf{nt}}\xspace}

\newcommand{\lhs}{\ensuremath{\mathsf{lhs}}\xspace}
\newcommand{\height}{\ensuremath{\mathsf{height}}\xspace}
\newcommand{\SLOrder}{\ensuremath{\leq_{\mathsf{NT}}}\xspace}
\newcommand{\nr}{\ensuremath{\rho}\xspace}
\newcommand{\findNode}{\ensuremath{\mathsf{findVertex}}\xspace}
\newcommand{\attach}{\ensuremath{\mathsf{attach}}\xspace}
\newcommand{\offset}{\ensuremath{\mathsf{off}}\xspace}
\newcommand{\tbl}{\ensuremath{\mathsf{tab}}\xspace}
\newcommand{\row}{\ensuremath{i}\xspace}
\newcommand{\success}{\ensuremath{\mathsf{succ}}\xspace}

\newcommand{\bitLen}{\ensuremath{\beta}}

\newcommand{\lblsize}{\small}
\newcommand{\idsize}{\small}

\newcommand{\tabBox}[4]{
    \coordinate[#1] (#2o1) {};
    \coordinate[above=35pt of #2o1] (#2o2) {};
    \coordinate[right=50pt of #2o2] (#2o3) {};
    \coordinate[below=35pt of #2o3] (#2o4) {};
    \coordinate[above right=5pt and 5pt of #2o1] (#2i1) {};
    \coordinate[above=20pt of #2i1] (#2i2) {};
    \coordinate[right=20pt of #2i2] (#2i3) {};
    \coordinate[right=20pt of #2i3] (#2i4) {};
    \coordinate[below=20pt of #2i4] (#2i5) {};
    \coordinate[left=20pt of #2i5] (#2i6) {};
    \coordinate[below=8pt of #2i2] (#2x1) {};
    \coordinate[below=8pt of #2i4] (#2x2) {};

    \coordinate[above=17.5pt of #2o1] (#2ol);
    \coordinate[below=17.5pt of #2o3] (#2or);
    \coordinate[above=10pt of #2i1] (#2il);
    \coordinate[below=10pt of #2i4] (#2ir);

    \draw (#2o1) -- (#2o2) -- (#2o3) -- (#2o4) -- (#2o1);
    \draw (#2i1) -- (#2i2) -- (#2i3) -- (#2i4) -- (#2i5) -- (#2i6) -- (#2i1);
    \draw (#2i3) -- (#2i6);
    \draw[dotted] (#2x1) -- (#2x2);

    \node[below right=-2pt and -2pt of #2o2] (#2idx) {\scriptsize $#3$};
    \node[below right=-1pt and -2pt of #2i2, anchor=north west] (#2nxt) {\scriptsize $\reg$};
    \node[below right=-2pt and -2pt of #2i3, anchor=north west] (#2nde) {\scriptsize $\nde$};
    \node[below right=8pt and 4pt of #2i3] (#2node) {\scriptsize $#4$};
    \coordinate[below right=15pt and 10pt of #2i2] (#2p) {};
    \coordinate[below right=13pt and 7pt of  #2i2] (#2pl) {};
    \coordinate[below right=13pt and 13pt of #2i2] (#2pr) {};
}

\algrenewcommand\textproc{\textsf}

\tikzset{
    ncbar angle/.initial=90,
    ncbar/.style={
        to path=(\tikztostart)
        -- ($(\tikztostart)!#1!\pgfkeysvalueof{/tikz/ncbar angle}:(\tikztotarget)$)
        -- ($(\tikztotarget)!($(\tikztostart)!#1!\pgfkeysvalueof{/tikz/ncbar angle}:(\tikztotarget)$)!\pgfkeysvalueof{/tikz/ncbar angle}:(\tikztostart)$)
        -- (\tikztotarget)
    },
    ncbar/.default=0.5cm,
}

\tikzset{square left brace/.style={ncbar=0.2cm}}
\tikzset{square right brace/.style={ncbar=-0.2cm}}

\begin{document}
\title{Constant Delay Traversal of Grammar-Compressed Graphs with Bounded Rank}

\author[bre]{Sebastian Maneth\corref{cor1}}
\ead{maneth@uni-bremen.de}

\author[oxf]{Fabian Peternek}
\ead{fabian.peternek@cs.ox.ac.uk}

\cortext[cor1]{Corresponding author}

\address[bre]{Universit\"{a}t Bremen, PO Box 330 440, 283334 Bremen, Germany}
\address[oxf]{University of Oxford, Wolfson Building, Parks Road, Oxford, OX1 3QD, United Kingdom}

\begin{abstract}
  We present a pointer-based data structure for constant time traversal of the edges of an
  edge-labeled (alphabet $\Sigma$) directed hypergraph (a graph where edges can be incident to more
  than two vertices, and the incident vertices are ordered) given as hyperedge-replacement grammar $G$.
  It is assumed  that the grammar has a fixed rank $\maxRank$ (maximal number of vertices connected to a
  nonterminal hyperedge) and that each vertex of the represented graph is incident to at most one
  $\sigma$-edge per direction ($\sigma \in \Sigma$). Precomputing the data structure needs
  $O(|G||\Sigma|\maxRank r h)$ space and $O(|G||\Sigma|\maxRank r h^2 )$ time, where $h$ is the
  height of the derivation tree of $G$ and $r$ is the maximal rank of a terminal edge occurring in
  the grammar.
\end{abstract}

\maketitle
\thispagestyle{empty}

\input{new_introduction}
\input{preliminaries}
\input{SLHRGrammars}
\input{traversing}
\input{limits}
\input{conclusions}
\section*{Acknowledgements}
This project has received funding from the European Union's Horizon 2020 research and innovation
programme under grant agreement No 682588. The authors thank the anonymous reviewers for their many
helpful comments.

\section*{References}
\bibliographystyle{elsarticle-num-names-alphasort}
\bibliography{refs}
\end{document}

%% file: new_introduction.tex
\section{Introduction}

Straight-line grammars are context-free grammars that produce one object only (see,
e.g.~\cite{DBLP:reference/algo/Bannai16,loh12,man18}). The term ``straight-line'' refers to the
restriction that such grammars have no recursion (i.e., there is no derivation from a nonterminal to
itself) and have no non-determinism (i.e., for each nonterminal there is exactly one rule). The
most well-known type of straight-line (SL) grammars are string grammars which are also known as
\emph{straight-line programs} (SLPs). SLPs are a convenient formalism for data compression due to
several reasons: they generalize a number of classical compression schemes, such as Lempel-Ziv and
they are mathematically clean and simple, which has fostered the development of a wealth of
algorithms for compressed data. Note that an SLP can be exponentially smaller than the object it
represents. For many problems algorithms exist that run over the compressed (grammar-based)
representation within a similar time bound (say, polynomial time) as over the original uncompressed
object; examples include pattern matching, testing equivalence, and evaluating simple queries (e.g.,
queries expressed by finite state automata), and more as given in a survey~\cite{loh12}. SLPs have been
generalized to trees~\cite{DBLP:journals/is/BusattoLM08,DBLP:journals/is/LohreyMM13}
(see the survey~\cite{DBLP:conf/dlt/Lohrey15}) by the use of so called ``straight-line linear
context-free tree grammars'', for short \emph{SL tree grammars}. In an SL tree grammar,
nonterminals may have argument trees, i.e., nonterminals may appear at inner nodes of the right-hand
side of a rule. A rule with right-hand side $t$ is applied to such an inner node, by replacing the
node by $t$ and by afterwards inserting copies of the subtrees of the node at leaves of $t$ that are
labeled by special place-holders called ``parameters''. In SL tree grammars, each parameter may
occur at most once, so that at most one copy of a subtree is inserted. This \emph{linearity}
restriction guarantees that at most exponential compression ratios are achieved (non-linear SL
context-free tree grammars can have double-exponential ratios; such grammars have not been studied
much~\cite{DBLP:journals/jcss/LohreyMS12}). Note that an SL tree grammar that produces a
monadic tree is an SLP. 

Recently, SLPs have been generalized to (hyper) graphs~\cite{DBLP:journals/is/ManethP18} by using
\emph{straight-line context-free hyperedge replacement
grammars}~\cite{DBLP:conf/gg/DrewesKH97,eng97}, for short, \emph{SL HR grammars}. In such a
grammar, hyperedges labeled by nonterminals are replaced by hypergraphs. If such a hyperedge is
incident with the vertices $v_1,\dots,v_k$ (in that order, i.e., hyperedges are ordered), then the
replacement consists of 
\begin{enumerate}
  \item removing the nonterminal hyperedge, 
  \item inserting a disjoint copy of the right-hand side $g$ of the rule, and 
  \item identifying the $j$-th ``external vertex'' of $g$ with $v_j$.
\end{enumerate}
Thus, external vertices are used in a similar way as the parameters in a tree grammar (in a linear
tree grammar, to be precise). In fact, every SL tree grammar can be seen as a particular tree
generating SL HR grammar~\cite[Section 4.4]{DBLP:journals/is/ManethP18}. SL HR grammars are more
complex than SL tree grammars: e.g., it is not known yet whether isomorphism of the graphs
represented by two SL HR grammars is decidable in polynomial time (for SLPs this is well-known and
can easily be used to solve also the equivalence problem of SL tree grammars in polynomial time).

Let us now consider the traversal of objects represented by SL grammars. E.g., in an SLP, we would
like to traverse the letters of the represented string from left-to-right. If implemented naively, a
single traversal step may require $O(h)$ time, where $h$ is the height of the derivation tree of the
grammar: the current letter is produced by the right-most leaf of a derivation subtree (say, a
subtree of the root node), and the next letter to the right is produced by the left-most leaf of the
next derivation subtree. We would like to have a data structure which allows to carry out any
traversal step in constant time, and which can be computed from the grammar in linear (or
polynomial) time. This problem was solved for SLPs by Gasieniec et
al.~\cite{DBLP:conf/dcc/GasieniecKPS05} and was later extended to SL tree
grammars~\cite{DBLP:journals/algorithmica/LohreyMR18}. The idea of these solutions is to generate
terminal objects only in the leaves of the derivation tree and to represent the leaves of the
derivation tree by certain ``left-most derivations''. Such derivations allow to move from one leaf
of the derivation tree to the next leaf (or the previous one) in constant time (see also
Section~\ref{sss:naive_delay} below).
An
independent method provides navigation and certain other tree queries in time $O(\log |t|)$ for an
unranked tree $t$, and is able to represent tree nodes in $O(\log |t|)$
space~\cite{DBLP:journals/siamcomp/BilleLRSSW15}.

One drawback of the left-most-derivation approach is that it is hard to adapt it to more general
grammar formalisms. For instance, we were not able to generalize the approach to SL graph grammars.
Here, the difficulty arises from the fact that we can \emph{not} assume that terminal objects are
only produced in leaves of the derivation tree. In fact, also for SL tree grammars the
left-most-derivation approach can only be applied after bringing the grammar in a particular
(``string like'') normal form where each nonterminal uses at most one
parameter~\cite{DBLP:journals/jcss/LohreyMS12}. It remains open in the
paper~\cite{DBLP:journals/algorithmica/LohreyMR18} and completely unclear how to apply the
left-most-derivation approach to SL tree grammars with nonterminals using several parameters. This
is unfortunate because the normal form causes a blow-up in grammar size that can be bounded by a
factor of $r$, where $r$ is the maximal rank of terminal
symbols~\cite{DBLP:journals/corr/abs-1902-03568,DBLP:journals/iandc/JezL16,DBLP:journals/jcss/LohreyMS12}.

We present a new data structure for constant delay traversal of hypergraphs represented by SL HR
grammars. A single traversal step consists of moving from a vertex $x$ along a direction of a
hyperedge $e$ to a vertex $y$. A direction is given by a pair $(i,j)$ of integers and means the
following: $x$ must be the $i$-th incident vertex of $e$, and $y$ must be the $j$-th incident vertex
of $e$. Our graph grammars must obey the following restriction: each vertex of the represented
graph is incident with at most one $\sigma$-labeled edge per direction, for every edge label
$\sigma\in\Sigma$. Thus, the graphs have a vertex degree that is bounded by $O(|\Sigma|r)$, where
$r$ is the maximal number of nodes incident to any given edge. This is a strong restriction. We
discuss whether it is necessary in Section~\ref{sse:limits}. We note that any tree grammar is
covered by this restriction and our data structure can therefore be used to navigate arbitrary SL
tree grammars (in particular ones not in normal form). 

Let us consider our solution in more detail. A rule of an SL graph grammar is of the form $A\to g$
where $A$ is a nonterminal of rank $\maxRank$, for some natural number $\maxRank$, and $g$ is a
hypergraph with $\maxRank$ external vertices (this is an ordered set of specially marked vertices,
as described earlier). The rule is applied to a hyperedge labeled $A$ by replacing the edge by the
graph $g$ and identifying the $i$-th vertex of the edge with the $i$-th external vertex of $g$.
This means that the vertices identified with external vertices have already been produced earlier in
the derivation; more precisely, they have been produced by rules where these vertices are
\emph{internal} (= not external). The key idea of our constant delay traversal is to maintain
mappings from external vertices to the rules that produced the corresponding internal vertices.
These mappings allow constant-time ``jumps'' in the derivation tree, similar to (but more general
than) left-most derivations. Intuitively, for every nonterminal, every of its internal
vertices~$x$, and every edge label~$\sigma$, we precompute a ``tableau'' that contains the shortest
derivation of the $\sigma$-edge incident to $x$ (if it exists). In a traversal step, we merge
tableaux, which requires to manipulate up to $\maxRank$-many pointers, where $\maxRank$ is the
maximal rank of a nonterminal in $G$, defined as the maximal number of external vertices appearing
in any right-hand side of $G$'s rules. Thus, our traversals are constant
delay for grammars of fixed maximal rank, and our precomputed data structures occupy space in
$O(|G||\Sigma|\maxRank hr)$ where $\maxRank$ is the maximal rank of the nonterminals in $G$, $h$ is
the height of the derivation tree of the grammar $G$, and $r$ is the maximal rank of a symbol in
$\Sigma$. 

\medskip

A preliminary version of this paper appeared at DCC 2018~\cite{DBLP:conf/dcc/ManethP18}. The present version
is extended by allowing for traversal on general hyperedges, commenting on limits of the methods,
and providing proofs.

%% file: preliminaries.tex
\section{Preliminaries}
For $k \geq 0$ we denote by $[k]$ the set $\{1,\ldots,k\}$.
A \emph{ranked alphabet} consists of an alphabet (that is, a finite set of symbols) $\Sigma$
together with a mapping $\rank: \Sigma \rightarrow \N$ that assigns a rank to every symbol in
$\Sigma$. For the rest of the paper we assume that $\Sigma$ is a fixed set $\{\sigma, \sigma_1,
\ldots \sigma_m\}$ for some $m \in \N$.
The empty string is denoted by $\varepsilon$, and for two strings $u,v$ the
concatenation of $u$ and $v$ is denoted by $u\cdot v$, or $uv$ if that is unambiguous. 
The length of a string $u$ is denoted by $|u|$.
We denote by $u[i]$ the symbol at position $i$ of $u$ (starting with~$1$).
Let $k = \max\{\rank(\sigma) \mid \sigma \in \Sigma\}$ be the maximal rank of any symbol within
$\Sigma$. We define ranked, ordered trees as a tuple $(V, \tlbl, \rt, E)$ where $V$ is a set of
nodes, $\tlbl:V \to \Sigma$ a node labeling, $\rt \in V$ is the root node, and $E \subseteq V \times [k]
\times V$ is a set of edges such that $(u,i,v) \in E$ means that $v$ is the $i$-th child of $u$ and
we call $u$ the parent of $v$. For a node $u$ we denote by $E(u) = \{(u,i,v) \in E\}$ the set of
edges originating in $u$.
We require that 
\begin{enumerate} 
  \item every node $v$ has exactly $\rank(\tlbl(v))$ children (i.e., $|E(v)| = \rank(\tlbl(v))$),
  \item for any node $u$ there is at most one pair $i,v$ such that $(u,i,v) \in E$, and
  \item every node has exactly one parent, except the node $\rt$ which has no parent.
\end{enumerate}
For a node $u$ in the tree $(V, \tlbl, \rt, E)$, the \emph{subtree rooted in $u$} is the tree $(V',
\tlbl, u, E')$ where $V' \subseteq V$ contains the nodes reachable from $u$ and $E'\subseteq E$ the
edges of $E$ involving nodes in $V'$.

The nodes of a tree $s$ can be addressed by their \emph{Dewey addresses}
$D(s)$. For a tree $s = (V, \tlbl, \rt, E)$, $D(s)$ is recursively defined as $D(s) =
\{\varepsilon\} \cup \bigcup_{(\rt, i, u) \in E(\rt)} i.D(s_u)$ where $s_u$ is the subtree of $s$
rooted at $u$. Thus $\varepsilon$ addresses the root node and $u.i$ the $i$-th child of the node
$u$. Intuitively the Dewey address of a node $u$ is the sequence of edge indices encountered on the path
from the root of the tree to $u$. We generally use a node $u \in V$ and its Dewey address
interchangeably. The size of $s$ is $|s| = |E|$.

A \emph{hypergraph over $\Sigma$} (or
simply, a \emph{graph}) is a tuple $g = (V, E, \att, \lbl, \ext)$ where $V$ is the
finite, non-empty set of vertices, $E$ is a finite set of edges, $\att: E \rightarrow V^+$ is the
attachment mapping, $\lbl: E \rightarrow \Sigma$ is the label mapping, and $\ext \in V^*$ is a
string of \emph{external vertices}. A vertex $x \in V$ is \emph{external} if $x \in \ext$ and is
\emph{internal} otherwise. The vertex $\ext[i]$ for $i \in [|\ext|]$ is also referred to as the
$i$-th external vertex of $g$.
We refer to the
components of $g$ by  $V_g, E_g, \att_g, \lbl_g$, and $\ext_g$. We define the rank of an edge as
$\rank(e) = |\att(e)|$ and require that $\rank(e) = \rank(\lbl(e)) \geq 1$ for every $e$ in $E$.
Thus an edge labeled with $\sigma$ has $\rank(\sigma)$ many attached nodes and we do not allow edges
without attached vertices. We add the following restriction: for every $e \in E$, $\att(e)$ contains
no vertex twice. 
The rank of a hypergraph $g$ is defined as $\rank(g) = |\ext_g|$. We call a graph \emph{unique-labeled},
if for any vertex $x$, $\sigma \in \Sigma$, and $i$ there is at most one edge $e$ with
$\att(e)[i] = x$ and $\lbl(e) = \sigma$. We denote the set of all hypergraphs over $\Sigma$ by
$\HGR(\Sigma)$.
Figure~\ref{fig:hypergraph_example} shows an example of a hypergraph with 2 external vertices. Formally, the
pictured graph has $V = \{v_1,v_2,v_3\}$, $E = \{e_1,e_2,e_3\}$, $\att = \{e_1 \mapsto v_2\cdot v_1, e_2
\mapsto v_1\cdot v_3, e_3 \mapsto v_2\cdot v_1\cdot v_3\}$, $\lbl = \{e_1 \mapsto a, e_2 \mapsto b, e_3
\mapsto A\}$, and $\ext = v_2 \cdot v_3$. Note that external vertices are filled black and their position
within $\ext$ is given by a number under or above the vertex. Hyperedges of rank greater than $2$ are
represented by labeled boxes. The boxes have edges to their attached vertices; these edges are labeled
by indices indicating the order of the attached vertices. Rank-2 edges are drawn directed, from their
first to second attached vertex.

\begin{figure}[!t]
    \centering
    \input{figures/hypergraph_example}
    \caption{Example of a drawing of a hypergraph with three vertices (two of which external) and three
    edges.}
    \label{fig:hypergraph_example}
\end{figure}
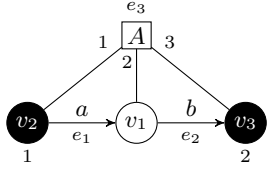

\begin{remark}
  Both, trees and graphs, are structures with ``nodes''. Since we use both types of structures in
  combination, we use the following convention to avoid confusion: we always refer to nodes of a
  tree as ``node(s)'', and exclusively use ``vertex/vertices'' for graphs.
\end{remark}
\paragraph*{Computation model}
We assume the use of a word RAM model with registers containing $\bitLen$-bit words for some number
$\beta \in \N$. The space of a data structure is measured by the number of registers it uses.
Registers are addressed by numbers, which we refer to as pointers, and can be randomly accessed in
constant time. Arithmetic operations (we only need addition) and comparisons of the contents of two
registers can be carried out in constant time. 

%% file: figures/hypergraph_example.tex
\begin{tikzpicture}
    \begin{scope}[every node/.style={circle, draw, inner sep=2pt, minimum size=9pt, node
        distance=25pt}]
        \node[fill=black] (1) {\idsize \color{white}$v_2$};
        \node[right=of 1] (2) {\idsize $v_1$};
        \node[right=of 2,fill=black] (3) {\idsize \color{white} $v_3$};
    \end{scope}
    \node[above=20pt of 2, rectangle, draw, inner sep=2pt, outer sep=0pt] (A) {\lblsize $A$};
    \node[above=-.2pt of A] {\scriptsize $e_3$};

    \draw[-stealth'] (1) -- node[above=-1pt] {\lblsize $a$} node[below] {\scriptsize $e_1$} (2);
    \draw[-stealth'] (2) -- node[above=-1pt] {\lblsize $b$} node[below] {\scriptsize $e_2$} (3);

    \draw (A) -- node[above left=-2pt, pos=.1] {\scriptsize $1$} (1);
    \draw (A) -- node[left=-2pt, pos=.25] {\scriptsize $2$} (2);
    \draw (A) -- node[above right=-2pt, pos=.1] {\scriptsize $3$} (3);

    \node[below=-1pt of 1] {\scriptsize 1};
    \node[below=-1pt of 3] {\scriptsize 2};
\end{tikzpicture}

%% file: SLHRGrammars.tex
\section{Representing graphs using SL HR grammars}
This section defines the grammar model used to represent graphs. We use hyperedge replacement
grammars, see e.g.~\cite{DBLP:books/sp/Habel92,DBLP:conf/gg/DrewesKH97,eng97}. The
straight-line restriction, as known from string (see e.g.~\cite{loh12,DBLP:reference/algo/Bannai16}) and tree
grammars (see~\cite{DBLP:conf/dlt/Lohrey15} for a survey) guarantees that a grammar produces
at most one object. It was studied for hyperedge replacement grammars
in~\cite{DBLP:journals/is/ManethP18}. Note that contrary to strings and trees, the restriction no
longer guarantees that the language of a grammar contains only one element. It merely guarantees
that every graph in the language of a graph grammar is isomorphic. We discuss the formal details
below and also explain how we get from a language of isomorphic graphs to one unique graph
(Section~\ref{sss:unique}).
\begin{definition}
  A \emph{hyperedge replacement grammar over $\Sigma$} (for short, HR grammar) is a tuple $G = (N, P, S)$,
    where $N$ is a ranked alphabet of \emph{nonterminals} with $N \cap \Sigma=\emptyset$, $P
    \subseteq N \times \HGR(\Sigma \cup N)$ is the set of \emph{rules} (or \emph{productions}) such
    that $\rank(A) = \rank(g)$ for every $(A,g) \in P$, and $S \in
    \HGR(\Sigma\cup N)$ is the \emph{start graph} with $\rank(S) = 0$.
\end{definition}
We often write $p:A \rightarrow g$ for a rule $p = (A,g)$ and call $A$ the left-hand side $\lhs(p)$
and $g$ the right-hand side $\rhs(p)$ of $p$. We call symbols in $\Sigma$ \emph{terminals}.
Consequently an edge is called \emph{terminal} if it is labeled by a terminal and \emph{nonterminal}
otherwise.
To define a derivation we first need to define how to replace a nonterminal edge $e$ by a graph $h$.
This replacement intuitively removes $e$ and replaces it with an isomorphic copy of $h$ such that
the external vertices of $h$ are merged, in order, with the vertices attached to $e$. The formal
definition is slightly more involved, see also Figure~\ref{fig:edgeReplacementEx} for a specific
example of an edge replacement (replacing the red edge $e$ in the graph $g$ with the graph $h$ given
below), annotated with all the formal details. We need to introduce some notation. Let $g$ be a
hypergraph and $\mathcal{V}$ the set of all possible vertices. Let $\nr: V_g \rightarrow \mathcal{V}$
be an injective function, and $\nr^*: V_g^* \rightarrow \mathcal{V}^*$ its extension to strings. We
call $\nr$ a \emph{vertex renaming on $g$}. We define $g[\nr] = (\{\nr(v) \mid v \in V_g\}, E_g,
\att^\nr, \lbl_g, \nr(\ext_g))$, where $\att^\nr(e) = \nr^*(\att_g(e))$ for all $e \in E_g$. For a
hypergraph $g$ and an edge $e \in E_g$ we denote by $g[-e]$ the hypergraph obtained from $g$ by
removing the edge $e$. For two hypergraphs $g,h$ we denote by $g \cup h$ the \emph{union} of the two
hypergraphs defined by $(V_g \cup V_h, E_g \uplus E_h, \att_g \uplus \att_h, \lbl_g \uplus \lbl_h,
\ext_g)$. Note that this union
\begin{enumerate}
    \item
merges vertices that exist in both $g$ and $h$ and
    \item 
creates disjoint copies of $E_h$, $\att_h$, and $\lbl_h$, and
    \item
uses the external vertices of $g$ only.
\end{enumerate}
Now, let $g,h$ be hypergraphs. Let $e \in E_g$ such that $\rank(e) = \rank(h)$.
Then let $\nr$ be a renaming of $h$ that identifies the external vertices of $h$ with the vertices
$e$ is attached to (i.e., $\nr(\ext_h[i]) = \att_g(e)[i]$ for all $i \in [\rank(h)]$) and uses fresh
identifiers for all internal nodes of $h$ (i.e., $\nr(\{v \mid v \in V_h, v $ is internal$\}) \cap
V_g = \emptyset$). The \emph{replacement of $e$ by $h$ in $g$} is defined as
\[g[e/h] = g[-e] \cup h[\nr].\]
\begin{figure}[!t]
  \centering
  \input{figures/edgeReplacementEx}
  \caption{Example of an edge replacement.}
  \label{fig:edgeReplacementEx}
\end{figure}
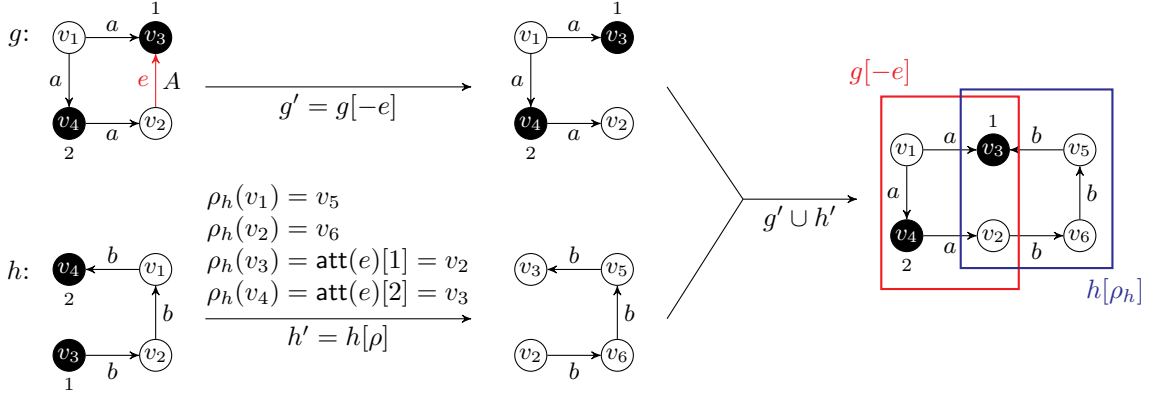
Let $g$ be a graph with a nonterminal edge $e$. A \emph{derivation step} is defined as $g
\Rightarrow h$ if there exist a nonterminal edge $e$ in $g$ and a rule $\lbl(e) \to g'$ such that $h
= g[e/g']$. A derivation $g \Rightarrow^* h$ consists of any number of derivation steps and we
define the \emph{language of $G$} as $L(G) = \{g \in \HGR(\Sigma) \mid S \Rightarrow^* g\}$.

\begin{remark}
  This grammar formalism may be unfamiliar to some readers. We therefore give an example of a string
  grammar and an equivalent HR grammar to show that HR grammars are a generalization of string
  grammars. Consider the grammar
  \begin{align*}
    S &\to AB \\
    A &\to ab \mid aAb \\
    B &\to c \mid cB
  \end{align*}
  It generates the language $\{a^nb^nc^m \mid n,m \geq 1\}$. Figure~\ref{fig:string_grammar_ex}
  shows an equivalent HR grammar, which generates graphs such that their edge labels, read in order,
  generate the same language. To derive the word $aaabbbcc$ from the above string grammar the
  following derivation could be used:
  \begin{align*}
    S &\Rightarrow AB \\
      &\Rightarrow aAbB \\
      &\Rightarrow aaAbbB \\
      &\Rightarrow aaabbbB \\
      &\Rightarrow aaabbbcB \\
      &\Rightarrow aaabbbcc
  \end{align*}
  An equivalent derivation of the HR grammar is shown in 
  Figure~\ref{fig:string_derivation_ex}. In this case, the external vertices of the graphs in the
  right-hand sides mark the beginning and end of the string. This is not necessary in string
  grammars, because strings implicitly have a given beginning and end.
  In the figures we identify different internal vertices by using
  different colors. To show their origin in the grammar, only the saturation changes during
  derivation, not the hue.
  \label{rem:comparison}
\end{remark}

\begin{definition}
    An HR grammar $G = (N,P,S)$ is called \emph{straight-line} (SL HR grammar) if 
    \begin{enumerate}
      \item the binary relation 
        \begin{align*}
          \SLOrder =& \{(A_1, A_2) \in N^2 \mid \exists g: (A_1,g)\in P, \exists e \in E_g:\lbl_g(e)
          = A_2\} \\
          &\cup \{(S, A) \in \{S\}\times N \mid \exists e \in E_S: \lbl_S(e) = A\}
        \end{align*}
      defines a single (ordered) directed acyclic graph with vertices $N \cup \{S\}$ rooted at $S$, and
      \item for every $A \in N$ there exists exactly one rule $p \in P$ with
        $\lhs(p) = A$.
    \end{enumerate}
    \label{def:SLHR}
\end{definition}

Note that $L(G) \neq \emptyset$ and $L(G)$ is a (potentially infinite) set of isomorphic graphs.
As the right-hand side for a nonterminal is unique in
SL HR grammars we denote the right-hand side $g$ of $p = (A,g)$ by $\rhs(A)$. By convention
whenever we state something over all right-hand sides of a grammar this includes the start graph.
The height of an SL HR grammar $\height(G)$ is the longest distance of any two nonterminals within
the DAG induced by $\SLOrder$.
We refer by $\maxRank = \max_{A\in N}\{\rank(A)\}$ to the maximal nonterminal rank in the SL HR
grammar, which we further assume to be fixed for the rest of this paper. In the following we often
say \emph{grammar} instead of \emph{SL HR grammar}. An SL HR grammar is \emph{unique-labeled} if the
graphs in $L(G)$ are unique-labeled. For the remainder of this paper we will assume every grammar
$G$ to be unique-labeled. 
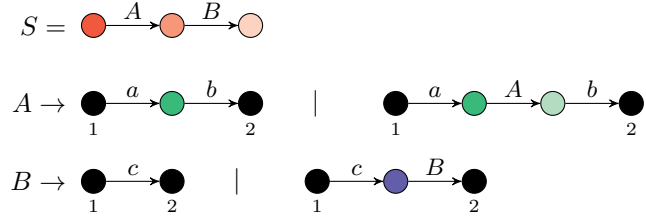
\begin{figure}[!t]
  \centering
  \input{figures/stringHRGrammar_ex}
  \caption{HR grammar generating graphs representing the string language $\{a^nb^nc^m \mid n,m \geq
  1\}$.}
  \label{fig:string_grammar_ex}
\end{figure}
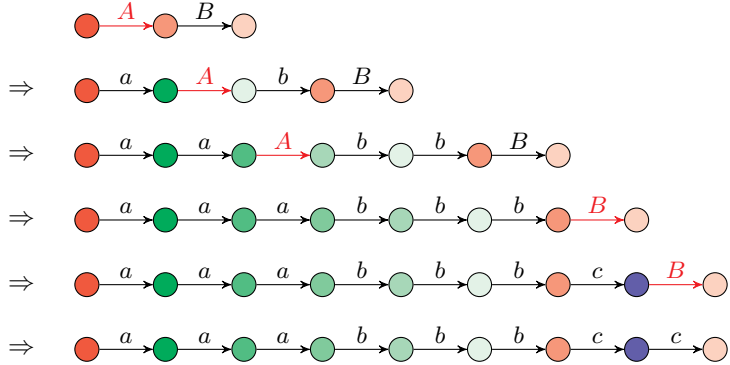
\begin{figure}[!t]
  \centering
  \input{figures/stringHRDerivation_ex}
  \caption{One possible derivation of the HR grammar given in Figure~\ref{fig:string_grammar_ex}.
    The resulting graph represents the string $aaabbbcc$.}
  \label{fig:string_derivation_ex}
\end{figure}
\subsection{Uniquely identifying vertices}\label{sss:unique}
As mentioned in the previous section, $L(G)$ for an SL HR grammar $G$ is a potentially infinite set
containing only isomorphic copies of one graph. Since we intend to traverse the vertices of the graph
represented by $G$ it is necessary to specify some way to uniquely identify one such vertex. This
consequently requires fixing one graph within $L(G)$ to be the value of $G$. In this section we
introduce methods to achieve this. Specifically we
\begin{enumerate}
  \item define vertices to be numbered with identifiers instead of the arbitrary symbolic sets used
    up to here,
  \item specify exactly which vertex replacement is used during a derivation step, and
  \item fix a specific derivation of $G$ by defining an order on the nonterminal edges.
\end{enumerate}
Of these steps the latter two are necessary to fix one graph within $L(G)$. The first one 
simplifies the definition of the vertex replacement and provides a sensible way of representing
graphs in an implementation. 

We begin with numbering the vertices. We require $V = [n]$ for some $n > 0$ for
every graph. Furthermore a graph with $k$ external vertices is required to have $\ext = (n-k+1)
\cdots n$, i.e., the last $k$ vertices of a graph are its external vertices. Since the vertex
numbers thus also imply their position within $\ext$ we now no longer add subscripts in the
drawings of graphs indicating this order. These restrictions allow us to specify a vertex renaming
during an edge replacement by simple arithmetic. Let $g, h$ be graphs with $n$ and $m$ vertices,
respectively. Let $e$ an edge in $g$ with $\rank(e) = \rank(h)$. Then we define $g[e/h]$ as above,
but specify $\nr$ as
\[
  \nr(x) = \begin{cases}
    \att_g(e)[i] & \text{ if $x$ is the $i$-th external node of $h$} \\
    n + x & \text{ otherwise}
  \end{cases}.
\]
That is, the identification of external vertices with the vertices attached to $e$ is the same as
before. The internal vertices of $h$ on the other hand are defined to count consecutively starting from
$n$, which is the largest vertex number in $g$. Note that, if both $g$ and $h$ have external
vertices the graph $g[e/h]$ is not guaranteed to conform to the above restrictions on graphs. In
particular it does not necessarily have its external vertices at the end. However, if $g$ is of rank
$0$ then $g[e/h]$ is also of rank $0$ and has the vertices $\{1,\ldots,n+(m-k)\}$.

The changes above already have the effect that $L(G)$ for an SL HR grammar $G$ is now always finite,
since the vertex renaming during a derivation step is fixed and there is only a finite number of
possible derivations. To now fix one particular graph $\val(G)$ from $L(G)$ as the graph represented
by $G$ we fix the order in which nonterminal edges are derived. We do this by defining a derivation
tree, which has a node for every nonterminal edge appearing in a full derivation of $G$. The
derivation order is then fixed by the depth-first left to right traversal of this tree. We define an
order on the nonterminal edges of a graph. Let $\leq_N$ be some (arbitrary but fixed) order on the
nonterminals of $G$. For a graph $g$ let $\NTSet_g$ be the set of nonterminal edges in $g$.
We define the derivation order $\leq_{\dt}$ on $\NTSet_g$ in the following way: let $e_1,e_2$ be two edges
from $\NTSet_g$. Then $e_1 \leq_{\dt} e_2$ if
\begin{itemize}
    \item $\att_g(e_i) <_{\text{lex}} \att_g(e_j)$ (here $<_{\text{lex}}$ is the lexicographical
        order), or
    \item $\att_g(e_i) = \att_g(e_j)$ and $\lbl_g(e_i) \leq_N \lbl_g(e_j)$.
\end{itemize}

Note that $e_1 =_{\dt} e_2$ is possible if $\att(e_1) = \att(e_2)$ and $\lbl(e_1) = \lbl(e_2)$. In
these cases it is irrelevant for the final result in which order the edges are derived, so they can
be ordered arbitrarily.

We now define $\val(A)$ for every nonterminal $A$ and define $\val(G)$ as $\val(S)$.  Simultaneously
we define the \emph{derivation-tree} $\dt(A)$ of $A$ and the derivation-tree of $G$ as $\dt(S)$.
Note that $\dt(A)$ is a ranked ordered tree with nodes labeled by rules of $G$; the number of children of a
node is equal to the number of nonterminal edges in the right-hand side of the rule in the label.
Let $g=\rhs(A)$.  If $\NTSet_g=\emptyset$ then $\val(A)=g$ and $\dt(A)=$ a single node labeled $A
\to g$. Otherwise let $\NTSet_g = \{e_1,\ldots,e_n\}$ such that $e_1
\leq_{\dt} e_2 \leq_{\dt} \cdots \leq_{\dt} e_n$ and for every $i \in [n]$ let $e_i$ be labeled by
$A_i$ and let $s_i = \dt(A_i)$ be trees with disjoint sets of nodes $V_i$. Then we define
\begin{itemize}
  \item 
$\val(A)=g[e_1/\val(A_1)]\dots [e_n/\val(A_n)]$, and
  \item
 $\dt(A)$ as a tree with nodes $\{v\}
\cup \bigcup_{i \in \rank(A)} V_i$, where $v$ is a new node labeled $A\to g$ and the $i$-th child of
$v$ is the root-node of $s_i$.
\end{itemize}

Thus $\val(A)$ is obtained from $\rhs(A)$ by replacing
all its nonterminal edges in the order specified by $\leq_{\dt}$.

Figures~\ref{fig:annotatedDT} and~\ref{fig:full_graph} show a full example of these notions. This
derivation tree and graph will serve as a running example for the methods introduced below. On the
left of Figure~\ref{fig:annotatedDT} is an SL HR grammar. Its derivation tree $\dt(G)$ as defined
above is given to the right of the grammar, each node of the tree contains its Dewey-address. For an
example of the order $\leq_{\dt}$ consider the graph $\rhs(A)$: it has two nonterminal edges $e_1$ and
$e_2$, with $\att(e_1) = 1\cdot 4$, $\att(e_2) = 3\cdot 2\cdot 1$, $\lbl(e_1) = B$, and $\lbl(e_2) =
C$. Since $\att(e_1) <_{\text{lex}} \att(e_2)$ the derivation order is $e_1 \leq_{\dt} e_2$ and
consequently the first child of node $u_2$ in the derivation tree is $B \to \rhs(B)$ followed by $C
\to \rhs(C)$. Figure~\ref{fig:full_graph} shows the graph $\val(G)$. All the nodes in the graph are
colored in the same colors as the internal nodes in the derivation tree given in
Figure~\ref{fig:annotatedDT} to show which specific derivation they originate from. We invite the
reader to reconstruct the full step by step derivation of $\val(G)$ from the given grammar and
derivation tree.

Let $A \in N$ and let 
$u$ be a node in $\dt(A)$. By convention we will use $A_u \to g_u$ to refer to the label of $u$.
For an example consider again the derivation tree in
Figure~\ref{fig:annotatedDT}. For the node $u_4 = 1.2$ the graph $g_{u_4}$ is $\rhs(C)$ and the
nonterminal $A_{u_4} = C$.
\begin{figure}[t!]
  \centering
  \scalebox{1}{\input{figures/derTree}}
  \caption{An SL HR grammar (left), and it's derivation tree (right). Light gray lines indicate for
    every external vertex, the vertex in the parent which they are merged with. 
}
  \label{fig:annotatedDT}
\end{figure}
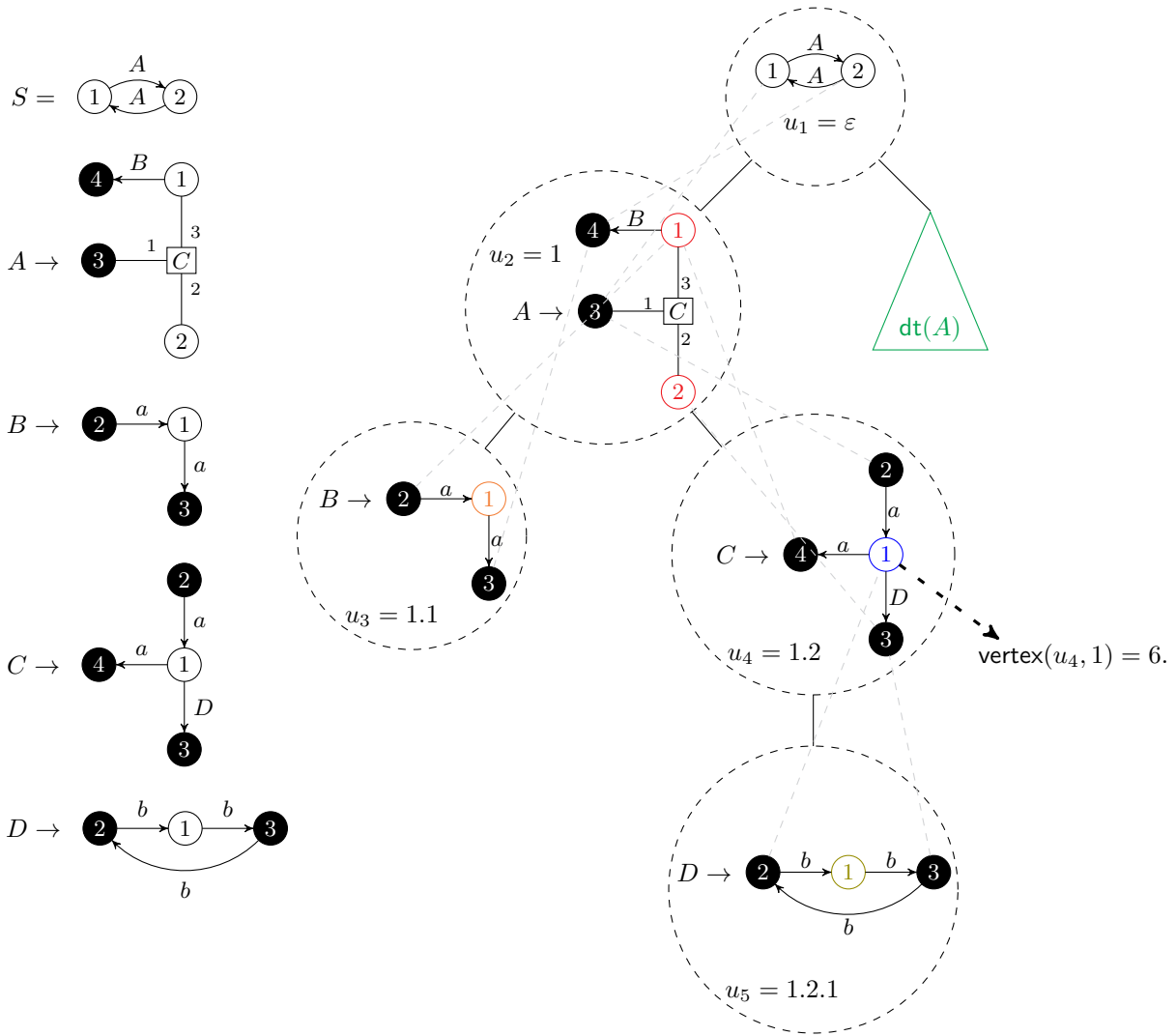
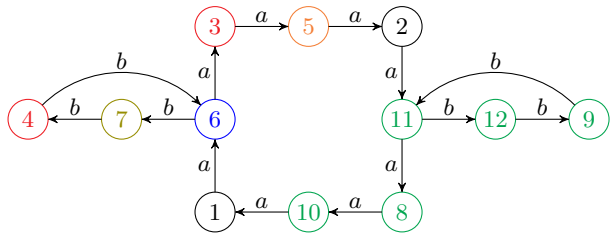
\begin{figure}[!t]
  \centering
  \input{figures/full_graph}
  \caption{The graph represented by the grammar given on the left of Figure~\ref{fig:annotatedDT}.
    Vertices are colored to show their origin within the derivation tree (right of
  Figure~\ref{fig:annotatedDT}).}
  \label{fig:full_graph}
\end{figure}

With these conventions regarding vertex numbers and derivation order, we make the following
observation: for a graph $g[e/h]$ we can identify whether a vertex originates from $g$ or $h$ from
its number. Vertices originating from $g$ are numbered $1$ to $|V_g|$, whereas vertices from $h$ are
numbered starting with $|V_g|+1$. Thus, given only $g$ and $h$ we know the vertex numbers in
$g[e/h]$ without the need to explicitly compute the edge replacement. Moreover, this still applies
when there are other nonterminal edges in $g$ that are derived before $e$. Assume there are
nonterminal edges $e_1,\ldots,e_n$ which are derived before $e$. For an edge $e_i$ ($1 \leq i \leq
n$) let $V_i$ be the set of internal vertices in $\val(A_i)$ where $A_i$ is the label of $e_i$. Then
the first available vertex number in the replacement of $e$ is $|V_g| + |V_1| + \cdots + |V_n|$.
Doing this recursively along the path of the derivation tree allows us, given a node $u$ in
$\dt(G)$ and a vertex $x$ in $g_u$, to compute the number of the vertex represented by $x$ in
$\val(G)$.
Formally we let $u$ be a node in the derivation tree $\dt(A)$ for some nonterminal $A$ and define a
mapping between tuples $(u, x)$ and the internal vertices of $\val(A)$. For a nonterminal $B$ we let
$\nodes(B)$ be the number of internal vertices in $\val(B)$. We recursively define $\firstID(u)$ as
$\firstID(u) = 0$ if $u = \varepsilon$ and if $u = v.i$ for some $i \in \N$ and $\NTSet_{g_u} =
\{e_1,\ldots,e_n\}$ with $e_1 \leq_{\dt} \cdots e_n\sib(E_{g_{u}}^{\nt}) = (e_1,\ldots,e_n)$ then we
define the offset
\[\firstID(u) = \firstID(v) + \sum_{j < i} \nodes(\lbl(e_j)) + |V_{g_{s,v}}| -
|\ext_{g_{v}}|.\]
Now we define \[\id(u,x) = \firstID(u) + x\] for any internal vertex $x$ in $g_{u}$, which
represents the vertex within $\val(A)$.
Figure~\ref{fig:annotatedDT} shows an example of a derivation tree with the root-node labeled by $S$
and its left child labeled by $A \to \rhs(A)$. The Figure also includes a specific example for
the $\id$-mapping (vertex 1 in context $u_4$). Its full computation is as follows:
\begin{align*}
  \id(u_4, 1) &=\firstID(u_4) + 1 \\
              &=\firstID(u_2) + |V_{\rhs(A)}| - |\ext_{\rhs(A)}| + \nodes(B) + 1\\
              &= \firstID(u_1) + |V_S| - |\ext_S| + |V_{\rhs(A)}| - |\ext_{\rhs(A)}| + \nodes(B)
              + 1 \\
              &= 0 + 2 - 0 + 4 - 2 + 1 +1 = 6. 
\end{align*}
Thus the vertex $1$ in $\rhs(C)$ represents the vertex $6$ in $\val(G)$. Do note that the same node
$u$ given as a Dewey-address may refer to different nodes if evaluated on different trees, e.g. for
two nonterminals $A$ and $B$ the address $u$ will refer to different nodes in $\dt(A)$ and $\dt(B)$,
assuming it is valid in both. The value $\firstID(u,x)$ can therefore change when referring to
different trees. It should always be clear from context on which tree this mapping is evaluated. In
our final algorithm it is always $\dt(G)$.

%% file: figures/edgeReplacementEx.tex
\begin{tikzpicture}
  \begin{scope}[every node/.style={circle, draw, inner sep=0pt, minimum size=12pt, node
    distance=20pt}]
    \node (g1) {\idsize $v_1$};
    \node[below=of g1, text=white, fill=black] (g4) {\idsize $v_4$};
    \node[right=of g1, text=white, fill=black] (g3) {\idsize $v_3$};
    \node[right=of g4] (g2) {\idsize $v_2$};

    \node[below=1.5 of g4, text=white, fill=black] (h4) {\idsize $v_4$};
    \node[below=of h4, text=white, fill=black] (h3) {\idsize $v_3$};
    \node[right=of h4] (h1) {\idsize $v_1$};
    \node[right=of h3] (h2) {\idsize $v_2$};
  \end{scope}
  \node[left=5pt of g1] {$g$:};
  \node[left=5pt of h4] {$h$:};
  \node[below=-1pt of g4] {\scriptsize 2};
  \node[above=-1pt of g3] {\scriptsize 1};
  \node[below=-1pt of h4] {\scriptsize 2};
  \node[below=-1pt of h3] {\scriptsize 1};

  \begin{scope}[every path/.style={-stealth'}]
    \draw (g1) -- node[left=-1pt] {\lblsize $a$} (g4);
    \draw (g1) -- node[above=-1pt] {\lblsize $a$} (g3);
    \draw (g4) -- node[below=-1pt] {\lblsize $a$} (g2);
    \draw[color=Red] (g2) -- 
      node[left=-1pt] {\lblsize $e$}
      node[right=-1pt,color=black] {\lblsize $A$} (g3);

    \draw (h3) -- node[below=-1pt] {\lblsize $b$} (h2);
    \draw (h2) -- node[right=-1pt] {\lblsize $b$} (h1);
    \draw (h1) -- node[above=-1pt] {\lblsize $b$} (h4);
  \end{scope}

  \coordinate[below right=.5 and .5 of g3] (gx1);
  \coordinate[right=3.5 of gx1] (gx2);
  \coordinate[below right=.5 and .5 of h1] (hx1);
  \coordinate[right=3.5 of hx1] (hx2);

  \begin{scope}[every path/.style={-stealth'}]
    \draw (gx1) --node[below=-1pt] {$g' = g[-e]$}
    (gx2);
    \draw (hx1) --node[below=-1pt] {$h' = h[\nr]$}
      node[above] {$
        \begin{array}[h]{l}
          \nr_h(v_1) = v_5\\
          \nr_h(v_2) = v_6\\
          \nr_h(v_3) = \att(e)[1] = v_2\\
          \nr_h(v_4) = \att(e)[2] = v_3
        \end{array}$
      }
    (hx2);
  \end{scope}

  \begin{scope}[every node/.style={circle, draw, inner sep=0pt, minimum size=12pt, node
    distance=20pt}]
    \node[right=4.5 of g3] (2g1) {\idsize $v_1$};
    \node[below=of 2g1, text=white, fill=black] (2g4) {\idsize $v_4$};
    \node[right=of 2g1, text=white, fill=black] (2g3) {\idsize $v_3$};
    \node[right=of 2g4] (2g2) {\idsize $v_2$};

    \node[below=1.5 of 2g4] (2h4) {\idsize $v_3$};
    \node[below=of 2h4] (2h3) {\idsize $v_2$};
    \node[right=of 2h4] (2h1) {\idsize $v_5$};
    \node[right=of 2h3] (2h2) {\idsize $v_6$};
  \end{scope}
  \node[below=-1pt of 2g4] {\scriptsize 2};
  \node[above=-1pt of 2g3] {\scriptsize 1};
  
  \begin{scope}[every path/.style={-stealth'}]
    \draw (2g1) -- node[left=-1pt] {\lblsize $a$} (2g4);
    \draw (2g1) -- node[above=-1pt] {\lblsize $a$} (2g3);
    \draw (2g4) -- node[below=-1pt] {\lblsize $a$} (2g2);

    \draw (2h3) -- node[below=-1pt] {\lblsize $b$} (2h2);
    \draw (2h2) -- node[right=-1pt] {\lblsize $b$} (2h1);
    \draw (2h1) -- node[above=-1pt] {\lblsize $b$} (2h4);
  \end{scope}

  \coordinate[below right=.5 and .5 of 2g3] (2gx1);
  \coordinate[below right=.5 and .5 of 2h1] (2hx1);
  \coordinate[below right=.85 and 1.5 of 2g2] (ghx1);
  \coordinate[right=1.5 of ghx1] (ghx2);
  
  \draw (2gx1) -- (ghx1);
  \draw (2hx1) -- (ghx1);
  \draw[-stealth'] (ghx1) --node[below=-1pt] {$g' \cup h'$} (ghx2);

  \begin{scope}[every node/.style={circle, draw, inner sep=0pt, minimum size=12pt, node
    distance=20pt}]
    \node[above right=.5 and .5 of ghx2] (3g1) {\idsize $v_1$};
    \node[below=of 3g1, text=white, fill=black] (3g4) {\idsize $v_4$};
    \node[right=of 3g1, text=white, fill=black] (3g3) {\idsize $v_3$};
    \node[right=of 3g4] (3g2) {\idsize $v_2$};

    \node[right=of 3g3] (3h1) {\idsize $v_5$};
    \node[right=of 3g2] (3h2) {\idsize $v_6$};
  \end{scope}

  \begin{scope}[every path/.style={-stealth'}]
    \draw (3g1) -- node[left=-1pt] {\lblsize $a$} (3g4);
    \draw (3g1) -- node[above=-1pt] {\lblsize $a$} (3g3);
    \draw (3g4) -- node[below=-1pt] {\lblsize $a$} (3g2);

    \draw (3g2) -- node[below=-1pt] {\lblsize $b$} (3h2);
    \draw (3h2) -- node[right=-1pt] {\lblsize $b$} (3h1);
    \draw (3h1) -- node[above=-1pt] {\lblsize $b$} (3g3);
  \end{scope}
  \node[below=-1pt of 3g4] (s1) {\scriptsize 2};
  \node[above=-1pt of 3g3] (s2) {\scriptsize 1};

  \node[draw, color=Red, thick, rectangle, fit=(3g1) (3g2) (s1) (s2)] (gBorder) {};
  \node[draw, color=Blue, thick, rectangle, inner sep=6pt, fit=(3h1) (3h2) (3g2) (s2)] (hBorder) {};
  \node[above right=0pt and -15pt of gBorder.north west, Red] {$g[-e]$};
  \node[below left=0pt and -15pt of hBorder.south east, Blue] {$h[\nr_h]$};
\end{tikzpicture}

%% file: figures/stringHRGrammar_ex.tex
\begin{tikzpicture}
  \begin{scope}[every node/.style={circle, draw, inner sep=0pt, minimum size=9pt, node
    distance=20pt}]
    \node[fill=Red!80] (S1) {};
    \node[right=of S1, fill=Red!50] (S2) {};
    \node[right=of S2, fill=Red!20] (S3) {};

    \node[below=20pt of S1, fill=black, text=white] (A2) {};
    \node[right=of A2, fill=Green!70] (A1) {};
    \node[right=of A1, fill=black, text=white] (A3) {};

    \node[right=.5 of A3, draw=none] (Amid) {$\mid$};

    \node[right=of Amid, fill=black, text=white] (A4) {};
    \node[right=of A4, fill=Green!70] (A5) {};
    \node[right=of A5, fill=Green!30] (A6) {};
    \node[right=of A6, fill=black, text=white] (A7) {};

    \node[below=20pt of A2, fill=black, text=white] (B2) {};
    \node[right=of B2, fill=black, text=white] (B1) {};

    \node[right=.5 of B1, draw=none] (Bmid) {$\mid$};

    \node[right=of Bmid, fill=black, text=white] (B4) {};
    \node[right=of B4, fill=Blue!70] (B5) {};
    \node[right=of B5, fill=black, text=white] (B6) {};
  \end{scope}
  \node[left=2pt of S1] (Srule) {$S=$};
  \node[left=2pt of A2] (Arule) {$A\to$};
  \node[left=2pt of B2] (Brule) {$B\to$};

  \node[below=-1pt of A2] {\scriptsize $1$};
  \node[below=-1pt of A3] {\scriptsize $2$};
  \node[below=-1pt of A4] {\scriptsize $1$};
  \node[below=-1pt of A7] {\scriptsize $2$};

  \node[below=-1pt of B2] {\scriptsize $1$};
  \node[below=-1pt of B1] {\scriptsize $2$};
  \node[below=-1pt of B4] {\scriptsize $1$};
  \node[below=-1pt of B6] {\scriptsize $2$};

  \begin{scope}[every path/.style={-stealth'}]
    \draw (S1) --node[above=-1pt] {\lblsize $A$} (S2);
    \draw (S2) --node[above=-1pt] {\lblsize $B$} (S3);

    \draw (A2) --node[above=-1pt] {\lblsize $a$} (A1);
    \draw (A1) --node[above=-1pt] {\lblsize $b$} (A3);

    \draw (A4) --node[above=-1pt] {\lblsize $a$} (A5);
    \draw (A5) --node[above=-1pt] {\lblsize $A$} (A6);
    \draw (A6) --node[above=-1pt] {\lblsize $b$} (A7);

    \draw (B2) --node[above=-1pt] {\lblsize $c$} (B1);

    \draw (B4) --node[above=-1pt] {\lblsize $c$} (B5);
    \draw (B5) --node[above=-1pt] {\lblsize $B$} (B6);
  \end{scope}
\end{tikzpicture}

%% file: figures/stringHRDerivation_ex.tex
\begin{tikzpicture}
  \begin{scope}[every node/.style={circle, draw, inner sep=0pt, minimum size=9pt, node
    distance=20pt}]
    \node[fill=Red!80] (11) {};
    \node[right=of 11, fill=Red!50] (12) {};
    \node[right=of 12, fill=Red!20] (13) {};

    \node[below=15pt of 11, fill=Red!80] (21) {};
    \node[right=of 21, fill=Green!90] (22) {};
    \node[right=of 22, fill=Green!10] (23) {};
    \node[right=of 23, fill=Red!50] (24) {};
    \node[right=of 24, fill=Red!20] (25) {};
    \node[left=.5 of 21, draw = none] {$\Rightarrow$};

    \node[below=15pt of 21, fill=Red!80] (31) {};
    \node[right=of 31, fill=Green!90] (32) {};
    \node[right=of 32, fill=Green!65] (33) {};
    \node[right=of 33, fill=Green!35] (34) {};
    \node[right=of 34, fill=Green!10] (35) {};
    \node[right=of 35, fill=Red!50] (36) {};
    \node[right=of 36, fill=Red!20] (37) {};
    \node[left=.5 of 31, draw = none] {$\Rightarrow$};

    \node[below=15pt of 31, fill=Red!80] (41) {};
    \node[right=of 41, fill=Green!90] (42) {};
    \node[right=of 42, fill=Green!65] (43) {};
    \node[right=of 43, fill=Green!50] (44) {};
    \node[right=of 44, fill=Green!35] (45) {};
    \node[right=of 45, fill=Green!10] (46) {};
    \node[right=of 46, fill=Red!50  ] (47) {};
    \node[right=of 47, fill=Red!20 ] (48) {};
    \node[left=.5 of 41, draw = none] {$\Rightarrow$};

    \node[below=15pt of 41, fill=Red!80] (51) {};
    \node[right=of 51, fill=Green!90] (52) {};
    \node[right=of 52, fill=Green!65] (53) {};
    \node[right=of 53, fill=Green!50] (54) {};
    \node[right=of 54, fill=Green!35] (55) {};
    \node[right=of 55, fill=Green!10] (56) {};
    \node[right=of 56, fill=Red!50  ] (57) {};
    \node[right=of 57, fill=Blue!70] (58) {};
    \node[right=of 58, fill=Red!20 ] (59) {};
    \node[left=.5 of 51, draw = none] {$\Rightarrow$};

    \node[below=15pt of 51, fill=Red!80] (61) {};
    \node[right=of 61, fill=Green!90] (62) {};
    \node[right=of 62, fill=Green!65] (63) {};
    \node[right=of 63, fill=Green!50] (64) {};
    \node[right=of 64, fill=Green!35] (65) {};
    \node[right=of 65, fill=Green!10] (66) {};
    \node[right=of 66, fill=Red!50  ] (67) {};
    \node[right=of 67, fill=Blue!70] (68) {};
    \node[right=of 68, fill=Red!20 ] (69) {};
    \node[left=.5 of 61, draw = none] {$\Rightarrow$};
  \end{scope}

  \begin{scope}[every path/.style={-stealth'}]
    \draw[Red] (11) --node[above=-1pt] {\lblsize $A$} (12);
    \draw (12) --node[above=-1pt] {\lblsize $B$} (13);

    \draw (21) --node[above=-1pt] {\lblsize $a$} (22);
    \draw[Red] (22) --node[above=-1pt] {\lblsize $A$} (23);
    \draw (23) --node[above=-1pt] {\lblsize $b$} (24);
    \draw (24) --node[above=-1pt] {\lblsize $B$} (25);

    \draw (31) --node[above=-1pt] {\lblsize $a$} (32);
    \draw (32) --node[above=-1pt] {\lblsize $a$} (33);
    \draw[Red] (33) --node[above=-1pt] {\lblsize $A$} (34);
    \draw (34) --node[above=-1pt] {\lblsize $b$} (35);
    \draw (35) --node[above=-1pt] {\lblsize $b$} (36);
    \draw (36) --node[above=-1pt] {\lblsize $B$} (37);

    \draw (41) --node[above=-1pt] {\lblsize $a$} (42);
    \draw (42) --node[above=-1pt] {\lblsize $a$} (43);
    \draw (43) --node[above=-1pt] {\lblsize $a$} (44);
    \draw (44) --node[above=-1pt] {\lblsize $b$} (45);
    \draw (45) --node[above=-1pt] {\lblsize $b$} (46);
    \draw (46) --node[above=-1pt] {\lblsize $b$} (47);
    \draw[Red] (47) --node[above=-1pt] {\lblsize $B$} (48);

    \draw (51) --node[above=-1pt] {\lblsize $a$} (52);
    \draw (52) --node[above=-1pt] {\lblsize $a$} (53);
    \draw (53) --node[above=-1pt] {\lblsize $a$} (54);
    \draw (54) --node[above=-1pt] {\lblsize $b$} (55);
    \draw (55) --node[above=-1pt] {\lblsize $b$} (56);
    \draw (56) --node[above=-1pt] {\lblsize $b$} (57);
    \draw (57) --node[above=-1pt] {\lblsize $c$} (58);
    \draw[Red] (58) --node[above=-1pt] {\lblsize $B$} (59);

    \draw (61) --node[above=-1pt] {\lblsize $a$} (62);
    \draw (62) --node[above=-1pt] {\lblsize $a$} (63);
    \draw (63) --node[above=-1pt] {\lblsize $a$} (64);
    \draw (64) --node[above=-1pt] {\lblsize $b$} (65);
    \draw (65) --node[above=-1pt] {\lblsize $b$} (66);
    \draw (66) --node[above=-1pt] {\lblsize $b$} (67);
    \draw (67) --node[above=-1pt] {\lblsize $c$} (68);
    \draw (68) --node[above=-1pt] {\lblsize $c$} (69);
  \end{scope}
\end{tikzpicture}

%% file: figures/derTree.tex
\begin{tikzpicture}[remember picture]
  % S
  \node[] (S) {$S =$};
  \begin{scope}[every node/.style={circle, draw, minimum size=9pt, inner sep=2pt, outer sep=0pt,
    node distance=20pt}]
    \node[right=5pt of S] (S1) {\idsize 1};
    \node[right=of S1] (S2) {\idsize 2};
  \end{scope}
  \begin{scope}[every path/.style={-stealth'}]
    \draw[bend left] (S2) to node[above] {\lblsize $A$} (S1);
    \draw[bend left] (S1) to node[above] {\lblsize $A$} (S2);
  \end{scope}

  % A 
  \node[below=50pt of S] (A) {$A\to$};
  \begin{scope}[every node/.style={circle, draw, minimum size=9pt, inner sep=2pt, outer sep=0pt,
    node distance=20pt}]
    \node[right=5pt of A, fill=black, text=white] (A3) {\idsize 3};
    \node[right=of A3, rectangle, inner sep=2pt] (AC) {\lblsize $C$};
    \node[below=of AC] (A2) {\idsize 2};
    \node[above=of AC] (A1) {\idsize 1};
    \node[left=of A1, fill=black, text=white] (A4) {\idsize 4};
  \end{scope}
  \begin{scope}[every path/.style={-stealth'}]
    \draw (A1) -- node[above] {\lblsize $B$} (A4);
  \end{scope}
  \draw (AC) -- node[above, pos=.3] {\scriptsize $1$} (A3);
  \draw (AC) -- node[right, pos=.3] {\scriptsize $2$} (A2);
  \draw (AC) -- node[right, pos=.3] {\scriptsize $3$} (A1);

  % B 
  \node[below=50pt of A] (B) {$B \to$};
  \begin{scope}[every node/.style={circle, draw, minimum size=9pt, inner sep=2pt, outer sep=0pt,
    node distance=20pt}]
    \node[right=5pt of B, fill=black, text=white] (B2) {\idsize 2};
    \node[right=of B2] (B1) {\idsize 1};
    \node[below=of B1, fill=black, text=white] (B3) {\idsize 3};
  \end{scope}
  \begin{scope}[every path/.style={-stealth'}]
    \draw (B2) -- node[above] {\lblsize $a$} (B1);
    \draw (B1) -- node[right] {\lblsize $a$} (B3);
  \end{scope}

  % C
  \node[below=80pt of B] (C) {$C \to$};
  \begin{scope}[every node/.style={circle, draw, minimum size=9pt, inner sep=2pt, outer sep=0pt,
    node distance=20pt}]
    \node[right=5pt of C, fill=black, text=white] (C4) {\idsize 4};
    \node[right=of C4] (C1) {\idsize 1};
    \node[above=of C1, fill=black, text=white] (C2) {\idsize 2};
    \node[below=of C1, fill=black, text=white] (C3) {\idsize 3};
  \end{scope}
  \begin{scope}[every path/.style={-stealth'}]
    \draw (C2) -- node[right] {\lblsize $a$} (C1);
    \draw (C1) -- node[above] {\lblsize $a$} (C4);
    \draw (C1) -- node[right] {\lblsize $D$} (C3);
  \end{scope}

  % D
  \node[below=50pt of C] (D) {$D \to$};
  \begin{scope}[every node/.style={circle, draw, minimum size=9pt, inner sep=2pt, outer sep=0pt,
    node distance=20pt}]
    \node[right=5pt of D, fill=black, text=white] (D2) {\idsize 2};
    \node[right=of D2] (D1) {\idsize 1};
    \node[right=of D1, fill=black, text=white] (D3) {\idsize 3};
  \end{scope}
  \begin{scope}[every path/.style={-stealth'}]
    \draw (D2) -- node[above=1pt] {\lblsize $b$} (D1);
    \draw (D1) -- node[above=1pt] {\lblsize $b$} (D3);
    \draw[bend left=45] (D3) to node[below=1pt] {\lblsize $b$} (D2);
  \end{scope}

  \node[right=9 of S, dashed, circle, draw, inner sep=0pt, outer sep=0pt] (u1)
  {
    \begin{tikzpicture}[solid, remember picture]
      \begin{scope}[every node/.style={circle, draw, minimum size=9pt, inner sep=2pt, outer sep=0pt,
        node distance=20pt}]
        \node[] (u1S1) {\idsize 1};
        \node[right=of u1S1] (u1S2) {\idsize 2};
      \end{scope}
      \begin{scope}[every path/.style={-stealth'}]
        \draw[bend left] (u1S2) to node[above] {\lblsize $A$} (u1S1);
        \draw[bend left] (u1S1) to node[above] {\lblsize $A$} (u1S2);
      \end{scope}
      \node[below left=7pt and -0pt of u1S2] {$u_1 = \varepsilon$};
    \end{tikzpicture}
  };

  \node[dashed, circle, draw, inner sep=0pt, outer sep=0pt, below left=20pt and 20pt of u1] (u2)
  {
    \begin{tikzpicture}[solid, remember picture]
      \node[] (u2A) {$A\to$};
      \begin{scope}[every node/.style={circle, draw, minimum size=9pt, inner sep=2pt, outer sep=0pt,
        node distance=20pt}]
        \node[right=5pt of u2A, fill=black, text=white] (u2A3) {\idsize 3};
        \node[right=of u2A3, rectangle, inner sep=2pt] (u2AC) {\lblsize $C$};
        \node[below=of u2AC, Red] (u2A2) {\idsize 2};
        \node[above=of u2AC, Red] (u2A1) {\idsize 1};
        \node[left=of u2A1, fill=black, text=white] (u2A4) {\idsize 4};
      \end{scope}
      \begin{scope}[every path/.style={-stealth'}]
        \draw (u2A1) -- node[above] {\lblsize $B$} (u2A4);
      \end{scope}
      \draw (u2AC) -- node[above, pos=.3] {\scriptsize $1$} (u2A3);
      \draw (u2AC) -- node[right, pos=.3] {\scriptsize $2$} (u2A2);
      \draw (u2AC) -- node[right, pos=.3] {\scriptsize $3$} (u2A1);
    \end{tikzpicture}
  };

  \node[dashed, circle, draw, inner sep=0pt, outer sep=0pt, below left=20pt and 5pt of u2] (u3)
  {
    \begin{tikzpicture}[solid, remember picture]
      \node[] (u3B) {$B \to$};
      \begin{scope}[every node/.style={circle, draw, minimum size=9pt, inner sep=2pt, outer sep=0pt,
        node distance=20pt}]
        \node[right=5pt of u3B, fill=black, text=white] (u3B2) {\idsize 2};
        \node[right=of u3B2, Orange] (u3B1) {\idsize 1};
        \node[below=of u3B1, fill=black, text=white] (u3B3) {\idsize 3};
      \end{scope}
      \begin{scope}[every path/.style={-stealth'}]
        \draw (u3B2) -- node[above] {\lblsize $a$} (u3B1);
        \draw (u3B1) -- node[right] {\lblsize $a$} (u3B3);
      \end{scope}
    \end{tikzpicture}
  };

  \node[dashed, circle, draw, inner sep=0pt, outer sep=0pt, below right=20pt and 5pt of u2] (u4)
  {
    \begin{tikzpicture}[solid, remember picture]
      \node[] (u4C) {$C \to$};
      \begin{scope}[every node/.style={circle, draw, minimum size=9pt, inner sep=2pt, outer sep=0pt,
        node distance=20pt}]
        \node[right=5pt of u4C, fill=black, text=white] (u4C4) {\idsize 4};
        \node[right=of u4C4, blue] (u4C1) {\idsize 1};
        \node[above=of u4C1, fill=black, text=white] (u4C2) {\idsize 2};
        \node[below=of u4C1, fill=black, text=white] (u4C3) {\idsize 3};
      \end{scope}
      \begin{scope}[every path/.style={-stealth'}]
        \draw (u4C2) -- node[right] {\lblsize $a$} (u4C1);
        \draw (u4C1) -- node[above] {\lblsize $a$} (u4C4);
        \draw (u4C1) -- node[right] {\lblsize $D$} (u4C3);
      \end{scope}
    \end{tikzpicture}
  };

  \node[dashed, circle, draw, inner sep=0pt, outer sep=0pt, below=20pt of u4] (u5)
  {
    \begin{tikzpicture}[solid, remember picture]
      \node[] (u5D) {$D \to$};
      \begin{scope}[every node/.style={circle, draw, minimum size=9pt, inner sep=2pt, outer sep=0pt,
        node distance=20pt}]
        \node[right=5pt of u5D, fill=black, text=white] (u5D2) {\idsize 2};
        \node[right=of u5D2, olive] (u5D1) {\idsize 1};
        \node[right=of u5D1, fill=black, text=white] (u5D3) {\idsize 3};
      \end{scope}
      \begin{scope}[every path/.style={-stealth'}]
        \draw (u5D2) -- node[above=1pt] {\lblsize $b$} (u5D1);
        \draw (u5D1) -- node[above=1pt] {\lblsize $b$} (u5D3);
        \draw[bend left=45] (u5D3) to node[below=1pt] {\lblsize $b$} (u5D2);
      \end{scope}
    \end{tikzpicture}
  };

  \node[below right=20pt and 20pt of u1, shape=isosceles triangle, shape border rotate=90, draw,
  align=center, anchor=north, Green] (Atree) {$\mathsf{dt}(A)$};

  \draw (u1) -- (u2);
  \draw (u2) -- (u3);
  \draw (u2) -- (u4);
  \draw (u4) -- (u5);
  \draw (u1) -- (Atree.north);

  \node[below right=1em and -1em of u2.north west] {$u_2 = 1$};
  \node[right=2pt of u3.south west] {$u_3 = 1.1$};
  \node[right=2pt of u4.south west] {$u_4 = 1.2$};
  \node[right=2pt of u5.south west] {$u_5 = 1.2.1$};

  \begin{scope}[overlay, every path/.style={dashed, Black!20}]
    \draw (u2A3) -- (u1S1);
    \draw (u2A4) -- (u1S2);

    \draw (u3B2) -- (u2A1);
    \draw (u3B3) -- (u2A4);

    \draw (u4C2) -- (u2A3);
    \draw (u4C3) -- (u2A2);
    \draw (u4C4) -- (u2A1);

    \draw (u5D2) -- (u4C1);
    \draw (u5D3) -- (u4C3);
  \end{scope}

%  \begin{scope}[overlay, every path/.style={very thick, -stealth', Blue, dotted}]
%    \draw (u1S1) .. controls +(-1,-1) and +(-.5,.5) .. node[right=-1pt,pos=.62]
%    {$\outEdge_{1,2}(S, 1, a)$} (u4C1);
%    \draw (u4C1) .. controls +(-0.5,-1) and +(0,.5) .. node[right=-1pt,pos=.5] {$\outEdge_{1,2}(C, 1, b)$} (u5D1);
%  \end{scope}
%  \begin{scope}[overlay, every path/.style={very thick, dashed, -stealth', Red}]
%    \draw (u4C4) .. controls +(0,1) and +(.5,-.5) .. node[left=-1pt,pos=.2] {} (u2A1);
%%    \draw (u5D3) .. controls +(2,1) and +(2,-1) .. node[right=-1pt,pos=.5] {$\intern$} (u4C3);
%%    \draw (u4C3) .. controls +(-3,.5) and +(0,-1) .. node[left=-1pt,pos=.5] {$\intern$} (u2A2);
%    \draw (u5D3) .. controls +(-6,5) and +(-1,-2) .. node[left=-1pt,pos=.5] {} (u2A2);
%  \end{scope}

%  \node[right=1.5 of u4C1, align=left, inner sep=0pt, outer sep=0pt, anchor=north west] (id1) {$\begin{aligned}
  \node[below right=1 and 1 of u4C1, align=left, inner sep=0pt, outer sep=1pt, anchor=north west] (id1) {$\begin{aligned}
%      \id(u_4, 1) &=\firstID(u_4) + 1 \\
%                  &=\firstID(u_2) + |V_{\rhs(A)}| - |\ext_{\rhs(A)}| + \nodes(B) + 1\\
%                  &= \firstID(u_1) + |V_S| - |\ext_S| + |V_{\rhs(A)}| - |\ext_{\rhs(A)}| + \nodes(B)
%                  + 1 \\
%                  &= 0 + 2 - 0 + 4 - 2 + 1 +1 = 6.
      \id(u_4, 1) &= 6.
%      \id(u_4, 1) &=\firstID(u_4) + 1 \\
%      &= 5 + 1 = 6
  \end{aligned}$};
  \begin{scope}[overlay, every path/.style={loosely dashed, very thick, -stealth'}]
      \draw (u4C1) -- (id1.168);
  \end{scope}
  
\end{tikzpicture}

%% file: figures/full_graph.tex
\begin{tikzpicture}
  \begin{scope}[every node/.style={circle, draw, inner sep=1pt, outer sep=0pt, node distance=20pt,
    minimum size=15pt}]
    \node[Green] (10) {\idsize 10};
    \node[left=of 10] (1) {\idsize 1};
    \node[above=of 1, blue] (6) {\idsize 6};
    \node[left=of 6, olive] (7) {\idsize 7};
    \node[left=of 7, Red] (4) {\idsize 4};
    \node[above=of 6, Red] (3) {\idsize 3};
    \node[right=of 3, Orange] (5) {\idsize 5};
    \node[right=of 5] (2) {\idsize 2};
    \node[below=of 2, Green] (11) {\idsize 11};
    \node[right=of 11, Green] (12) {\idsize 12};
    \node[right=of 12, Green] (9) {\idsize 9};
    \node[below=of 11, Green] (8) {\idsize 8};
  \end{scope}

  \begin{scope}[every path/.style={-stealth'}]
    \draw (10) -- node[above=-2pt] {\lblsize $a$} (1);
    \draw (1) -- node[left=-2pt] {\lblsize $a$} (6);
    \draw (6) -- node[above=-2pt] {\lblsize $b$} (7);
    \draw (7) -- node[above=-2pt] {\lblsize $b$} (4);
    \draw[bend left=45] (4) to node[above=-2pt] {\lblsize $b$} (6);
    \draw (6) -- node[left=-2pt] {\lblsize $a$} (3);
    \draw (3) -- node[above=-2pt] {\lblsize $a$} (5);
    \draw (5) -- node[above=-2pt] {\lblsize $a$} (2);
    \draw (2) -- node[left=-2pt] {\lblsize $a$} (11);
    \draw (11) -- node[above=-2pt] {\lblsize $b$} (12);
    \draw (12) -- node[above=-2pt] {\lblsize $b$} (9);
    \draw[bend right=45] (9) to node[above=-2pt] {\lblsize $b$} (11);
    \draw (11) -- node[left=-2pt] {\lblsize $a$} (8);
    \draw (8) -- node[above=-2pt] {\lblsize $a$} (10);
  \end{scope}
\end{tikzpicture}

%% file: traversing.tex
\section{Traversal of graphs represented by SL HR grammars}\label{sse:traversing}
Let us clarify what we mean by traversing a graph. Let $g$ be a graph, $x \in V_g$ a vertex, and $e
\in E_g$ a hyperedge attached to $x$. Let $k$ such that $\att_g(e)[k] = x$, and $l \in [\rank(e)]$.
The operation we wish to support is $\outEdge_{k,l}(x, \sigma) = y$ where $\sigma = \lbl_g(e)$ and
$y = \att_g(e)[l]$. Figure~\ref{fig:traverseIllus} shows a schematic example of this operation. We
call this operation a \emph{traversal step} and refer to the vertices $x$ and $y$ as
\emph{$\sigma$-$k$-$l$-neighbors}. We therefore are not looking for random access, but instead,
given a node, we wish to traverse to one of its neighbors. Note that in a unique-labeled graph every
vertex $x$ only has at most one $\sigma$-$k$-$l$-neighbor for any combination of $\sigma$, $k$, and
$l$. Further note that for ``regular'' directed edges only one of the numbers $k$ and $l$ is really
needed, as there are only two valid combinations: a $1$-$2$-neighbor is a successor of the current
node, and a $2$-$1$-neighbor a predecessor. The more general notation used here is to allow for any
traversal along hyperedges of higher rank.
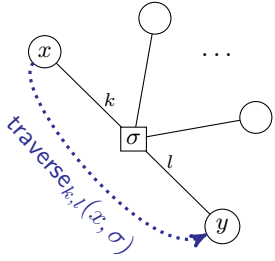
\begin{figure}[!t]
  \centering
  \input{figures/traverseIllus}
  \caption{Illustration of the $\outEdge$-mapping}
  \label{fig:traverseIllus}
\end{figure}

For a grammar $G$ we wish to traverse $\val(G)$ in this fashion, but without explicitly computing
$\val(G)$. We explain how using an example and intuitively define the necessary notions along the
way. Figure~\ref{fig:traversal_dt} contains the graph $\val(G)$ on the right and one branch of the
derivation tree $\dt(G)$ on the left. Consider first $\val(G)$: starting from vertex $1$, we intend
to follow the $a$-edge once and then the $b$-edge twice. The appropriate traversal steps are
indicated by annotations. We next explain for each of these traversal steps, how they can be
recovered in $\dt(G)$. In a first step we need to find the vertex that represents vertex $1$ in
$\val(G)$, which is found immediately in the start graph: $\id(u_1,1) = 1$. We next extend the
$\outEdge$-mapping to rules of a grammar. For a nonterminal $A$ and a vertex $x$ in $\rhs(A)$ the
mapping $\outEdge_{k,l}(A,x,\sigma)$ intuitively points us to the vertex $y$ within a node label of
the derivation tree $\dt(A)$ such that the vertices represented by $x$ and $y$ in $\val(A)$ are
$\sigma$-$k$-$l$-neighbors. In our example the mapping $\outEdge_{1,2}(S,1,a)$ points towards the
vertex $1$ in the node $u_4$ of $\dt(S)$. Note that the arrow visualizing this mapping is drawn in
such a way, as to also visualize how this vertex is found: from vertex $1$ in $u_1$ we ``enter'' the
$A$-rule and then the $C$-rule from their first external vertices each until the edge labeled $a$ is
found. Thus, the mapping $\outEdge_{1,2}(S,1,a) = (u_4, 1)$, i.e., it returns a vertex specified by
a tuple consisting of a node in $\dt(S)$ and a vertex in $g_{u_4}$. To conclude the traversal step
we compute $\id(u_4, 1) = 6$, which is the vertex number in $\val(G)$.

The next step starts at vertex $(u_4, 1)$ and we intend to follow the $b$-edge. We begin in the same
way as in the previous step: ``entering'' the $D$-rule from its first external vertex and following
the edge labeled $b$ we arrive at vertex $(u_5, 1)$. However, this is not quite how the
$\outEdge$-mapping should be defined in this case: There may be multiple nodes in the derivation
tree labeled $C \to \rhs(C)$. And since the grammar is straight-line all of these nodes have the
same children, and thus should have the same result for $\outEdge_{1,2}(C,1,b)$. Therefore
$\outEdge_{1,2}(C,1,b)$ should map to a Dewey address relative to the derivation tree $\dt(C)$.
Since the $D$-node is the first (and only) child of the $C$-node its Dewey address within $\dt(C)$
is $1$ and thus $\outEdge_{1,2}(C,1,b) = (1, 1)$. As we know that we started the traversal step in
vertex $(u_4,1)$ we can recover the node relative to $\dt(S)$ in this way, by concatenating the
relative address to $u_4$ and the position after two traversal steps is the vertex $(u_4.1, 1)$ with
$\id(u_4.1, 1) = \id(u_5, 1) = 7$.

Finally, we wish to follow one more $b$-edge starting at $(u_5,1)$. Finding the neighboring vertex
is easy in this case, since the $b$-edge is found in the same rule. Therefore $\outEdge_{1,2}(D,1,b)
= (\varepsilon, 3)$, i.e., the vertex $(u_5, 3)$. However, this vertex is external. As an external
vertex does not directly represent a vertex of $\val(G)$ the final step is to find the internal
vertex it is merged with. Fortunately, this can always be found in an ancestor within the derivation
tree. As we can see in Figure~\ref{fig:traversal_dt} (green dashed arrow) the vertex is merged first
with $(u_4, 3)$ and then with $(u_2, 2)$, which is internal. Therefore we define $\internal(u_5,3) =
(u_2,2)$ and in general $\internal(u, x) = (v, y)$ where $u$ is a node of the derivation tree, $x$
is a vertex in $g_u$, $v$ is an ancestor of $u$ and $y$ the internal vertex in $g_v$ that is merged
with $x$ during derivation. If $x$ already is an internal vertex, we just let $\internal(u,x) =
(u, x)$.

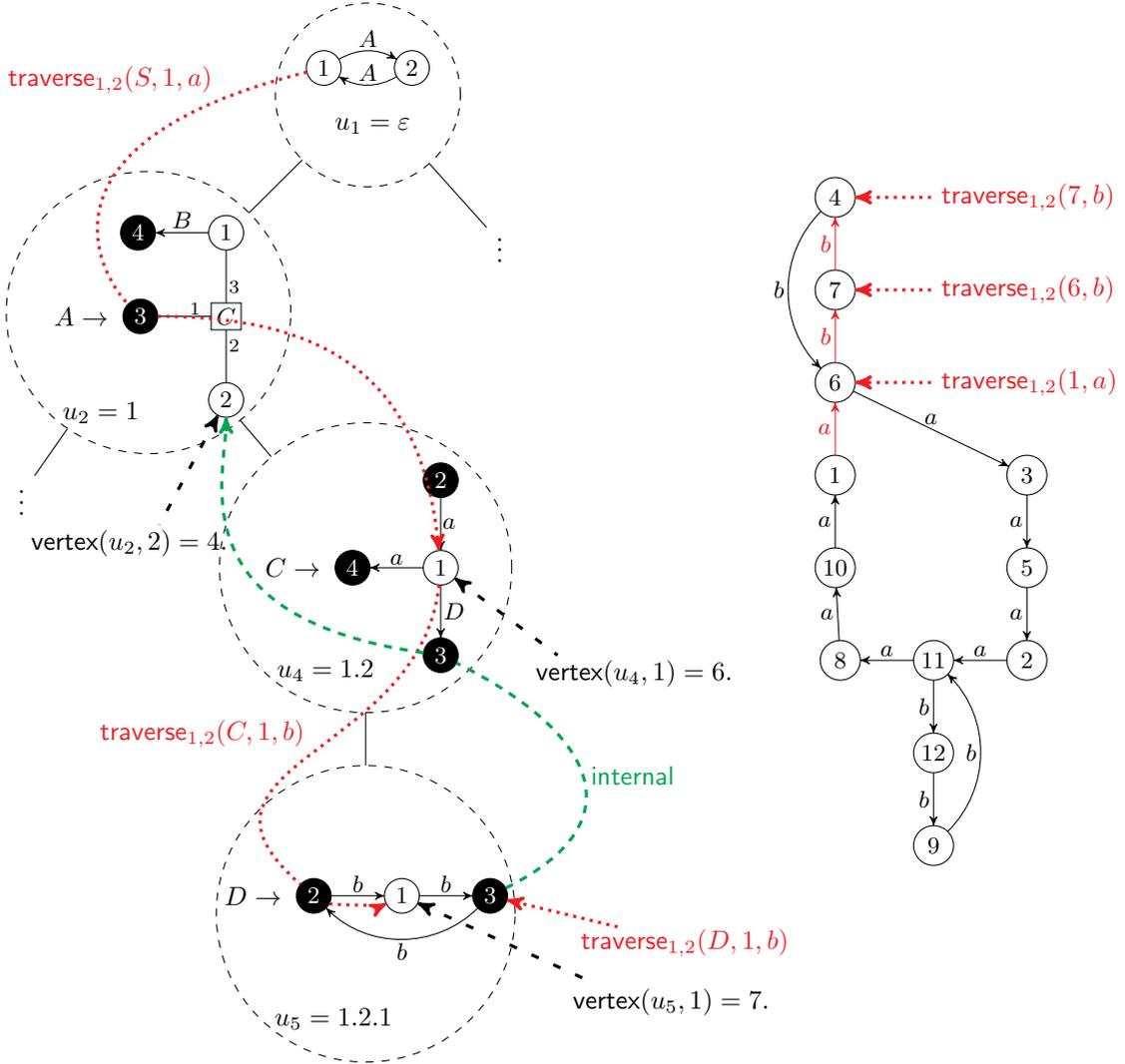
\begin{figure}[!t]
  \centering
  \input{figures/traversal_dt}
  \caption{Part of the derivation tree from Figure~\ref{fig:annotatedDT} and the full graph with
  added annotations about 3 traversal steps.}
  \label{fig:traversal_dt}
\end{figure}
We can formally define the $\outEdge$-mapping recursively in the following way:
\begin{itemize}
    \item $\outEdge_{k,l}(A,x,\sigma) = (\varepsilon, y)$, if there exists $e \in E_{\rhs(A)}$ with
      $\att(e)[k] = x$, $\att(e)[l] = y$ and $\lbl(e) = \sigma$.
    \item $\outEdge_{k,l}(A,x,\sigma) = (v.i,y)$ if $x = \att(e_i)[j]$ and $\outEdge(\lbl(e_i),
      \ext_{\rhs(\lbl(e_i))}[j], \sigma) = (v,y)$.
\end{itemize}
Note that this definition implies the mapping can be precomputed for every combination of $A$, $x$,
and $\sigma$ by one bottom-up (i.e., in reverse $\SLOrder$) pass of the grammar.
A similar definition can be given of the aforementioned $\nodes(A)$, which is the number of internal
nodes in $\val(A)$ for a nonterminal $A$. This mapping can be defined as $\nodes(A) = \sum
\{\nodes(\lbl(e)) \mid e \in E^{\nt}_{\rhs(A)}\} + |V_{\rhs(A)}| - |\ext_{\rhs(A)}|$ which also
allows for a bottom-up computation in one pass.

In summary, to traverse a
$\sigma$-edge that is attached at index $k$ to $\id(u,x)$ in $\val(G)$ towards the node attached at
index $l$ we
\begin{enumerate}
  \item compute $\outEdge_{k,l}(A_u, x, \sigma) = (v', y')$,
  \item let $w = u.v'$ be the address of the node containing $y'$ in its label with respect to
    $\dt(G)$,
  \item compute $\internal(w, y') = (v,y)$, and
  \item compute $\id(v,y)$.
\end{enumerate}

In Section~\ref{sss:tableaux} below we introduce a data structure that allows us to do each of these
steps in constant time for unique-label SL-HR grammars of bounded rank. This is our main result.
\subsection{Delay of naive implementation and other approaches}\label{sss:naive_delay}
For the methods outlined in the introduction to Section~\ref{sse:traversing} the delay of a naive
implementation can only be given as $O(\height(G))$. It is not hard to construct a grammar (e.g.
Figure~\ref{fig:rootToLeafJump}) where every step from one node to another corresponds to a step
from the root to a leaf in the derivation tree (or vice versa), and thus the time needed to compute
$\internal$ or $\outEdge$ for this grammar is $O(\height(G))$.
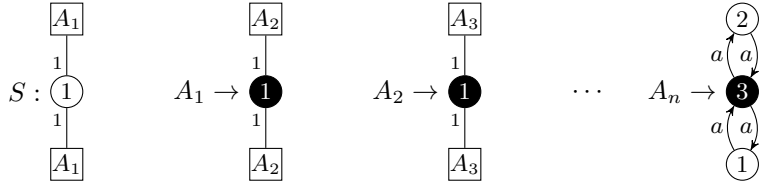
\begin{figure}[!t]
  \centering
  \input{figures/rootToLeafJumps}
  \caption{Grammar where every traversal step corresponds to a jump between the root and a leaf in the
  derivation tree.}
  \label{fig:rootToLeafJump}
\end{figure}

For straight-line programs (context-free string grammars with the straight-line restriction), or
SLPs, it is known that one traversal step (here: going from one symbol to the next) can be done in
constant time using linear time/space preprocessing~\cite{DBLP:conf/dcc/GasieniecKPS05}. This works
by visiting the leafs of the derivation tree of an SLP in order, which results in the string
represented. A data structure based on tries is used to guarantee that the move from one leaf of the
derivation tree to the next can always be done in constant time. The data structure solves the
``next link problem'' in constant time. For a given leaf $u$ this problem asks for the node $w$ and
the number $k$ such that $u$ is in the $k$-th subtree of $w$ and the next leaf $v$ is in the
$(k+1)$-th subtree of $w$, i.e., $w$ is the lowest common ancestor (LCA) of $u$ and $v$. This operation
can be carried out in constant time (after linear-time preprocessing) using any of several
well-known tree indexing data structures (e.g. a variant of finding
  LCA~\cite{DBLP:journals/siamcomp/SchieberV88} described in the
  book~\cite{DBLP:books/cu/Gusfield1997}, or over first-child/next-sibling encoded binary trees,
  with simple and efficient LCA structures~\cite{DBLP:journals/tcs/BenderF04} --- cf. the discussion
at the end of the article~\cite{DBLP:journals/algorithmica/LohreyMR18}). The
article~\cite{DBLP:journals/algorithmica/LohreyMR18} further generalizes the method to SL tree
grammars provided the grammar is given in a normal form, which structures the derivation into
``string-like'' parts that are branched off from. It is unlikely that this method generalizes to
SL-HR grammars for at least two reasons: first there is a major difference in the grammar
formalisms. In string and tree grammars the objects that are traversed (consecutive symbols of the
string, nodes in the tree) are also the terminals of the grammar. In HR grammars however, terminals
and nonterminals are (hyper-)edges of the graph, but we wish to traverse the graph's vertices along
the edges. An important effect is that it is not reasonably possible to confine the objects of
interest (i.e., the vertices) to the leafs of the derivation tree. Such a restriction would be very
strong and would make it impossible to, for example, construct a grammar that represents a string
graph of $2^n$ consecutive $a$-edges with a grammar of $O(n)$ size. It can therefore be assumed that
traversal of the represented graph requires navigation on arbitrary nodes in the derivation tree,
instead of only navigating from leaf to leaf.  Furthermore, as we see in the example discussed in
the introduction to Section~\ref{sse:traversing}, traversing from one vertex in $\val(G)$ to another
may incur ``jumps'' in the derivation tree, i.e., the navigation in the derivation tree is not
restricted to single parent/child steps, but may require longer paths.

\subsection{Tableaux and how to use them}\label{sss:tableaux}
We seek a data structure that allows to compute, in constant time, each of the three operations
$\outEdge$, $\internal$, and $\id$, and which can be precomputed in polynomial time. The data
structure we construct in the following essentially achieves this by precomputing $\internal$ and
$\firstID$ for every valid $\outEdge$-input. To get an intuitive understanding of this, we first
make some more observations about the example from Figure~\ref{fig:traversal_dt} above and then
extend the example with our data structure. We then define it and explain how it is built
afterwards.

We note the following property of $\outEdge$ and $\internal$ that will be important: starting from a
node $u$ in the derivation tree, $\outEdge$ will always lead to a successor of $u$, whereas
$\internal$ will always lead to an ancestor. Neither mapping ever branches to a sibling. For the
traversal we therefore need to maintain a structure that represents only one branch of the
derivation tree. As we already observed before, we can compute $\outEdge(A,\sigma,x)$ for any
combination of $A$, $\sigma$, and $x$ by considering only $\dt(A)$, i.e., a subtree of $\dt(S)$.
This yields a branch of $\dt(A)$ which we can concatenate with the current branch of $\dt(S)$ to get
the full branch for the target vertex (compare $\outEdge_{1,2}(C,1,b)$ in the above example). Since
$\outEdge$ only depends on $\dt(A)$ and there are at most $|N||\Sigma||G|$ possible combinations of
$A$, $\sigma$, and $x$ these ``partial branches'' can be precomputed in polynomial time. This is the
core of our data structure: we model every partial branch using a matrix (we call it a ``tableau'')
of pointers, where each row of the matrix corresponds directly to a node in the branch and each
column corresponds to the external vertices. The intention is that a cell -- again, row and column
denote a node $v$ and external vertex within $g_v$ -- has a pointer to another cell which contains
the internal vertex that is merged with the external vertex the first cell refers to. The immediate
problem is, that while $\outEdge$ is defined exclusively with relation to $\dt(A)$,
$\internal$ is only valid on full branches of $\dt(S)$ since the internal vertex merged with an
external vertex may be in any ancestor within the derivation tree. Compare again the above example
in Figure~\ref{fig:traversal_dt}: here $\internal(u_5, 3) = (u_2, 2)$, but we actually reached $u_5$
by traversing from a vertex in $u_4$ using $\outEdge_{1,2}(C,1,b)$. We thus cannot precompute the
full mapping if we only consider $\dt(C)$. The solution is simple: we precompute as much of it as we
can, and when concatenating a partial branch with the current full branch of $\dt(S)$ we connect the
ends. Similarly for $\firstID$: for any $\outEdge(A,\sigma,x)$ we precompute $\firstID$ relative to
$\dt(A)$, and then add an appropriate offset once the branches are concatenated.

To exemplify these intuitions, let us consider how our proposed data structure would handle the
example from Figure~\ref{fig:traversal_dt}. Figure~\ref{fig:traversal_tab_ex} contains the branch of
the derivation tree next to two matrices which model the branches $\outEdge_{1,2}(S,1,a)$ and
$\outEdge_{1,2}(C,1,B)$, respectively. Both matrices have additional annotations $\nt$ and
$\firstID$, and we usually refer to the entire structure (with the additional annotations) as a
``tableau''. Consider first the upper tableau $t_1$. It models the branch $u_4 = 1.2$ of the
derivation tree $\dt(G)$ which is what we need for $\outEdge_{1,2}(S,1,a)$. Each of the tableau's
three rows directly corresponds to one of the nodes in the derivation tree (denoted by light gray
lines). Every cell of the main matrix of the tableau is split into two parts: a pointer $\reg$ and a
number $\nde$. We explain in the following how this tableau is used to determine $\id$ and
$\internal$. To compute the vertex number of vertex $1$ in node $u_4$ with respect to $\val(G)$ we
can use the $\firstID$-values provided by the tableau: the third row corresponds to $u_4$ and has
$\firstID = 5$. Thus $\id(u_4,1) = 5+1 = 6$, just as discussed previously. Analogously, if we wanted
to determine the vertex number of vertex $2$ in node $u_2$ we would use the $\firstID$-value in the
second row of the tableau, which is $2$ and thus $\id(u_2,2) = 2+2 = 4$. The other information we
can take from the tableau is the $\internal$-mapping. In Figure~\ref{fig:traversal_tab_ex} we
denoted by a dashed green edge the mapping $\internal(u_4,3)$, which is the second external vertex
of the graph $\rhs(C)$. Consequently we find, in the second column of the third row of $t_1$, an
edge (also colored green) pointing to a cell in the row above. This cell has a $\nde$-value of $2$,
which tells us that the vertex $2$ in the node corresponding to the second row (i.e., $u_2$) is the
internal vertex merged with this external vertex. The reader may have noticed that some arrows in
the tableau point to rows below instead of above, which contradicts the previous intuition that the
$\internal$-mapping always points to an ancestor in the derivation tree. We will explain the reason
for this in the next paragraph, but for now already mention the rules for these arrows:
\begin{itemize}
  \item An arrow going up (we speak of an \emph{upwards pointer}) \emph{always} ends at a cell
    with a $\nde$-value containing the number of the vertex that is merged with the external vertex
    corresponding to the start of the arrow. If such a vertex does not exist within the scope of
    the tableau, the pointer instead ends at the cell of the first row representing the external
    vertex the current vertex is merged with.
  \item An arrow going down (we refer to these as \emph{downwards pointers}) \emph{always} ends at a
    cell containing an upwards pointer, which, when followed, leads to the
    correct internal vertex that is merged with the external vertex at either arrows end.
\end{itemize}
Thus there are two cases for computing the $\internal$-mapping using the tableau: if $\reg$ contains
an upwards pointer, we follow it and are done. If $\reg$ contains a downwards pointer, we follow it
and also the upwards pointer at its end. This way, after at most 2 steps, we always end at the
correct position.

\begin{figure}[!t]
  \centering
  \input{figures/traversal_tab_ex}
  \caption{Branch of a derivation tree (left) and two tableaux, each representing parts of the
  branch.}
  \label{fig:traversal_tab_ex}
\end{figure}
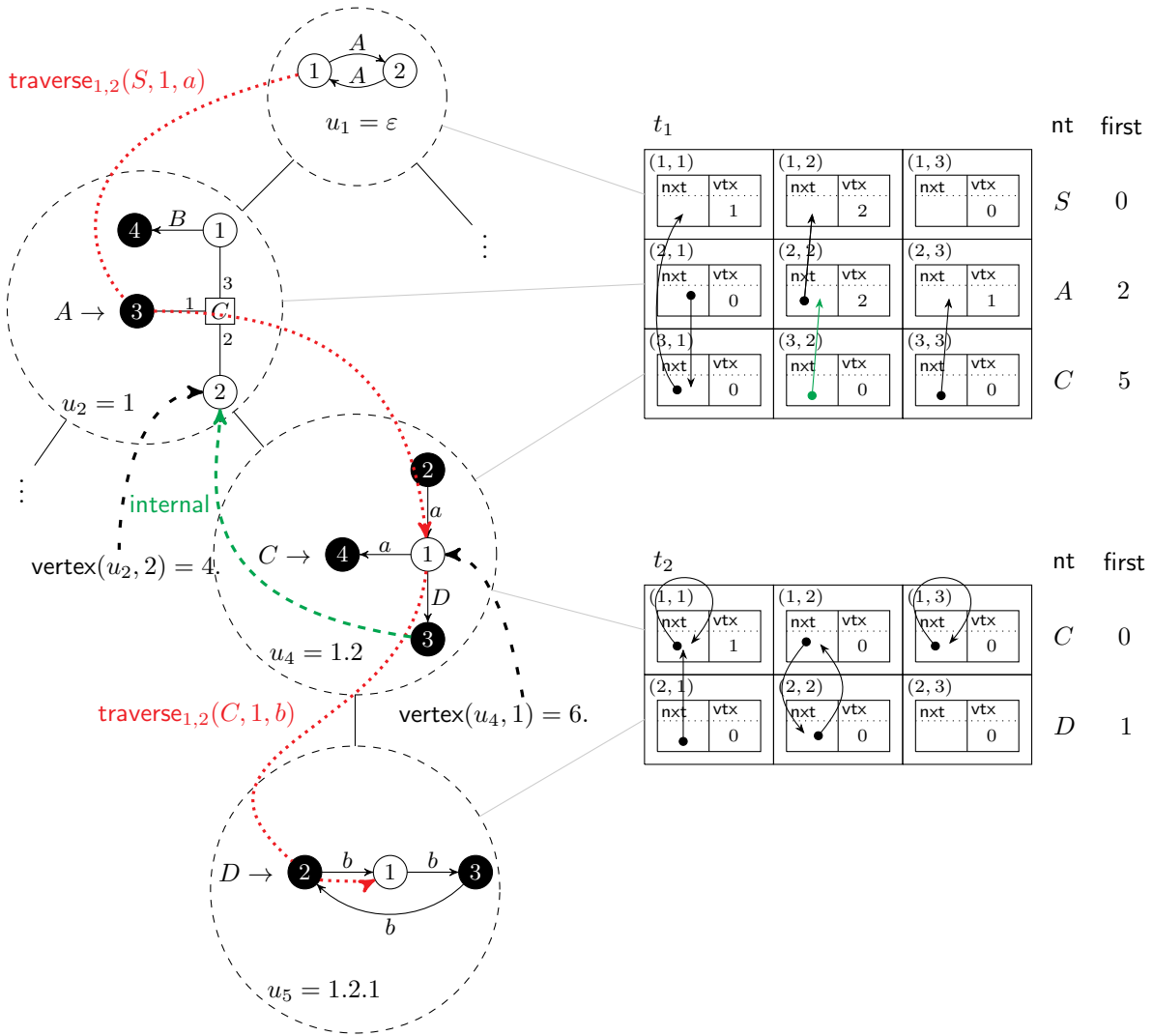

We did not yet mention the second tableau $t_2$ in Figure~\ref{fig:traversal_tab_ex}. It is built in
the same way as $t_1$, and correctly models the branch of the derivation tree relevant for
$\outEdge_{1,2}(C,1,b)$, \emph{with respect to $\val(C)$}. In $\val(C)$ the external vertices of
$\rhs(C)$ are never merged with any other node, therefore all three $\reg$ pointers in the first row
of $t_2$ are cyclic. If a $\reg$ pointer just cycles back to the origin (even if that happens in two
steps as is the case for the second external vertex here), this means that there is no internal
vertex in the scope of this tableau merged with the external vertex. Note that this can never happen
for full branches of $\dt(S)$, since the start graph is defined to be of rank $0$. This behavior
matches the definition of $\outEdge_{1,2}(C,1,b)$ which gives a Dewey address relative to $\dt(C)$.
Just like this address is concatenated with the one of the initial node ($u_4$ in our example), the
tableau $t_2$ needs to be ``concatenated'' to $t_1$ to model the entire branch of the derivation
tree. The result of this process can be seen in Figure~\ref{fig:traversal_tab_att_ex}. Since both
the last row of $t_1$ and the first row of $t_2$ correspond to the same node in the derivation tree
we can simply connect them column-wise. Herein lies the reason for the existence of downwards
pointers: they occur if more than one external vertex within the branch merges with the same
internal vertex.  To make sure that in these cases we only ever have to change one pointer when
concatenating two tableaux, we use the two types of pointers and only need to change the
upwards pointer (seen here in the second column of $t_2$). Otherwise the time to concatenate two
tableaux could not be bounded by $O(\maxRank)$. One more addition needs to be made at this point.
The $\firstID$ values in $t_2$ are given with respect to $\dt(C)$. The final part of the
tableau-data structure is thus an offset value. When concatenating $t_2$ to $t_1$, we use the
$\firstID$-value of the last row in $t_1$ (and any offset it may already have) as the offset in
$t_2$. This ensures that the last row of $t_1$ and the first row of $t_2$ have the same $\firstID$
values, which should be the case as they both represent the same node. In
Figure~\ref{fig:traversal_tab_att_ex} we can thus see that $\internal(u_5, 3)$ is correctly modeled
by the concatenated tableaux (green arrows) and that $\id(u_5,1)$ can be correctly computed using the
$\firstID$-value of the corresponding row in $t_2$ (its last row) and the offset: $\id(u_5,1) = 1 +
1 + 5 = 7$. We invite the reader to check that the other external vertices are also correctly mapped
to their internal equivalents.

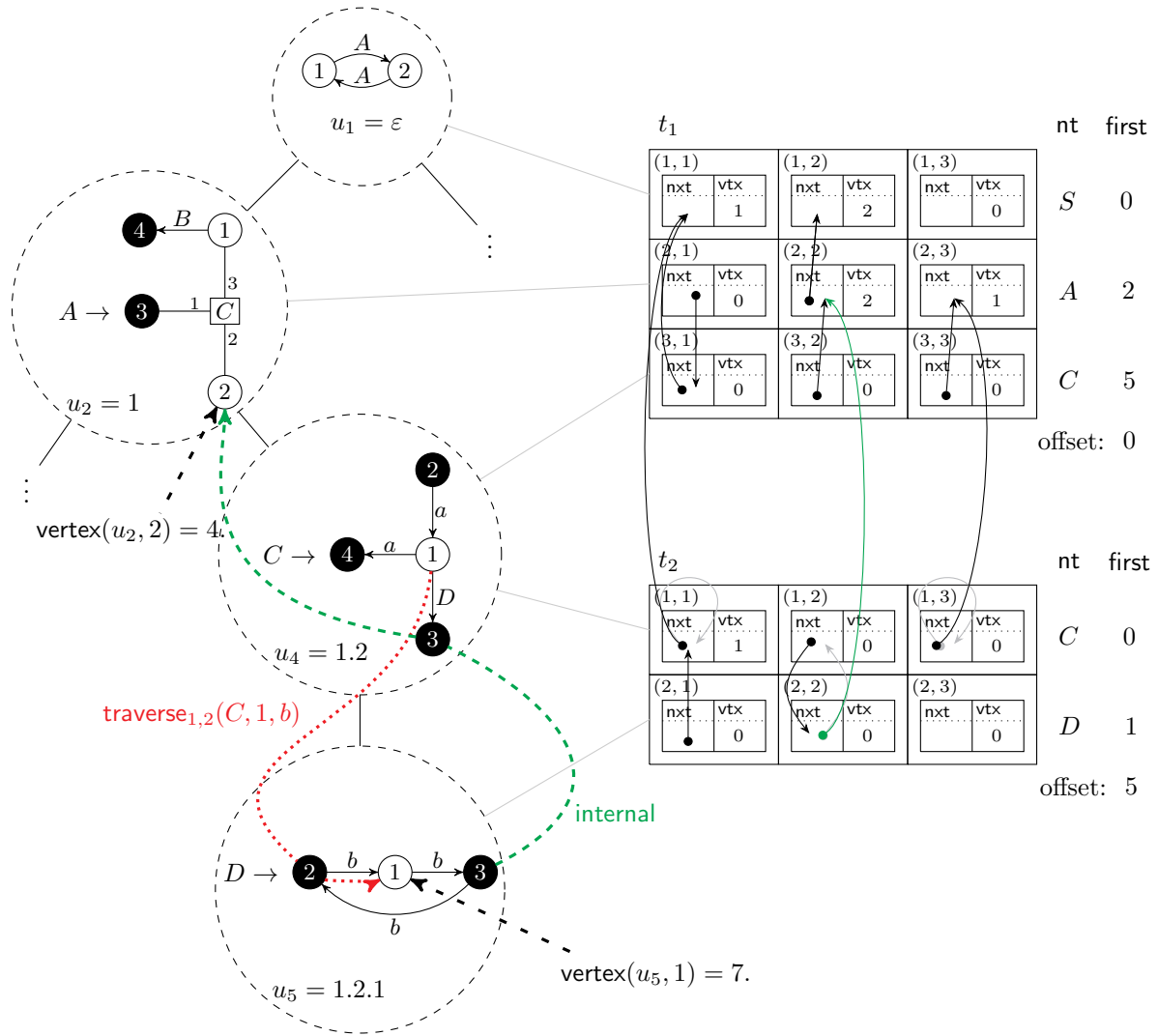
\begin{figure}[!t]
  \centering
  \input{figures/traversal_tab_att_ex}
  \caption{Concatenating the pointers of the two tableaux and setting an offset turns them into one
  tableau that represents the full branch of the derivation tree.}
  \label{fig:traversal_tab_att_ex}
\end{figure}

\subsection{Precomputing Tableaux}
In this section we will explain how to precompute the necessary tableaux in advance, but first let
us formalize their definition following the description in the previous section. We begin by
defining an auxiliary function to streamline how the tableau is used to support the $\internal$
mapping. As detailed in the previous section, we always follow one or two $\reg$ pointers depending on
whether the first one was downwards or not. Let $t(i,j)$ be the cell at row $i$ and column $j$ of
the matrix $t$. We then use $\reg(t(i,j)) = t'(i',j')$ to refer to the cell referenced by the $\reg$
pointer in $t(i,j)$, i.e., the $\reg$ pointer in cell $(i,j)$ of $t$ references the cell $(i',j')$
of $t'$. We formalize the distinction between upwards and downwards pointers we introduced above: we
call $\reg$ an upwards pointer if $t' \neq t$ or $t' = t$ and $i' \leq i$ and a downwards pointer if
$t' = t$ and $i' > i$. Then we define $\findNode(t(i,j))$ as 

\begin{equation*}
  \findNode(t(i,j)) = 
  \begin{cases}
    \reg(t(i,j)) & \text{if $\reg(t(i,j))$ is upwards} \\
    \reg(\reg(t(i,j))) & \text{if $\reg(t(i,j))$ is downwards}
  \end{cases}.
\end{equation*}

This function always maps to the cell containing the internal vertex merged with the $j$-th external
vertex of row $i$. As with the $\reg$ pointer, we use $\nde(t(i,j))$ to reference the vertex number
given in the cell $t(i,j)$. Thus the number of the internal vertex merged with the $j$-th external
vertex of row $i$ is $\nde(\findNode(t(i,j)))$. We formally define the properties of a tableau:

\begin{definition}
  Let $A$ be a nonterminal, let $u_m$ be a node in $\dt(A)$, and let $u_1, u_2, \ldots, u_{m-1}$ be
  the nodes on the path from the root of $\dt(A)$ ($= u_1$) to $u_m$. For any $i \in [m]$ let $u_i$
  be labeled with $A_i \to g_i$. Then $\tbl(u_m)$ is any tuple $(t, t^{\nt}, t^{\firstID},
  t^{\offset})$ where $t$ is an $m \times \maxRank$ matrix, $t^\nt$ and $t^{\firstID}$ are $m$
  column-vectors, $t^{\offset}$ is a number, and for every $i \in [m]$ the following properties are
  fulfilled:
  \begin{enumerate}
      \item $t^{\nt}(i) = A_i$
      \item $\tab^{\firstID}(i) = \firstID(u_i)$ (with respect to $\dt(A)$)
      \item  For $j \in [\rank(A_i)]$ let $y$ be the $j$-th external vertex of $g_i$. The cells of
        $t$ are filled such that they fulfill the following conditions:
        \begin{enumerate}
          \item $\internal(u_i,y) = (u_r,z)$ where $u_r$ is an ancestor of $u_i$ in $\dt(A)$ (i.e.,
            $r \leq i$). Then $\findNode(t(i,j)) = t(r, j')$ for some $j'$ and $\nde(t(r,j')) = z$.
          \item $\internal(u_i,y)$ is undefined. Then $\findNode(t(i,j)) = t(1,j')$ for some $j'$
            and $\nde(t(1,j')) = 0$.
          \item The pointer $\reg(t(i,j))$ is an upwards pointer if and only if there is no $i' > i$
            and $j'$ such that $\findNode(t(i,j)) = \findNode(t(i', j'))$.
          \item Unless otherwise specified by the above conditions, $\reg$ and $\nde$ for a cell are
            $0$.
        \end{enumerate}
      \item $\tab^{\offset} = 0$
  \end{enumerate}
    \label{def:correctTableau}
\end{definition}

This definition specifies the informal definition of the previous section: There is one row per
ancestor of $u_m$, each with appropriate nonterminal, and $\firstID$-values (items 1 and 2 in the
definition). Further, the cells are specified to contain tuples of $(\reg, \nde)$ such that the
internal vertex merged with a given external vertex is found in at most two steps (item 3). We
say that such a tableau \emph{models} the path from the root of $\dt(A)$ to the node $u_m$, or
simply the tableau models $u_m$. Note that, while a tableau $t$ by this definition models the node
$u_m$, its first $i$ rows also model its ancestor $u_i$ for any $i \in [m]$. A tableau thus can be
interpreted as an encoding of a specific branch of the derivation tree. The tableau for a traversal
step is consequently defined as the tableau modeling the path in the derivation tree:

\begin{definition}
  For a given nonterminal $A$, node $x$ in $\rhs(A)$, edge label $\sigma$ and indices $k,l$ let
  $\outEdge_{k,l}(A,x,\sigma) = (u, y)$. Then we define the tableau $\tbl_k(A,x,\sigma) = \tbl(u)$.
  \label{def:tableau}
\end{definition}

We precompute the tableaux for every possible combination of nonterminal $A \in N$, label $\sigma
\in \Sigma$, vertex $x \in V_{\rhs(A)}$, and outgoing neighbor index $k$. It is not necessary to
specify the incoming neighbor index $l$ at this point: let $l$ and $r$ be two different such
indices, then $\outEdge_{k,l}(A,x,\sigma) = (u,y)$ and $\outEdge_{k,r}(A,x,\sigma) = (u,y')$. So
while the vertex component of the tuples is different, in both cases the same node of the derivation
tree is referenced. Since the tableaux only represent the nodes of the derivation tree both
traversals can be represented by the same tableau. The tableaux are computed inductively using
Algorithm~\ref{alg:create}. For any given nonterminal $A$ we can represent the root node of $\dt(A)$
using the tableau $(t_\varepsilon, (A), (0), 0)$ where $t_{\varepsilon}$ is a $1 \times \maxRank$
matrix with

\begin{itemize}
  \item $\reg(t_{\varepsilon}(1,j)) = t_{\varepsilon}(1,j)$ for $j \leq \rank(A)$
  \item $\reg(t_{\varepsilon}(1,j)) = 0$ for $j > \rank(A)$, and
  \item $\nde(t_{\varepsilon}(1,j)) = 0$ for all $j$.
\end{itemize}

This generates upwards pointers for every external vertex of $\rhs(A)$, but since $A \to \rhs(A)$ is
the root of $\dt(A)$ these pointers just loop back to their origin. Note that this tableau models
the path consisting only of the root of $\dt(A)$ by Definition~\ref{def:correctTableau}. Starting
from this, the tableau $\tbl_k(A, x, \sigma)$ is constructed row by row, such that every
intermediate step is a valid tableau modeling the path in the derivation tree up to that point.
Assume therefore, that we already have a tableau modeling the path to $v$ in $\dt(A)$, and we wish
to extend it to model $u$, which is a child of $v$. This situation is illustrated in
Figure~\ref{fig:creation_illus}. The derivation tree $\dt(A)$ is on the left of the figure, the
nodes $v$ and $u$ are expanded with their labels ($B \to g_v$ and $C \to g_u$, respectively), though
only the relevant parts of the graphs $g_v$ and $g_u$ are shown. The external vertices of $g_u$ are
connected by colored lines with the vertices in the parent node $v$ they are merged with.
Information on the ancestors of $v$ is not necessary for the extension of the tableau $t_v$ and
therefore left out of the figure. The bottom row of the tableau $t_v$ is shown on the upper right of
the figure. Even though only the bottom row is shown, we can deduce some information on the external
vertices of $g_v$ from the $\reg$ pointers in it: namely that the first two external vertices (the
ones numbered 2 and 3 in $g_v$) are merged with other external vertices in the parent of $v$,
whereas the third external vertex (numbered 4) is merged with an internal vertex. We know this,
because the first two columns have incoming downwards pointers, whereas the third one does not.
Recall that downwards pointers are only used to denote that more than one external vertex is merged
with the same internal vertex. The bottom right of the figure shows the tableau $t_u$, which is the
result of Algorithm~\ref{alg:create} with input $A$ and $u$. We now explain how the algorithm
extends $t_v$ to obtain $t_u$ (lines~\ref{line:start} to~\ref{line:return} of the algorithm).

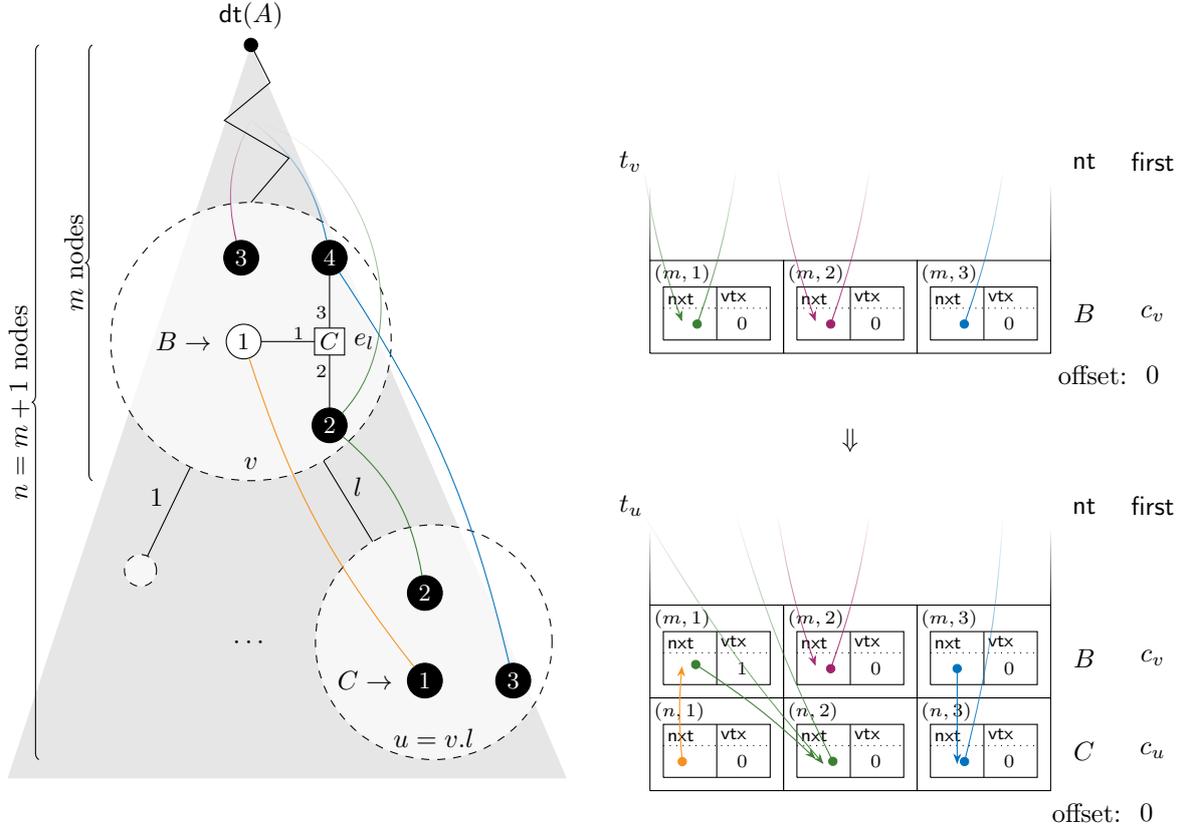
\begin{figure}[!t]
  \centering
  \input{figures/creation_illus}
  \caption{Extending the tableau modeling $v$ to make a valid tableau modeling $u$, a child of $v$.}
  \label{fig:creation_illus}
\end{figure}

To extend $t_v$ such that the resulting tableau models $t_u$, we need to add a row to the matrix
$t_v$, and elements to $t_v^{\nt}$ and $t_{v}^{\firstID}$. We focus on the former for now, since
extending the vectors $t_v^{\nt}$ and $t_v^{\firstID}$ is comparatively straightforward. We let $u$
be the $l$-th child of $v$. Thus $v$ has at least $l$ nonterminal edges, and we let $e_l$ be the
$l$-th one by the derivation order. To now extend $t_v$ into $t_u$ we need to make sure the $\reg$
pointers on the newly added last row all correctly represent the external vertices of $u$. Doing
this may involve making changes to already existing $\reg$ pointers in the rows above, if an
external vertex of $u$ is merged with an external vertex of $v$.  Let $y$ be the $j$-th external
vertex of $g_u$. During derivation it is merged with the $j$-th vertex attached to $e_l$, which we
name $x$ (cf. line~\ref{line:10} in Algorithm~\ref{alg:create}). There are now two possible cases:
$x$ could be internal or external. The simpler case is, if $x$ is internal. Then we add a $\reg$
pointer in column $j$ of the last row to reference the same column in the row above.  The $\nde$
element of that cell is further set to contain the vertex number of $x$. In
Figure~\ref{fig:creation_illus} this happens for the first ($j=1$) external vertex (orange color).
Lines~\ref{line:11} and~\ref{line:12} in Algorithm~\ref{alg:create} handle this case. In the second
case $x$ is itself an external vertex. Therefore both $x$ and $y$ ultimately represent the same
internal vertex found in an ancestor of $v$. In this case we intuitively move the entire cell of
representing $x$ in $t_v$ to the cell representing $y$ in $t_u$ and add one downwards pointer ``from
$x$ to $y$''. This way $\findNode$ yields the same result in both cases and Condition~3 of
Definition~\ref{def:correctTableau} is preserved.  In Figure~\ref{fig:creation_illus} this happens
for the second and third external vertices of $u$ (colored green and blue). We encourage the reader
to verify that lines~\ref{line:13} to~\ref{line:17} of Algorithm~\ref{alg:create} have exactly this
effect on $t_u$ in Figure~\ref{fig:creation_illus}. As previously mentioned, extending $t_v^{\nt}$
and $t_v^{\firstID}$ is straightforward by comparison: $t_v^{\nt}$ is extended by one element
containing $C$ since that is the nonterminal in the label of $u$. The element added to
$t_v^{\firstID}$ is the sum of
\begin{enumerate}
  \item the last element of $t_v^{\firstID}$ ($c_v$ in Figure~\ref{fig:creation_illus}),
  \item the number of internal vertices in $g_v$, and
  \item the sum of internal vertices generated by the $l-1$ nonterminal edges in $v$ that are
    derived before $e_l$.
\end{enumerate}
This sum equates to the definition of $\firstID(u)$. Algorithm~\ref{alg:create} computes
this sum in line~\ref{line:id}. 

\begin{algorithm}[!t]
  \begin{algorithmic}[1]
    \Function{createTableau}{Nonterminal $A$, node $u$ of $\dt(A)$}
    \If{$u = \varepsilon$}\Comment{Return empty tableau $t_\varepsilon$}
      \State $\tab_\varepsilon \gets $ empty $1 \times \maxRank$ tableau
      \State $\reg(\tab_\varepsilon(1,j)) \gets \tab(1,j)$ for $j \in [\rank(A)]$
      \State $\tab_\varepsilon^{\nt} = (A)$
      \State $\tab_\varepsilon^{\firstID} = (0)$
      \State $\tab_\varepsilon^{\offset} = 0$
      \State \textbf{return} $\tab_\varepsilon$
    \Else
      \State $u = v.l$ for some $l \in \N$\label{line:start}
      \State $\tab_v = $\textsf{createTableau}$(A,v)$ \Comment{Use tableau for $v$ as base of
      tableau for $u$.}
      \State Let $m$ be the number of rows in $\tab_v$ and $n = m+1$
      \State $\tab_u \gets$ extend $\tab_v$ by one additional row creating an $n \times \maxRank$ tableau
      \State Let $B = \tab_{u}^{\nt}(m)$ and $\NTSet_{\rhs(B)} = \{e_1,\ldots,e_k\}$ with $e_1
      \leq_{\dt} \cdots \leq_{\dt} e_k$
      \For{$j \in [\rank(e_l)]$}
        \State $x \gets \att(e_l)[j]$ \label{line:10}
        \State $\nde(\tab_u(n,j)) \gets 0$
        \If{$x$ is internal}
          \State $\reg(\tab_u(n,j)) \gets \tab_u(m,j)$\label{line:11}
          \State $\nde(\tab_u(m,j)) \gets x$ \label{line:12}
        \Else
          \State $j' \gets $position of $x$ in $\ext_{\rhs(B)}$\label{line:13}
          \State $X_j \gets \{(a,b) \mid \reg(\tab_u(a,b)) = \tab_u(m,j') $ where $a < m\}$
          \State $\reg(\tab_u(n,j)) \gets \reg(\tab_u(m,j'))$
          \State $\reg(\tab_u(m,j')) \gets \tab_u(n,j)$
          \State For all $(a,b) \in X_j:$ $\reg(\tab_u(a,b)) \gets \tab_u(n,j)$ \label{line:17}
        \EndIf
      \EndFor
      \State $\tab_u^{\nt}(n) \gets \lbl(e_l)$ 
      \State $\tab_u^{\firstID}(n) \gets \tab_u^{\firstID}(m) + |V_{g_{\dt(A),v}}| -
      |\ext_{g_{\dt(A),v}}| + \sum_{j < l}\nodes(\lbl(e_j))$ \label{line:id}
      \State \textbf{return} $\tab_u$\label{line:return}
    \EndIf
    \EndFunction
  \end{algorithmic}
  \caption{Algorithm to create a tableau.} 
  \label{alg:create}
\end{algorithm}

We can now prove Algorithm~\ref{alg:create} correct in the
following sense: 
\begin{lemma}
  For $A\in N$, $x \in V_{\rhs(A)}$, $\sigma \in \Sigma$, and indices $k,l$ let
  $\outEdge_{k,l}(A,x,\sigma) = (u, y)$. Then the tableau $t = \mathsf{createTableau}(A, u)$
  computed by Algorithm~\ref{alg:create} equals the tableau $\tbl_k(A,x,\sigma)$ as defined in
  Definition~\ref{def:tableau}.
    \label{lem:correctConst}
\end{lemma}
\begin{proof}
  To prove the claim it suffices to show that $t$ models the path to $u$ in $\dt(A)$ as defined in
  Definition~\ref{def:correctTableau}, since $\tbl_k(A,x,\sigma) = \tbl(u)$. We prove this
  inductively over $u$. For $u = \varepsilon$ the tableau $\tab_\varepsilon$ obviously
  $\dt(A)$-models $u$. Otherwise $u$ can be expressed as $u = v.l$ for some $l \in \N$.  Let the
  labels of $v$ and $u$ be $B \to g_v$ and $C \to g_u$, respectively. By induction we obtain
  $\tab_v$, which models the path from the root of $\dt(A)$ to $v$. Let $\tab_v$ have $m$ rows, and
  let $\tab_u$ be the $n = m+1$ row tableau computed by extending $\tab_v$ using
  Algorithm~\ref{alg:create}. Then $\tab_u$ is correct according to
  Definition~\ref{def:correctTableau}: Condition~1 of the definition obviously holds. As previously
  noted, Line~\ref{line:id} of Algorithm~\ref{alg:create} is just the definition of $\firstID(u)$,
  thus condition~2 is also fulfilled. This leaves condition~3. Let $j \in [\rank(C)]$ and let $x$ be
  defined as in Line~\ref{line:10} of Algorithm~\ref{alg:create}. If $x$ is internal then
  $\findNode(\tab_u(m,j))$ fulfills condition~3 by induction and $\findNode(\tab_u(n,j))$ correctly
  references the cell containing $x$ by Line~\ref{line:12} of the algorithm. If otherwise $x$ is
  external then let $\internal_{\dt(A)}(v,x) = (w,y)$ and let $x$ be the $j'$-th external node of
  $g_v$.  By induction $\findNode(\tab_v(m,j'))$ references a cell $c$ in the row representing $w$
  of $\tab_v$, which contains $y$ as its $\nde$-part. Note that it has to be an upwards pointer,
  because there is no row below $m$ in $\tab_v$. After Line~\ref{line:17} of the algorithm
  $\reg(\tab_u(n,j)) = c$ is an upwards pointer and every downwards pointer previously pointing to
  $\tab_u(m,j)$ now instead points to $\tab_u(n,j)$ and thus has the same $\findNode$-result as
  before (i.e.,~$c$). The same argument holds if $\internal_{\dt(A)}(v,x)$ is undefined, where
  $\tab_v(m,j')$ is some cell in the first row of $\tab_v$. Finally, $t^{\offset} = 0$ in
  $t_{\varepsilon}$ and remains unchanged in the further iterations of the algorithm.
\end{proof}
\subsection{Concatenating tableaux}\label{sss:attaching}
As explained in Section~\ref{sss:tableaux} above, to support arbitrary traversal, it is necessary to
concatenate tableaux to model longer paths in the derivation tree. The algorithm to do so is 
straightforward (given in Algorithm~\ref{alg:attach}). An example of this operation was already
shown in Figure~\ref{fig:traversal_tab_att_ex}. The example is of a concatenation in the intuitive
sense: $t_2$ is attached to $t_1$ at the latter's bottom row. One could however also say that the
last row of $t_1$ is in fact replaced by $t_2$. This may be a better point of view to take, because
the concatenation is not always an attachment at the bottom row. Instead a tableau may be
concatenated at any intermediate row. This happens when, due to a traversal ending at an external
vertex, the current position within the derivation tree is not represented by the bottom row of the
tableau, but instead by row $i$. If the next traversal then adds a tableau to concatenate, it needs
to be concatenated to row $i$ of the first tableau to correctly model the full branch of the
derivation tree. Algorithm~\ref{alg:attach} thus takes as a third parameter a row number specifying
from where in $t_1$ to concatenate $t_2$. The example in Figure~\ref{fig:traversal_tab_att_ex} is
the result of the operation $\mathsf{concat}(t_2, t_1, 3)$.
\begin{algorithm}[!t]
\begin{algorithmic}[1]
  \Require{$\tab_1^{\nt}(\row) = \tab_2^{\nt}(1) = B$}
    \Function{concat}{Tableaux $\tab_2$, $\tab_1$, number $\row$}
        \For{$j \in [\rank(B)]$}
          \State $c \gets \textbf{if } \reg(\tab_2(1,j))\text{ is a downwards pointer} \textbf{ then
          } \reg(\tab_2(1,j)) \textbf{ else } \tab_2(1,j)$
            \State $\reg(c) \gets \findNode(\tab_1(\row,j))$\label{importantLine}
        \EndFor
        \State $\tab_2^{\offset} \gets \tab_1^{\offset} + \tab_1^{\firstID}(\row)$
    \EndFunction
  \end{algorithmic}
\caption{Algorithm to attach $\tab_2$ to the $\row$-th row of $\tab_1$.}
\label{alg:attach}
\end{algorithm}

Concatenating a tableau $t_2$ with $m$ rows to a tableau $t_1$ at its $k$-th row results in a
``virtual'' tableau $t$ with $(k-1) + m$ rows. To this end we define for any $i \in
\{1,\ldots,(k-1)+m\}$ and $j \in [\maxRank]$ the following mappings:

\begin{align*}
  t(i,j) =&
  \begin{cases}
    t_1(i,j) & \text{if } i < k\\
    t_2(i-(k-1), j) & \text{if } i \geq k
  \end{cases} \\
  t^{\nt}(i) =&
  \begin{cases}
    t_1^{\nt}(i) & \text{if } i < k\\
    t_2^{\nt}(i-(k-1)) & \text{if } i \geq k
  \end{cases} \\
  t^{\firstID}(i) =&
  \begin{cases}
    t_1^{\firstID}(i) & \text{if } i < k\\
    t_2^{\firstID}(i-(k-1)) + t_2^{\offset} & \text{if } i \geq k
  \end{cases} \\
  t^{\offset} =& 0
\end{align*}

The intention being that we continue to use the first $k-1$ rows of $t_1$ as before, but virtually
replace row $k$ with $t_2$. Using this we can show that two concatenated tableaux model a
concatenated path in a derivation tree.

\begin{lemma}
  Let $A_1$ be a nonterminal, let $u_m$ be a node in $\dt(A_1)$, and let $u_1,\ldots,u_{m-1}$ be the
  nodes on the path from the root of $\dt(A_1)$ to $u_m$. For $i \in [m]$ let $u_i$ be labeled by
  $A_i \to g_i$. Let $t_1$ be a tableau whose first $r$ rows
  model the path $u_1,\ldots, u_r$ in $\dt(A_1)$ and let $t_2$ be a tableau that models the path
  $u_r,\ldots,u_m$ in $\dt(A_r)$. Then the virtual tableau $t$ resulting of
  $\mathsf{concat}(t_2,t_1,r)$ as per Algorithm~\ref{alg:attach} models the path from $u_1$ to $u_m$
  in $\dt(A_1)$. Computing $\mathsf{concat}(t_2,t_1,r)$ takes $O(\maxRank)$ time, where $\maxRank$
  is the maximal nonterminal rank.
  \label{lem:attachmentCorrect}
\end{lemma}
\begin{proof}
  We show the conditions of Definition~\ref{def:correctTableau} to be true for $t$ (Condition~4 is
  trivially fulfilled). 

  \paragraph*{Condition~1} As per the premise, $t_1^{\nt}(i) = A_i$ for $i \in [r-1]$ and
  $t_2^{\nt}(i-(r-1)) = A_i$ for $i \in \{r,\ldots,m\}$. Thus $t^{\nt}(i) = A_i$ for $i \in [m]$.
  
  \paragraph*{Condition~2} Again, per the premise $t_1^{\firstID}(i) + t_1^{\offset} =
  \firstID(u_i)$ for $1 \leq i < r$. After running $\mathsf{concat}(t_2, t_1, r)$ we further have
  $t_2^{\offset} = t_1^{\firstID}(r) + t_1^{\offset}$. It follows that $t_2^{\firstID}(1) +
  t_2^{\offset} = \firstID(u_r)$, since $t_1^{\firstID}(r) + t_1^{\offset} = \firstID(u_r)$ and
  $t_2^{\firstID}(1) = 0$ by Definition~\ref{def:correctTableau}. It follows that
  $t_2^{\firstID}(i-(r-1)) + t_2^{\offset} = \firstID(u_i)$ for $i \in \{r,\ldots,m\}$ and thus
  $t^{\firstID}(i) + t^{\offset} = \firstID(u_i)$ for $i \in [m]$.

  \paragraph*{Condition~3} The tableau $t_1$ is unchanged by $\mathsf{concat}(t_2,t_1,r)$ and thus
  for every $i \in [r-1]$ and $j \in [\maxRank]$ there exist indices $i' \in [i]$ and $j' \in
  [\maxRank]$ such that $\findNode(t_1(i,j)) = t_1(i',j')$ and $\internal(u_i,x) = (u_{i'},
  \nde(t_1(i',j')))$ where $x$ is the $j$-th external vertex of $g_i$.

  Similarly, most of $t_2$ is unchanged: $t_2$ is constructed such that $\findNode(t_2(i,j))$ for
  any $i,j$ yields the internal vertex merged with the $j$-th external vertex in the rule
  represented by the $i$-th row of $t_2$. Whenever such a node exists in $\dt(A_r)$ the
  corresponding $\reg$-pointers are unchanged. In all other cases the $\mathsf{concat}(t_2,t_1,r)$
  operation replaces the upwards pointers in $t_2$ by new ones referencing cells in $t_1$. These
  cells are found by taking the $\reg$-pointers of the $r$-th row of $t_1$, which models the
  node $u_r$ in $\dt(A_1)$.

  Formally let $i \in \{r,\ldots,m\}$, $j\in [\maxRank]$, and let $x$ be the $j$-th external vertex
  of $g_i$. Whenever $\internal(u_i, x) = (u_{i'}, y)$ such that $i' \in \{r, \ldots, i\}$ (i.e.,
  the target node is in $\dt(A_r)$) then $\findNode(t_2(i-(r-1),j)) = t_2(i'-(r-1), j')$ with
  $\nde(t_2(i'-(r-1),j')) = y$ by Definition~\ref{def:correctTableau}.

  This leaves cases where $\internal(u_i, x) = (u_{i'}, y)$ such that $i' \in [r-1]$. For this case
  note that before ever running a concatenation with $t_2$ the result of $\findNode(t_2(i-(r-1),j))
  = t_2(1,j')$ is the cell in the first row of $t_2$ representing the external vertex of $g_r$
  merged with the $j$-th external vertex of $g_i$ (i.e., $x$). Both of these external vertices thus
  represent the same internal vertex. Further, $\findNode(t_1(r, j')) = t_1(i', j'')$ is such that
  $\nde(t_1(i', j'')) =  y$ since $t_1$ models the path from $u_1$ to $u_r$ in $\dt(A')$. After
  line~\ref{importantLine} in Algorithm~\ref{alg:attach} $\findNode(t_2(i-(r-1),j)) =
  \findNode(t_1(r, j'))$, and thus the condition holds in these cases as well. It follows that for
  every $i \in [m]$ and $j\in [\maxRank]$ there exist indices $i' \in [i]$ and $j' \in [\maxRank]$
  such that $\findNode(t(i,j)) = t(i', j')$ and $\internal(u_i, x) = (u_{i'}, \nde(t(i',j')))$ where
  $x$ is the $j$-th external vertex of $g_i$.

  Finally the algorithm uses one call of $\findNode$ and follows one additional $\reg$ pointer per
  column of the tableaux. Since the tableaux have $\maxRank$ columns we obtain a runtime of
  $O(\maxRank)$.
\end{proof}
Note that the argument in the proof of Lemma~\ref{lem:attachmentCorrect} still holds if $t_1$ is
itself a ``virtual tableau''. It follows that concatenations of an arbitrary number of tableaux can
model a path in the derivation tree of a grammar.

\subsection{Traversals using tableaux}\label{sss:navigation}
To traverse a graph using the method outlined at the beginning of Section~\ref{sse:traversing} we
need to store the current node $u$ within $\dt(S)$. As detailed in the previous section we represent
this node using a sequence of tableaux $\tab_1,\ldots, \tab_n$. However as these tableaux are all
chained together via Algorithm~\ref{alg:attach}, we only need to store a pointer to $\tab_n =
\tbl(A,x,\sigma)$ for some $A\in N$, $x \in \rhs(A)$ and $\sigma \in \Sigma$. We call this pointer
$\curTab$. As tableaux usually represent more than one node in the derivation tree we further need a
row index to mark which row of the tableau actually represents $u$. We call this index $\pos$.
Finally our current position is a vertex $y \in g_u$, which we call $\cur$. To find the matching
internal vertex of a node $x \in g_u$ let $\internal(x, \curTab, \pos) = (x, \curTab, \pos)$ if $x$
is internal. Otherwise let $x$ be the $j$-th external vertex of $g_u$ and let $c =
\findNode(\curTab(\pos,j))$ be a cell in the $j'$-th row of $\tbl(A, y, \sigma)$ for some $A\in N$,
$y \in \rhs(A)$, and $\sigma \in \Sigma$. Then $\internal(x, \curTab, \pos) = (\nde(c),
\tbl(A,y,\sigma), j')$.
\begin{theorem}
  Given a unique-labeled SL HR grammar $G = (N,P,S)$ of bounded rank $\maxRank$, height $h$, and
  maximal terminal rank in $\val(G)$ of $r$ we can traverse the edges of the graph $\val(G)$
  starting from any node in the start graph with $O(\maxRank)$ time per traversal step. The
  precomputation uses $O(|G||\Sigma|\maxRank h^2r)$ time and $O(|G||\Sigma|\maxRank hr)$
  space.
  \label{thm:traversal}
\end{theorem}
\begin{proof}
  Let $M$ be the largest vertex number appearing in $G$. For any $A \in N$, $x \in [M]$,
  $\sigma \in \Sigma$, and $k,l \in [\rank(\sigma)]$ we precompute the tableaux $\tbl_k(A,x,\sigma)$
  (if defined, i.e., $x$ exists in $\rhs(A)$) and for $\outEdge_{k,l}(A,x,\sigma) = (v,y)$ we define
  $\success_{k,l}(A,x,\sigma) = y$, which is also stored during the precomputation. Then there are
  $|G| \cdot |\Sigma| \cdot r$ many such tableaux to compute and computing one takes
  $O(h^2\maxRank)$ time. Therefore the precomputation takes $O(|G||\Sigma|\maxRank h^2r)$ time but
  only $O(|G||\Sigma|\maxRank hr)$ space, because
  only the final tableaux constructed in Algorithm~\ref{alg:create} are stored. The
  $\success$-mapping needs an additional $O(r^2)$ space, but since $|G|r > r^2$ it is omitted in the
  upper bounds given. We then initialize the global state with $\curTab = \tab_\varepsilon$, $\pos =
  1$, and $\cur = x$ for any $x \in V_S$. Now assume we want to traverse the $k$-th outgoing
  $\sigma'$-edge of the current node towards the $l$-th neighbor, then we do the following: 
\begin{enumerate}
  \item Let $A = \curTab^{\nt}(\pos)$ and $x = \cur$.
  \item If $\tbl_k(A,x,\sigma') \neq t_\varepsilon$:\\
    $\attach(\curTab, \tbl_k(A, x, \sigma'))$ \\
    Let $\curTab = \tbl_k(A,x,\sigma')$ \\
    Let $\pos = $ index of last row of $\tbl_k(A, x, \sigma')$
  \item Let $\cur = \success_{k,l}(A, x, \sigma')$.
  \item Let $\internal(\cur, \curTab, \pos) = (y, \tab, j)$ and let $\cur = y$, $\curTab = \tab$,
    and $\pos = j$.
\end{enumerate}
  These steps take $O(1)$ time, except for the $\attach$ operation in Step 2, which takes
  $O(\maxRank)$ time. Note Step $2$ is only done if the tableau is not the empty tableau
  $\tab_\varepsilon$. The reason for this is that $\tbl_k(A,x,\sigma) = \tab_\varepsilon$ only if
  the path in the derivation tree does not change (i.e., the edge to be traversed is found in
  $\rhs(A)$). In this case only $\cur$ changes.
  Compute $\id(\cur) = \curTab^{\firstID}(\pos) + \curTab^{\offset} + \cur$ to find the
  corresponding node ID in $\val(G)$.
\end{proof}

The space needed can be given more precisely: using $\log(|G|)$-bit registers, four registers are
necessary for every cell of a tableau: three for the $\reg$ pointer (address of target tableau and
two numbers for row and column) and one for $\nde$.

%% file: figures/traverseIllus.tex
\begin{tikzpicture}
  \begin{scope}[every node/.style={rectangle, draw, minimum size=9pt, inner sep=2pt, outer sep=0pt,
    node distance=20pt}]
    \node (e) {\lblsize $\sigma$};
  \end{scope}
  \node[coordinate, above left=of e] (x1) {};
  \node[coordinate, below right=of e] (x2) {};

  \begin{scope}[every node/.style={circle, draw, minimum size=12pt, inner sep=2pt, outer sep=0pt,
    node distance=20pt}]
    \path (x1) .. controls +(1.5,1.5) and +(1.5,1.5) .. 
      node[pos=0] (u) {\idsize $x$}
      node[pos=.3] (1) {}
      node[pos=.5,draw=none] (d) {\dots}
      node[pos=.7] (2) {}
      node[pos=1] (v) {\idsize $y$}
      (x2);
  \end{scope}
  \draw (e) -- node[above right=-4pt,pos=.3] {\scriptsize $k$} (u);
  \draw (e) -- node[above right=-4pt,pos=.3] {\scriptsize $l$} (v);
  \draw (e) -- (1);
  \draw (e) -- (2);

  \draw[very thick, dotted, Blue, -stealth'] (u) .. controls +(-.5,-.5) and +(-1,-.5) .. node[sloped, below] {$\outEdge_{k,l}(x, \sigma)$}
  (v);
\end{tikzpicture}

%% file: figures/traversal_dt.tex
\begin{tikzpicture}[remember picture]
  \node[dashed, circle, draw, inner sep=0pt, outer sep=0pt] (u1)
  {
    \begin{tikzpicture}[solid, remember picture]
      \begin{scope}[every node/.style={circle, draw, minimum size=9pt, inner sep=2pt, outer sep=0pt,
        node distance=20pt}]
        \node[] (u1S1) {\idsize 1};
        \node[right=of u1S1] (u1S2) {\idsize 2};
      \end{scope}
      \begin{scope}[every path/.style={-stealth'}]
        \draw[bend left] (u1S2) to node[above] {\lblsize $A$} (u1S1);
        \draw[bend left] (u1S1) to node[above] {\lblsize $A$} (u1S2);
      \end{scope}
      \node[below left=7pt and -0pt of u1S2] {$u_1 = \varepsilon$};
    \end{tikzpicture}
  };

  \node[dashed, circle, draw, inner sep=0pt, outer sep=0pt, below left=20pt and 20pt of u1] (u2)
  {
    \begin{tikzpicture}[solid, remember picture]
      \node[] (u2A) {$A\to$};
      \begin{scope}[every node/.style={circle, draw, minimum size=9pt, inner sep=2pt, outer sep=0pt,
        node distance=20pt}]
        \node[right=5pt of u2A, fill=black, text=white] (u2A3) {\idsize 3};
        \node[right=of u2A3, rectangle, inner sep=2pt] (u2AC) {\lblsize $C$};
        \node[below=of u2AC] (u2A2) {\idsize 2};
        \node[above=of u2AC] (u2A1) {\idsize 1};
        \node[left=of u2A1, fill=black, text=white] (u2A4) {\idsize 4};
      \end{scope}
      \begin{scope}[every path/.style={-stealth'}]
        \draw (u2A1) -- node[above] {\lblsize $B$} (u2A4);
      \end{scope}
      \draw (u2AC) -- node[above, pos=.3] {\scriptsize $1$} (u2A3);
      \draw (u2AC) -- node[right, pos=.3] {\scriptsize $2$} (u2A2);
      \draw (u2AC) -- node[right, pos=.3] {\scriptsize $3$} (u2A1);
    \end{tikzpicture}
  };

  \node[below left=20pt and 5pt of u2] (u3) {$\vdots$};

  \node[dashed, circle, draw, inner sep=0pt, outer sep=0pt, below right=20pt and 5pt of u2] (u4)
  {
    \begin{tikzpicture}[solid, remember picture]
      \node[] (u4C) {$C \to$};
      \begin{scope}[every node/.style={circle, draw, minimum size=9pt, inner sep=2pt, outer sep=0pt,
        node distance=20pt}]
        \node[right=5pt of u4C, fill=black, text=white] (u4C4) {\idsize 4};
        \node[right=of u4C4] (u4C1) {\idsize 1};
        \node[above=of u4C1, fill=black, text=white] (u4C2) {\idsize 2};
        \node[below=of u4C1, fill=black, text=white] (u4C3) {\idsize 3};
      \end{scope}
      \begin{scope}[every path/.style={-stealth'}]
        \draw (u4C2) -- node[right] {\lblsize $a$} (u4C1);
        \draw (u4C1) -- node[above] {\lblsize $a$} (u4C4);
        \draw (u4C1) -- node[right] {\lblsize $D$} (u4C3);
      \end{scope}
    \end{tikzpicture}
  };

  \node[dashed, circle, draw, inner sep=0pt, outer sep=0pt, below=20pt of u4] (u5)
  {
    \begin{tikzpicture}[solid, remember picture]
      \node[] (u5D) {$D \to$};
      \begin{scope}[every node/.style={circle, draw, minimum size=9pt, inner sep=2pt, outer sep=0pt,
        node distance=20pt}]
        \node[right=5pt of u5D, fill=black, text=white] (u5D2) {\idsize 2};
        \node[right=of u5D2] (u5D1) {\idsize 1};
        \node[right=of u5D1, fill=black, text=white] (u5D3) {\idsize 3};
      \end{scope}
      \begin{scope}[every path/.style={-stealth'}]
        \draw (u5D2) -- node[above=1pt] {\lblsize $b$} (u5D1);
        \draw (u5D1) -- node[above=1pt] {\lblsize $b$} node[below=1pt] (bEdge) {} (u5D3);
        \draw[bend left=45] (u5D3) to node[below=1pt] {\lblsize $b$} (u5D2);
      \end{scope}
    \end{tikzpicture}
  };

  \node[below right=20pt and 20pt of u1] (Atree) {$\vdots$};
%  \node[below right=20pt and 20pt of u1, shape=isosceles triangle, shape border rotate=90, draw,
%  align=center, anchor=north, Green] (Atree) {$\mathsf{dt}(A)$};

  \draw (u1) -- (u2);
  \draw (u2) -- (u3);
  \draw (u2) -- (u4);
  \draw (u4) -- (u5);
  \draw (u1) -- (Atree);

  \node[right=2pt of u2.south west] {$u_2 = 1$};
  \node[right=2pt of u4.south west] {$u_4 = 1.2$};
  \node[right=2pt of u5.south west] {$u_5 = 1.2.1$};

  \begin{scope}[overlay, every path/.style={very thick, -stealth', Red, dotted}]
    \draw (u1S1) .. controls +(-4,-1) and +(-.5,.5) .. node[above left=-2pt,pos=.12]
    {$\outEdge_{1,2}(S, 1, a)$} (u2A3) .. controls +(2.5,0) and +(-.25,2.5) .. (u4C1) ;
    \draw (u4C1) .. controls +(-0.25,-2.5) and +(-2,1.75) .. node[above left=-2pt,pos=.5] {$\outEdge_{1,2}(C,
    1, b)$} (u5D2) .. controls +(.25,-.125) and +(-.25,-.125) .. (u5D1);
    \node[below right=.25 and 1 of u5D3, outer sep=1pt, inner sep=0pt] (trav) {$\outEdge_{1,2}(D, 1, b)$};
    \draw (trav) -- (u5D3);
  \end{scope}
  \begin{scope}[overlay, every path/.style={very thick, dashed, -stealth', Green}]
    \draw (u5D3) .. controls +(2,1) and +(2,-1) .. node[right=-1pt,pos=.5] {$\internal$} (u4C3) ..
    controls +(-3.5,.5) and +(0,-1) .. (u2A2);
  \end{scope}

  \node[below right=1 and 1 of u4C1, align=left, inner sep=0pt, outer sep=1pt, anchor=north west] (id1) {$\begin{aligned}
      \id(u_4, 1) &= 6.
  \end{aligned}$};

  \node[below left=1.5 and 2.5 of u2A2, align=left, inner sep=0pt, outer sep=1pt, anchor=north west] (id4) {$\begin{aligned}
      \id(u_2, 2) &= 4.
  \end{aligned}$};

  \node[below right=1 and 2 of u5D1, align=left, inner sep=0pt, outer sep=1pt, anchor=north west] (id7) {$\begin{aligned}
      \id(u_5, 1) &= 7.
  \end{aligned}$};

  \begin{scope}[overlay, every path/.style={loosely dashed, very thick, stealth'-}]
      \draw (u4C1) -- (id1.168);
      \draw (u2A2) -- (id4.45);
      \draw (u5D1) -- (id7.168);
  \end{scope}

  \begin{scope}[every node/.style={circle, draw, inner sep=1pt, outer sep=0pt, node distance=20pt,
    minimum size=15pt}]
    \node[right=4 of u4] (10) {\idsize 10};
    \node[above=of 10] (1) {\idsize 1};
    \node[above=of 1] (6) {\idsize 6};
    \node[above=of 6] (7) {\idsize 7};
    \node[above=of 7] (4) {\idsize 4};
    \node[right=2 of 1] (3) {\idsize 3};
    \node[below=of 3] (5) {\idsize 5};
    \node[below=of 5] (2) {\idsize 2};
    \node[left=of 2] (11) {\idsize 11};
    \node[below=of 11] (12) {\idsize 12};
    \node[below=of 12] (9) {\idsize 9};
    \node[left=of 11] (8) {\idsize 8};
  \end{scope}

  \begin{scope}[every path/.style={-stealth'}]
    \draw (10) -- node[left=-2pt] {\lblsize $a$} (1);
    \draw[Red] (1) -- node[left=-2pt] {\lblsize $a$} (6);
    \draw[Red] (6) -- node[left=-2pt] {\lblsize $b$} (7);
    \draw[Red] (7) -- node[left=-2pt] {\lblsize $b$} (4);
    \draw[bend right=45] (4) to node[left=-2pt] {\lblsize $b$} (6);
    \draw (6) -- node[above=-2pt] {\lblsize $a$} (3);
    \draw (3) -- node[left=-2pt] {\lblsize $a$} (5);
    \draw (5) -- node[left=-2pt] {\lblsize $a$} (2);
    \draw (2) -- node[above=-2pt] {\lblsize $a$} (11);
    \draw (11) -- node[left=-2pt] {\lblsize $b$} (12);
    \draw (12) -- node[left=-2pt] {\lblsize $b$} (9);
    \draw[bend right=45] (9) to node[left=-2pt] {\lblsize $b$} (11);
    \draw (11) -- node[above=-2pt] {\lblsize $a$} (8);
    \draw (8) -- node[left=-2pt] {\lblsize $a$} (10);
  \end{scope}
  
  \begin{scope}[overlay, every path/.style={very thick, -stealth', Red, dotted}]
    \node[right=of 6] (trav1) {$\outEdge_{1,2}(1,a)$};
    \node[right=of 7] (trav2) {$\outEdge_{1,2}(6,b)$};
    \node[right=of 4] (trav3) {$\outEdge_{1,2}(7,b)$};

    \draw (trav1) -- (6);
    \draw (trav2) -- (7);
    \draw (trav3) -- (4);
  \end{scope}
\end{tikzpicture}

%% file: figures/rootToLeafJumps.tex
\begin{tikzpicture}
  \node (S) {$S:$};
  \begin{scope}[every node/.style={circle, draw, inner sep=0pt, minimum size=12pt, node distance=15pt}]
    \node[right=0pt of S] (S1) {\idsize 1};
    \node[above=of S1, rectangle] (S2) {\lblsize $A_1$};
    \node[below=of S1, rectangle] (S3) {\lblsize $A_1$};
  \end{scope}
  \begin{scope}
    \draw (S1) --node[left=-2pt, pos=.3] {\scriptsize $1$} (S2);
    \draw (S1) --node[left=-2pt, pos=.3] {\scriptsize $1$} (S3);
  \end{scope}

  \node[right=1.5 of S] (A1) {$A_1 \to$};
  \begin{scope}[every node/.style={circle, draw, inner sep=0pt, minimum size=12pt, node distance=15pt}]
    \node[right=0pt of A1, fill=black, text=white] (A11) {\idsize 1};
    \node[above=of A11, rectangle] (A12) {\lblsize $A_2$};
    \node[below=of A11, rectangle] (A13) {\lblsize $A_2$};
  \end{scope}
  \begin{scope}
    \draw (A11) --node[left=-2pt, pos=.3] {\scriptsize $1$} (A12);
    \draw (A11) --node[left=-2pt, pos=.3] {\scriptsize $1$} (A13);
  \end{scope}

  \node[right=1.5 of A1] (A2) {$A_2 \to$};
  \begin{scope}[every node/.style={circle, draw, inner sep=0pt, minimum size=12pt, node distance=15pt}]
    \node[right=0pt of A2, fill=black, text=white] (A21) {\idsize 1};
    \node[above=of A21, rectangle] (A22) {\lblsize $A_3$};
    \node[below=of A21, rectangle] (A23) {\lblsize $A_3$};
  \end{scope}
  \begin{scope}
    \draw (A21) --node[left=-2pt, pos=.3] {\scriptsize $1$} (A22);
    \draw (A21) --node[left=-2pt, pos=.3] {\scriptsize $1$} (A23);
  \end{scope}
  
  \node[right=1.5 of A2] (dots) {$\cdots$};

  \node[right=2.5 of A2] (An) {$A_n \to$};
  \begin{scope}[every node/.style={circle, draw, inner sep=0pt, minimum size=12pt, node distance=15pt}]
    \node[right=0pt of An, fill=black, text=white] (An3) {\idsize 3};
    \node[above=of An3] (An2) {\idsize 2};
    \node[below=of An3] (An1) {\idsize 1};
  \end{scope}
  \begin{scope}[every path/.style={-stealth'}]
    \draw (An3) to[bend left] node[left=-2pt] {\lblsize $a$} (An2);
    \draw (An3) to[bend left] node[left=-2pt] {\lblsize $a$} (An1);
    \draw (An1) to[bend left] node[left=-2pt] {\lblsize $a$} (An3);
    \draw (An2) to[bend left] node[left=-2pt] {\lblsize $a$} (An3);
  \end{scope}
\end{tikzpicture}

%% file: figures/traversal_tab_ex.tex
\begin{tikzpicture}[remember picture]
  \node[dashed, circle, draw, inner sep=0pt, outer sep=0pt] (u1)
  {
    \begin{tikzpicture}[solid, remember picture]
      \begin{scope}[every node/.style={circle, draw, minimum size=9pt, inner sep=2pt, outer sep=0pt,
        node distance=20pt}]
        \node[] (u1S1) {\idsize 1};
        \node[right=of u1S1] (u1S2) {\idsize 2};
      \end{scope}
      \begin{scope}[every path/.style={-stealth'}]
        \draw[bend left] (u1S2) to node[above] {\lblsize $A$} (u1S1);
        \draw[bend left] (u1S1) to node[above] {\lblsize $A$} (u1S2);
      \end{scope}
      \node[below left=7pt and -0pt of u1S2] {$u_1 = \varepsilon$};
    \end{tikzpicture}
  };

  \node[dashed, circle, draw, inner sep=0pt, outer sep=0pt, below left=20pt and 20pt of u1] (u2)
  {
    \begin{tikzpicture}[solid, remember picture]
      \node[] (u2A) {$A\to$};
      \begin{scope}[every node/.style={circle, draw, minimum size=9pt, inner sep=2pt, outer sep=0pt,
        node distance=20pt}]
        \node[right=5pt of u2A, fill=black, text=white] (u2A3) {\idsize 3};
        \node[right=of u2A3, rectangle, inner sep=2pt] (u2AC) {\lblsize $C$};
        \node[below=of u2AC] (u2A2) {\idsize 2};
        \node[above=of u2AC] (u2A1) {\idsize 1};
        \node[left=of u2A1, fill=black, text=white] (u2A4) {\idsize 4};
      \end{scope}
      \begin{scope}[every path/.style={-stealth'}]
        \draw (u2A1) -- node[above] {\lblsize $B$} (u2A4);
      \end{scope}
      \draw (u2AC) -- node[above, pos=.3] {\scriptsize $1$} (u2A3);
      \draw (u2AC) -- node[right, pos=.3] {\scriptsize $2$} (u2A2);
      \draw (u2AC) -- node[right, pos=.3] {\scriptsize $3$} (u2A1);
    \end{tikzpicture}
  };

  \node[below left=20pt and 5pt of u2] (u3) {$\vdots$};

  \node[dashed, circle, draw, inner sep=0pt, outer sep=0pt, below right=20pt and 5pt of u2] (u4)
  {
    \begin{tikzpicture}[solid, remember picture]
      \node[] (u4C) {$C \to$};
      \begin{scope}[every node/.style={circle, draw, minimum size=9pt, inner sep=2pt, outer sep=0pt,
        node distance=20pt}]
        \node[right=5pt of u4C, fill=black, text=white] (u4C4) {\idsize 4};
        \node[right=of u4C4] (u4C1) {\idsize 1};
        \node[above=of u4C1, fill=black, text=white] (u4C2) {\idsize 2};
        \node[below=of u4C1, fill=black, text=white] (u4C3) {\idsize 3};
      \end{scope}
      \begin{scope}[every path/.style={-stealth'}]
        \draw (u4C2) -- node[right] {\lblsize $a$} (u4C1);
        \draw (u4C1) -- node[above] {\lblsize $a$} (u4C4);
        \draw (u4C1) -- node[right] {\lblsize $D$} (u4C3);
      \end{scope}
    \end{tikzpicture}
  };

  \node[dashed, circle, draw, inner sep=0pt, outer sep=0pt, below=20pt of u4] (u5)
  {
    \begin{tikzpicture}[solid, remember picture]
      \node[] (u5D) {$D \to$};
      \begin{scope}[every node/.style={circle, draw, minimum size=9pt, inner sep=2pt, outer sep=0pt,
        node distance=20pt}]
        \node[right=5pt of u5D, fill=black, text=white] (u5D2) {\idsize 2};
        \node[right=of u5D2] (u5D1) {\idsize 1};
        \node[right=of u5D1, fill=black, text=white] (u5D3) {\idsize 3};
      \end{scope}
      \begin{scope}[every path/.style={-stealth'}]
        \draw (u5D2) -- node[above=1pt] {\lblsize $b$} (u5D1);
        \draw (u5D1) -- node[above=1pt] {\lblsize $b$} node[below=1pt] (bEdge) {} (u5D3);
        \draw[bend left=45] (u5D3) to node[below=1pt] {\lblsize $b$} (u5D2);
      \end{scope}
    \end{tikzpicture}
  };

  \node[below right=20pt and 20pt of u1] (Atree) {$\vdots$};
%  \node[below right=20pt and 20pt of u1, shape=isosceles triangle, shape border rotate=90, draw,
%  align=center, anchor=north, Green] (Atree) {$\mathsf{dt}(A)$};

  \draw (u1) -- (u2);
  \draw (u2) -- (u3);
  \draw (u2) -- (u4);
  \draw (u4) -- (u5);
  \draw (u1) -- (Atree);

  \node[right=2pt of u2.south west] {$u_2 = 1$};
  \node[right=2pt of u4.south west] {$u_4 = 1.2$};
  \node[right=2pt of u5.south west] {$u_5 = 1.2.1$};

  \begin{scope}[overlay, every path/.style={very thick, -stealth', Red, dotted}]
    \draw (u1S1) .. controls +(-4,-1) and +(-.5,.5) .. node[above left=-2pt,pos=.12]
    {$\outEdge_{1,2}(S, 1, a)$} (u2A3) .. controls +(2.5,0) and +(-.25,2.5) .. (u4C1) ;
    \draw (u4C1) .. controls +(-0.25,-2.5) and +(-2,1.75) .. node[above left=-2pt,pos=.5] {$\outEdge_{1,2}(C,
    1, b)$} (u5D2) .. controls +(.25,-.125) and +(-.25,-.125) .. (u5D1);
%    \node[below right=.25 and 1 of u5D3, outer sep=1pt, inner sep=0pt] (trav) {$\outEdge_{1,2}(D, 1, b)$};
%    \draw (trav) -- (u5D3);
  \end{scope}
  \begin{scope}[overlay, every path/.style={very thick, dashed, -stealth', Green}]
%    \draw (u5D3) .. controls +(2,1) and +(2,-1) .. node[right=-1pt,pos=.5] {$\internal$} (u4C3) ..
    \draw (u4C3) .. controls +(-3.5,.5) and +(0,-1) .. node[left=-1pt,pos=.6] {$\internal$} (u2A2);
%    \draw[Orange] (u2A3) .. controls +(-2,2) and +(-1,0) .. node[left=1pt, pos=.40] {$\internal$} (u1S1);
  \end{scope}

  \node[below right=1.8 and -0.65 of u4C1, align=left, inner sep=0pt, outer sep=1pt, anchor=north west] (id6) {$\begin{aligned}
      \id(u_4, 1) &= 6.
  \end{aligned}$};

  \node[below left=2.0 and 2.5 of u2A2, align=left, inner sep=0pt, outer sep=1pt, anchor=north west] (id4) {$\begin{aligned}
      \id(u_2, 2) &= 4.
  \end{aligned}$};

%  \node[below right=1 and 2 of u5D1, align=left, inner sep=0pt, outer sep=1pt, anchor=north west] (id7) {$\begin{aligned}
%      \id(u_5, 1) &= 7.
%  \end{aligned}$};

  \begin{scope}[overlay, every path/.style={loosely dashed, very thick, stealth'-}]
    \draw (u4C1) .. controls +(1,0) and +(0,0) .. (id6.45);
    \draw (u2A2) .. controls +(-1,0) and +(0,1) .. (id4.135);
%      \draw (u5D1) -- (id7.168);
  \end{scope}

  \tabBox{right=2 of Atree}{100}{(1,1)}{1}
  \tabBox{at=(100o4)}{101}{(1,2)}{2}
  \tabBox{at=(101o4)}{102}{(1,3)}{0}
  \tabBox{below=35pt of 100o1}{110}{(2,1)}{0}
  \tabBox{at=(110o4)}{111}{(2,2)}{2}
  \tabBox{at=(111o4)}{112}{(2,3)}{1}
  \tabBox{below=35pt of 110o1}{120}{(3,1)}{0}
  \tabBox{at=(120o4)}{121}{(3,2)}{0}
  \tabBox{at=(121o4)}{122}{(3,3)}{0}

  \tabBox{below=100pt of 120o1}{200}{(1,1)}{1}
  \tabBox{at=(200o4)}{201}{(1,2)}{0}
  \tabBox{at=(201o4)}{202}{(1,3)}{0}
  \tabBox{below=35pt of 200o1}{210}{(2,1)}{0}
  \tabBox{at=(210o4)}{211}{(2,2)}{0}
  \tabBox{at=(211o4)}{212}{(2,3)}{0}

  \draw[black!20] (u1) -- (100ol);
  \draw[black!20] (u2) -- (110ol);
  \draw[black!20] (u4) -- (120ol);

  \draw[black!20] (u4) -- (200ol);
  \draw[black!20] (u5) -- (210ol);

  \node[above=2pt of 100o2, anchor=south west] (t1) {$t_1$};% = \tbl_1(S,1,a)$};
  \node[above=2pt of 200o2, anchor=south west] (t2) {$t_2$};% = \tbl_1(C,1,b)$};

  \node[right=10pt of 102ir] (t1nt1) {$S$};
  \node[right=10pt of 112ir] (t1nt2) {$A$};
  \node[right=10pt of 122ir] (t1nt3) {$C$};
  \node[at=(t1nt1 |- t1)] (t1ntl) {$\nt$};
  \node[right=10pt of t1nt1] (t1id1) {$0$};
  \node[right=10pt of t1nt2] (t1id2) {$2$};
  \node[right=10pt of t1nt3] (t1id3) {$5$};
%  \node[below=05pt of t1id3] (t1ide) {$0$};
  \node[at=(t1id1 |- t1)] (t1idl) {$\firstID$};

  \node[right=10pt of 202ir] (t2nt1) {$C$};
  \node[right=10pt of 212ir] (t2nt2) {$D$};
  \node[at=(t2 -| t2nt1)] (t2ntl) {$\nt$};
  \node[right=10pt of t2nt1] (t2id1) {$0$};
  \node[right=10pt of t2nt2] (t2id2) {$1$};
%  \node[below=05pt of t2id2] (t2ide) {$0$};
  \node[at=(t2 -| t2id1)] (t2idl) {$\firstID$};

%  \draw (t1nt1.north west) to [square right brace]  (t1nt3.south west);
%  \draw (t1nt1.north east) to [square left brace] (t1nt3.south east);
%  \draw (t1id1.north west) to [square right brace]  (t1ide.south west);
%  \draw (t1id1.north east) to [square left brace] (t1ide.south east);
%  \draw (t2nt1.north west) to [square right brace]  (t2nt2.south west);
%  \draw (t2nt1.north east) to [square left brace] (t2nt2.south east);
%  \draw (t2id1.north west) to [square right brace]  (t2ide.south west);
%  \draw (t2id1.north east) to [square left brace] (t2ide.south east);

  \begin{scope}[every path/.style={{Circle[length=3.5pt, width=3.5pt]}-{Stealth[length=4pt, width=3pt]}, shorten <=-2.8pt}]
    \draw (111pl) -- (101p);
    \draw[] (110pr) -- (120pr);
    \draw[] (120pl) .. controls +(-0.4,0.5) and +(-0.4,-0.5) .. (100p);
    \draw (111pl) -- (101p);
    \draw[Green] (121p) -- (111pr);
    \draw (122p) -- (112pr);

      \draw (210p) -- (200p);
      \draw (201pl) .. controls +(-0.4,-0.5) and +(-0.4,0.5) .. (211pl);
      \draw[] (200pl) .. controls +(-1, 1.2) and +(1,1.2) .. (200pr);
      \draw[] (211pr) .. controls +(0.4,0.5) and +(0.4,-0.5) .. (201pr);
      \draw[] (202pl) .. controls +(-1, 1.2) and +(1,1.2) .. (202pr);

%      \draw[Red] (200pl) .. controls +(-0.7,0.5) and +(-0.7,-0.5) .. (100p);
%      \draw[Red] (211pr) .. controls +(0.7,0.5) and +(0.7,-0.5) .. (111pr);
%      \draw[Red] (202pl) .. controls +(0.7,0.5) and +(0.7,-0.5) .. (112pr);
  \end{scope}
\end{tikzpicture}

%% file: figures/traversal_tab_att_ex.tex
\begin{tikzpicture}[remember picture]
  \node[dashed, circle, draw, inner sep=0pt, outer sep=0pt] (u1)
  {
    \begin{tikzpicture}[solid, remember picture]
      \begin{scope}[every node/.style={circle, draw, minimum size=9pt, inner sep=2pt, outer sep=0pt,
        node distance=20pt}]
        \node[] (u1S1) {\idsize 1};
        \node[right=of u1S1] (u1S2) {\idsize 2};
      \end{scope}
      \begin{scope}[every path/.style={-stealth'}]
        \draw[bend left] (u1S2) to node[above] {\lblsize $A$} (u1S1);
        \draw[bend left] (u1S1) to node[above] {\lblsize $A$} (u1S2);
      \end{scope}
      \node[below left=7pt and -0pt of u1S2] {$u_1 = \varepsilon$};
    \end{tikzpicture}
  };

  \node[dashed, circle, draw, inner sep=0pt, outer sep=0pt, below left=20pt and 20pt of u1] (u2)
  {
    \begin{tikzpicture}[solid, remember picture]
      \node[] (u2A) {$A\to$};
      \begin{scope}[every node/.style={circle, draw, minimum size=9pt, inner sep=2pt, outer sep=0pt,
        node distance=20pt}]
        \node[right=5pt of u2A, fill=black, text=white] (u2A3) {\idsize 3};
        \node[right=of u2A3, rectangle, inner sep=2pt] (u2AC) {\lblsize $C$};
        \node[below=of u2AC] (u2A2) {\idsize 2};
        \node[above=of u2AC] (u2A1) {\idsize 1};
        \node[left=of u2A1, fill=black, text=white] (u2A4) {\idsize 4};
      \end{scope}
      \begin{scope}[every path/.style={-stealth'}]
        \draw (u2A1) -- node[above] {\lblsize $B$} (u2A4);
      \end{scope}
      \draw (u2AC) -- node[above, pos=.3] {\scriptsize $1$} (u2A3);
      \draw (u2AC) -- node[right, pos=.3] {\scriptsize $2$} (u2A2);
      \draw (u2AC) -- node[right, pos=.3] {\scriptsize $3$} (u2A1);
    \end{tikzpicture}
  };

  \node[below left=20pt and 5pt of u2] (u3) {$\vdots$};

  \node[dashed, circle, draw, inner sep=0pt, outer sep=0pt, below right=20pt and 5pt of u2] (u4)
  {
    \begin{tikzpicture}[solid, remember picture]
      \node[] (u4C) {$C \to$};
      \begin{scope}[every node/.style={circle, draw, minimum size=9pt, inner sep=2pt, outer sep=0pt,
        node distance=20pt}]
        \node[right=5pt of u4C, fill=black, text=white] (u4C4) {\idsize 4};
        \node[right=of u4C4] (u4C1) {\idsize 1};
        \node[above=of u4C1, fill=black, text=white] (u4C2) {\idsize 2};
        \node[below=of u4C1, fill=black, text=white] (u4C3) {\idsize 3};
      \end{scope}
      \begin{scope}[every path/.style={-stealth'}]
        \draw (u4C2) -- node[right] {\lblsize $a$} (u4C1);
        \draw (u4C1) -- node[above] {\lblsize $a$} (u4C4);
        \draw (u4C1) -- node[right] {\lblsize $D$} (u4C3);
      \end{scope}
    \end{tikzpicture}
  };

  \node[dashed, circle, draw, inner sep=0pt, outer sep=0pt, below=20pt of u4] (u5)
  {
    \begin{tikzpicture}[solid, remember picture]
      \node[] (u5D) {$D \to$};
      \begin{scope}[every node/.style={circle, draw, minimum size=9pt, inner sep=2pt, outer sep=0pt,
        node distance=20pt}]
        \node[right=5pt of u5D, fill=black, text=white] (u5D2) {\idsize 2};
        \node[right=of u5D2] (u5D1) {\idsize 1};
        \node[right=of u5D1, fill=black, text=white] (u5D3) {\idsize 3};
      \end{scope}
      \begin{scope}[every path/.style={-stealth'}]
        \draw (u5D2) -- node[above=1pt] {\lblsize $b$} (u5D1);
        \draw (u5D1) -- node[above=1pt] {\lblsize $b$} node[below=1pt] (bEdge) {} (u5D3);
        \draw[bend left=45] (u5D3) to node[below=1pt] {\lblsize $b$} (u5D2);
      \end{scope}
    \end{tikzpicture}
  };

  \node[below right=20pt and 20pt of u1] (Atree) {$\vdots$};
%  \node[below right=20pt and 20pt of u1, shape=isosceles triangle, shape border rotate=90, draw,
%  align=center, anchor=north, Green] (Atree) {$\mathsf{dt}(A)$};

  \draw (u1) -- (u2);
  \draw (u2) -- (u3);
  \draw (u2) -- (u4);
  \draw (u4) -- (u5);
  \draw (u1) -- (Atree);

  \node[right=2pt of u2.south west] {$u_2 = 1$};
  \node[right=2pt of u4.south west] {$u_4 = 1.2$};
  \node[right=2pt of u5.south west] {$u_5 = 1.2.1$};

  \begin{scope}[overlay, every path/.style={very thick, -stealth', Red, dotted}]
%    \draw (u1S1) .. controls +(-4,-1) and +(-.5,.5) .. node[above left=-2pt,pos=.12]
%    {$\outEdge_{1,2}(S, 1, a)$} (u2A3) .. controls +(2.5,0) and +(-.25,2.5) .. (u4C1) ;
    \draw (u4C1) .. controls +(-0.25,-2.5) and +(-2,1.75) .. node[above left=-2pt,pos=.5] {$\outEdge_{1,2}(C,
    1, b)$} (u5D2) .. controls +(.25,-.125) and +(-.25,-.125) .. (u5D1);
%    \node[below right=.25 and 1 of u5D3, outer sep=1pt, inner sep=0pt] (trav) {$\outEdge_{1,2}(D, 1, b)$};
%    \draw (trav) -- (u5D3);
  \end{scope}
  \begin{scope}[overlay, every path/.style={very thick, dashed, -stealth', Green}]
    \draw (u5D3) .. controls +(2,1) and +(2,-1) .. node[right=1pt,pos=.25] {$\internal$} (u4C3) ..
    controls +(-3.5,.5) and +(0,-1) .. (u2A2);
%    \draw[Orange] (u2A3) .. controls +(-2,2) and +(-1,0) .. node[left=1pt, pos=.40] {$\internal$} (u1S1);
  \end{scope}

%  \node[above right=2.2 and -0.75 of u4C1, align=left, inner sep=0pt, outer sep=1pt, anchor=north west] (id6) {$\begin{aligned}
%      \id(u_4, 1) &= 6.
%  \end{aligned}$};

  \node[below left=1.5 and 2.5 of u2A2, align=left, inner sep=0pt, outer sep=1pt, anchor=north west] (id4) {$\begin{aligned}
      \id(u_2, 2) &= 4.
  \end{aligned}$};

  \node[below right=1 and 2 of u5D1, align=left, inner sep=0pt, outer sep=1pt, anchor=north west] (id7) {$\begin{aligned}
      \id(u_5, 1) &= 7.
  \end{aligned}$};

  \begin{scope}[overlay, every path/.style={loosely dashed, very thick, stealth'-}]
%    \draw (u4C1) .. controls +(1,0) and +(0,0) .. (id6.315);
      \draw (u2A2) -- (id4.45);
      \draw (u5D1) -- (id7.168);
  \end{scope}

  \tabBox{right=2 of Atree}{100}{(1,1)}{1}
  \tabBox{at=(100o4)}{101}{(1,2)}{2}
  \tabBox{at=(101o4)}{102}{(1,3)}{0}
  \tabBox{below=35pt of 100o1}{110}{(2,1)}{0}
  \tabBox{at=(110o4)}{111}{(2,2)}{2}
  \tabBox{at=(111o4)}{112}{(2,3)}{1}
  \tabBox{below=35pt of 110o1}{120}{(3,1)}{0}
  \tabBox{at=(120o4)}{121}{(3,2)}{0}
  \tabBox{at=(121o4)}{122}{(3,3)}{0}

  \tabBox{below=100pt of 120o1}{200}{(1,1)}{1}
  \tabBox{at=(200o4)}{201}{(1,2)}{0}
  \tabBox{at=(201o4)}{202}{(1,3)}{0}
  \tabBox{below=35pt of 200o1}{210}{(2,1)}{0}
  \tabBox{at=(210o4)}{211}{(2,2)}{0}
  \tabBox{at=(211o4)}{212}{(2,3)}{0}

  \draw[black!20] (u1) -- (100ol);
  \draw[black!20] (u2) -- (110ol);
  \draw[black!20] (u4) -- (120ol);

  \draw[black!20] (u4) -- (200ol);
  \draw[black!20] (u5) -- (210ol);

  \node[above=2pt of 100o2, anchor=south west] (t1) {$t_1$};% = \tbl_1(S,1,a)$};
  \node[above=2pt of 200o2, anchor=south west] (t2) {$t_2$};% = \tbl_1(C,1,b)$};

  \node[right=10pt of 102ir] (t1nt1) {$S$};
  \node[right=10pt of 112ir] (t1nt2) {$A$};
  \node[right=10pt of 122ir] (t1nt3) {$C$};
  \node[at=(t1nt1 |- t1)] (t1ntl) {$\nt$};
  \node[right=10pt of t1nt1] (t1id1) {$0$};
  \node[right=10pt of t1nt2] (t1id2) {$2$};
  \node[right=10pt of t1nt3] (t1id3) {$5$};
  \node[below=10pt of t1id3] (t1ide) {$0$};
  \node[at=(t1id1 |- t1)] (t1idl) {$\firstID$};
  \node[left=0pt of t1ide] (t1off) {offset:};

  \node[right=10pt of 202ir] (t2nt1) {$C$};
  \node[right=10pt of 212ir] (t2nt2) {$D$};
  \node[at=(t2 -| t2nt1)] (t2ntl) {$\nt$};
  \node[right=10pt of t2nt1] (t2id1) {$0$};
  \node[right=10pt of t2nt2] (t2id2) {$1$};
  \node[below=10pt of t2id2] (t2ide) {$5$};
  \node[at=(t2 -| t2id1)] (t2idl) {$\firstID$};
  \node[left=0pt of t2ide] (t2off) {offset:};

%  \draw (t1nt1.north west) to [square right brace]  (t1nt3.south west);
%  \draw (t1nt1.north east) to [square left brace] (t1nt3.south east);
%  \draw (t1id1.north west) to [square right brace]  (t1ide.south west);
%  \draw (t1id1.north east) to [square left brace] (t1ide.south east);
%  \draw (t2nt1.north west) to [square right brace]  (t2nt2.south west);
%  \draw (t2nt1.north east) to [square left brace] (t2nt2.south east);
%  \draw (t2id1.north west) to [square right brace]  (t2ide.south west);
%  \draw (t2id1.north east) to [square left brace] (t2ide.south east);

  \begin{scope}[every path/.style={{Circle[length=3.5pt, width=3.5pt]}-{Stealth[length=4pt, width=3pt]}, shorten <=-2.8pt}]
    \draw (111pl) -- (101p);
    \draw[] (110pr) -- (120pr);
    \draw[] (120pl) .. controls +(-0.4,0.5) and +(-0.4,-0.5) .. (100p);
    \draw (111pl) -- (101p);
    \draw (121p) -- (111pr);
    \draw[] (122p) -- (112pr);

      \draw (210p) -- (200p);
      \draw (201pl) .. controls +(-0.4,-0.5) and +(-0.4,0.5) .. (211pl);
      \draw[Black!30] (200pl) .. controls +(-1, 1.2) and +(1,1.2) .. (200pr);
      \draw[Black!30] (211pr) .. controls +(0.4,0.5) and +(0.4,-0.5) .. (201pr);
      \draw[Black!30] (202pl) .. controls +(-1, 1.2) and +(1,1.2) .. (202pr);

      \draw[] (200pl) .. controls +(-0.7,0.5) and +(-0.7,-0.5) .. (100p);
      \draw[Green] (211pr) .. controls +(0.7,0.5) and +(0.7,-0.5) .. (111pr);
      \draw[] (202pl) .. controls +(0.7,0.5) and +(0.7,-0.5) .. (112pr);
  \end{scope}
\end{tikzpicture}

%% file: figures/creation_illus.tex
\colorlet{color1}{BurntOrange}
\colorlet{color2}{OliveGreen}
\colorlet{color3}{RoyalBlue}
\colorlet{color4}{RedViolet}

\begin{tikzpicture}[remember picture]
  \node[draw, circle, dashed, inner sep=0pt, outer sep=0pt, fill=white, fill opacity=.7, text opacity=1] (v) {
    \begin{tikzpicture}[solid, remember picture]
      \node[] (vB) {$B\to$};
      \begin{scope}[every node/.style={circle, draw, minimum size=9pt, inner sep=2pt, outer sep=0pt,
        node distance=20pt, fill=white}]
        \node[right=5pt of vB] (vB1) {\idsize 1};
        \node[right=of vB1, rectangle, inner sep=2pt] (vBC) {\lblsize $C$};
        \node[below=of vBC, fill=black, text=white] (vB2) {\idsize 2};
        \node[above=of vBC, fill=black, text=white] (vB4) {\idsize 4};
        \node[left=of vB4, fill=black, text=white] (vB3) {\idsize 3};
      \end{scope}
      \draw (vBC) -- node[above, pos=.3] {\scriptsize $1$} (vB1);
      \draw (vBC) -- node[left, pos=.3] {\scriptsize $2$} (vB2);
      \draw (vBC) -- node[left, pos=.3] {\scriptsize $3$} (vB4);
    \end{tikzpicture}
  };
  \node[right=0pt of vBC] (e) {$e_l$};

  \node[below right=45pt and 0pt of v, dashed, circle, draw, inner sep=0pt, outer sep=0pt,
  fill=white, fill opacity=.7, text opacity=1] (u)
  {
    \begin{tikzpicture}[solid, remember picture]
      \node[] (uC) {$C\to$};
      \begin{scope}[every node/.style={circle, draw, minimum size=9pt, inner sep=2pt, outer sep=0pt,
        node distance=20pt}]
        \node[right=5pt of uC, fill=black, text=white] (uC1) {\idsize 1};
        \node[right=of uC1, fill=black, text=white] (uC3) {\idsize 3};
        \node[above= of uC1, fill=black, text=white] (uC2) {\idsize 2};
      \end{scope}
    \end{tikzpicture}
  };

  \node[below left=45pt and 0pt of v, draw, circle, dashed, minimum size=12pt, inner sep=0pt, outer
  sep = 0pt,fill = white, fill opacity=.7] (x) {};
  \node[at=(v |- u)] (dots) {$\cdots$};

  \draw (v) -- node[above right=-2pt] {$l$} (u);
  \draw (v) -- node[above left=-2pt] {$1$} (x);

  \node[above=2 of v, circle, draw, fill=black, minimum size=5pt, inner sep=0pt, outer sep=0pt] (root) {};
  \node[above=0pt of root] {$\dt(A)$};
  \draw (root) -- +(.25, -.5) -- +(-.35, -1) -- +(.5, -1.5) -- (v.north);

  \node[above=1pt of v.south] {$v$};
  \node[above=1pt of u.south] {$u=v.l$};

  \begin{scope}[every path/.style={draw}]
    \path[color=color1, bend left=10] (uC1) to (vB1);
    \path[color=color2, bend right=20] (uC2) to (vB2);
    \path[color=color3, bend right=10] (uC3) to (vB4);
  \end{scope}

  \begin{scope}[every path/.style={draw, path fading=north}]
    \path[color=color2, bend right=60] (vB2) to ($(root) + (0,-1)$);
    \path[color=color4, bend left=20] (vB3) to ($(root) + (0,-1)$);
    \path[color=color3, bend right=20] (vB4) to ($(root) + (0,-1)$);
  \end{scope}

  \node[coordinate, left=2 of root] (m) {};
  \node[coordinate, left=2em of m] (n) {};

  \tabBox{below right=2 and 5.25 of v.north}{100}{(m,1)}{0}
  \tabBox{at=(100o4)}{101}{(m,2)}{0}
  \tabBox{at=(101o4)}{102}{(m,3)}{0}

  \tabBox{below=130pt of 100o1}{200}{(m,1)}{1}
  \tabBox{at=(200o4)}{201}{(m,2)}{0}
  \tabBox{at=(201o4)}{202}{(m,3)}{0}
  \tabBox{below=35pt of 200o1}{210}{(n,1)}{0}
  \tabBox{at=(210o4)}{211}{(n,2)}{0}
  \tabBox{at=(211o4)}{212}{(n,3)}{0}

  \node[above left=30pt and 1.5em of 100o2, anchor=south west] (t1) {$t_v$};% = \tbl_1(S,1,a)$};
  \node[above left=30pt and 1.5em of 200o2, anchor=south west] (t2) {$t_u$};% = \tbl_1(C,1,b)$};

  \node[right=10pt of 102ir] (t1nt1) {$B$};
  \node[at=(t1nt1 |- t1)] (t1ntl) {$\nt$};
  \node[right=10pt of t1nt1] (t1id1) {$c_v$};
  \node[at=(t1id1 |- t1)] (t1idl) {$\firstID$};
  \node[below=10pt of t1id1] (t1ide) {$0$};
  \node[left=0pt of t1ide] (t1off) {offset:};

  \node[right=10pt of 202ir] (t2nt1) {$B$};
  \node[right=10pt of 212ir] (t2nt2) {$C$};
  \node[at=(t2nt1 |- t2)] (t2ntl) {$\nt$};
  \node[right=10pt of t2nt1] (t2id1) {$c_v$};
  \node[right=10pt of t2nt2] (t2id2) {$c_u$};
  \node[at=(t2id1 |- t2)] (t2idl) {$\firstID$};
  \node[below=10pt of t2id2.south west, anchor=north west] (t2ide) {$0$};
  \node[left=0pt of t2ide] (t2off) {offset:};

  \begin{scope}[every path/.style={path fading=north, draw}]
    \path (100o2) to +(0,1) ++(1pt,0); % <- Extends bounding box of path to non-zero value, necessary for fading.
    \path (102o3) to +(0,1) ++(1pt,0); %    (admittedly not sure how it works/what it does exactly)
    \path (200o2) to +(0,1) ++(1pt,0);
    \path (202o3) to +(0,1) ++(1pt,0);
  \end{scope}

  \begin{scope}[every path/.style={{Circle[length=3.5pt, width=3.5pt]}-, shorten <=-2.8pt, path fading=north, draw}]
    \path[bend right=05, color2] (100pr) to ($(100pr) + (.5,2)$);
    \path[bend right=05, color4] (101pr) to ($(101pr) + (.5,2)$);
    \path[bend right=05, color3] (102pr) to ($(102pr) + (.5,2)$);
    \path[bend left=05, color2] (211pr) to ($(200pr) + (.5,2)$);
    \path[bend right=05, color4] (201pr) to ($(201pr) + (.5,2)$);
    \path[bend right=05, color3] (212pr) to ($(202pr) + (.5,2)$);
  \end{scope}

  \begin{scope}[every path/.style={-{Stealth[length=4pt, width=3pt]}, shorten <=-2.8pt, path fading=north, draw}]
    \path[bend right=05, color2] ($(100pl) + (-.5,2)$) to (100pl);
    \path[bend right=05, color4] ($(101pl) + (-.5,2)$) to (101pl);
%    \path[bend right=05, color3] ($(102pl) + (-.5,2)$) to (102pl);
    \path[bend right=05, color2] ($(200pl) + (-.5,2)$) to (211p);
    \path[bend right=05, color4] ($(201pl) + (-.5,2)$) to (201pl);
%    \path[bend right=05, color3] ($(202pl) + (-.5,2)$) to (212pl);
  \end{scope}
  \begin{scope}[every path/.style={{Circle[length=3.5pt, width=3.5pt]}-{Stealth[length=4pt, width=3pt]}, shorten <=-2.8pt, draw}]
    \path[bend left=05, color1] (210pl) to (200pl);
    \path[bend left=05, color2] (200pr) to (211pl);
    \path[color3] (202p) to (212p);
  \end{scope}
  
  \path[draw, decorate, decoration=brace] (m |- v.south) to node[sloped, above] {$m$ nodes} (m);
  \path[draw, decorate, decoration=brace] (n |- u.south) to node[sloped, above] {$n = m+1$ nodes} (n);

  \node[at=(v.south east -| 101i6)] {$\Downarrow$};

  \begin{scope}[on background layer]
    \path[fill=black!10] (root) -- ($(x |- u.south) + (-1.75,-0.25)$) -- ($(u.south) + (1.75,-0.25)$) -- (root);
  \end{scope}
\end{tikzpicture}

%% file: limits.tex
\section{Limits of the method}\label{sse:limits}
We finally want to give some indication as to why it is challenging to lift the unique-label
restriction using the method presented here. Consider the grammar $G_n$ given in
Figure~\ref{fig:starGrammar1}. It generates a star-graph with one center vertex and $2^n$
$a$-neighbors. For the tableaux-method to work it would now be necessary to represent every node of
the derivation tree with only $O(n)$ precomputed tableaux. However consider the derivation tree and
tableaux given in Figure~\ref{fig:starDT} for $n=3$. Clearly the tableaux modeling half of the
derivation tree are all different. The other half (not shown in the figure) is also different but
the differences are located only in the $\nde$-values within the first rows, which is minor. That is
the fifth tableaux (counting from the left) would be identical to the first with the only difference
being that $\nde$ in cell $(0,1)$ is $3$ instead of $2$. Moreover every leaf contains a vertex that
is a neighbor of the $1$-node in the start graph. Thus if the goal is to support constant-time
navigation to the $i$-th neighbor all of these tableaux need to be precomputed. Since there are
$O(2^n)$ many such tableaux this is not feasible.
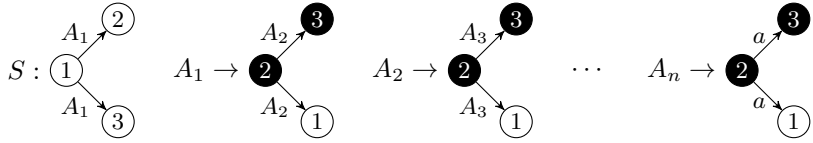
\begin{figure}[!t]
  \centering
  \input{figures/traversal_star_counterEx_grammar}
  \caption{Grammar generating a Star such that there are exponentially many different contexts in
  the derivation tree.}
  \label{fig:starGrammar1}
\end{figure}
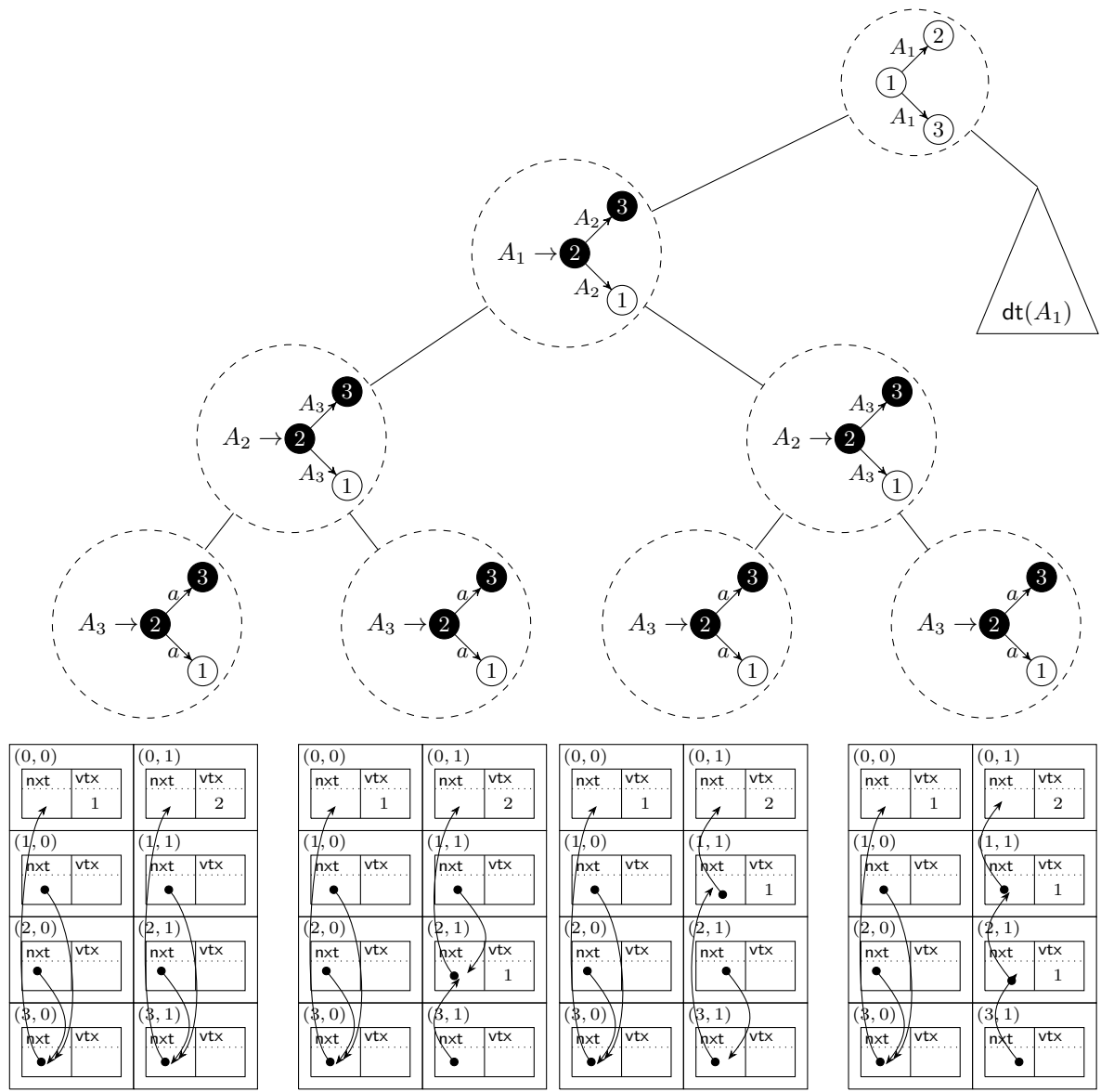
\begin{figure}[!t]
  \centering
  \input{figures/starDT}
  \caption{Derivation tree for the grammar in Figure~\ref{fig:starGrammar1} with $n = 3$. Every leaf
is annotated with the tableau that $\dt(S)$-models the leaf.}
  \label{fig:starDT}
\end{figure}

We do note however that the grammar is a somewhat unnatural way of generating such a star-pattern.
The more natural grammar in Figure~\ref{fig:starGrammar2} on the other hand has the property that
every relevant node of the derivation tree is indeed modeled by almost the same tableau: the only
difference between two tableaux here is that they use different $\firstID$-vectors. Using
precomputed offsets it would likely be possible to adapt the $\firstID$-vectors at runtime to the
concrete node.  This way only a polynomial number of tableaux would be needed to traverse this
grammar.
\begin{figure}[!t]
  \centering
  \input{figures/traversal_star_works}
  \caption{Alternative grammar equivalent to the grammar in Figure~\ref{fig:starGrammar1} but with
polynomially many distinct tableaux.}
  \label{fig:starGrammar2}
\end{figure}
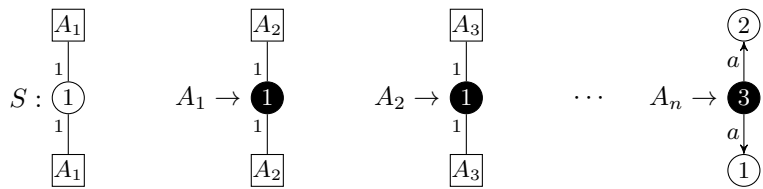

While these limitations are strong we do stress that any tree adheres to them and therefore the
method can be used to traverse trees represented by SL tree grammars without the need of a
conversion into a normal form. Furthermore if only unidirectional traversal is required the
restrictions can be lifted slightly: assume the graph $\val(G)$ for a grammar $G$ has only rank 2
edges, and we are only interested in following them from source to target. Then the unique-label
restriction only needs to apply to the source, not the target as we do not need to construct
tableaux for the other direction.

%% file: figures/traversal_star_counterEx_grammar.tex
\begin{tikzpicture}
  \node (S) {$S:$};
  \begin{scope}[every node/.style={circle, draw, inner sep=0pt, minimum size=12pt, node distance=15pt}]
    \node[right=0pt of S] (S1) {\idsize 1};
    \node[above right=of S1] (S2) {\idsize 2};
    \node[below right=of S1] (S3) {\idsize 3};
  \end{scope}
  \begin{scope}[every path/.style={-stealth'}]
    \draw (S1) --node[above left=-4pt] {\lblsize $A_1$} (S2);
    \draw (S1) --node[below left=-4pt] {\lblsize $A_1$} (S3);
  \end{scope}

  \node[right=1.5 of S] (A1) {$A_1 \to$};
  \begin{scope}[every node/.style={circle, draw, inner sep=0pt, minimum size=12pt, node distance=15pt}]
    \node[right=0pt of A1, fill=black, text=white] (A12) {\idsize 2};
    \node[above right=of A12, fill=black, text=white] (A13) {\idsize 3};
    \node[below right=of A12] (A11) {\idsize 1};
  \end{scope}
  \begin{scope}[every path/.style={-stealth'}]
    \draw (A12) --node[above left=-4pt] {\lblsize $A_2$} (A13);
    \draw (A12) --node[below left=-4pt] {\lblsize $A_2$} (A11);
  \end{scope}

  \node[right=1.5 of A1] (A2) {$A_2 \to$};
  \begin{scope}[every node/.style={circle, draw, inner sep=0pt, minimum size=12pt, node distance=15pt}]
    \node[right=0pt of A2, fill=black, text=white] (A22) {\idsize 2};
    \node[above right=of A22, fill=black, text=white] (A23) {\idsize 3};
    \node[below right=of A22] (A21) {\idsize 1};
  \end{scope}
  \begin{scope}[every path/.style={-stealth'}]
    \draw (A22) --node[above left=-4pt] {\lblsize $A_3$} (A23);
    \draw (A22) --node[below left=-4pt] {\lblsize $A_3$} (A21);
  \end{scope}
  
  \node[right=1.5 of A2] (dots) {$\cdots$};

  \node[right=2.5 of A2] (An) {$A_n \to$};
  \begin{scope}[every node/.style={circle, draw, inner sep=0pt, minimum size=12pt, node distance=15pt}]
    \node[right=0pt of An, fill=black, text=white] (An2) {\idsize 2};
    \node[above right=of An2, fill=black, text=white] (An3) {\idsize 3};
    \node[below right=of An2] (An1) {\idsize 1};
  \end{scope}
  \begin{scope}[every path/.style={-stealth'}]
    \draw (An2) --node[above left=-4pt] {\lblsize $a$} (An3);
    \draw (An2) --node[below left=-4pt] {\lblsize $a$} (An1);
  \end{scope}
\end{tikzpicture}

%% file: figures/starDT.tex
\scalebox{1}{
\begin{tikzpicture}[remember picture]
  \node[dashed, circle, draw, inner sep=0pt, outer sep=0pt] (u)
  {
    \begin{tikzpicture}[solid, remember picture]
      \begin{scope}[every node/.style={circle, draw, inner sep=0pt, minimum size=12pt, node distance=15pt}]
        \node (S1) {\idsize 1};
        \node[above right=of S1] (S2) {\idsize 2};
        \node[below right=of S1] (S3) {\idsize 3};
      \end{scope}
      \begin{scope}[every path/.style={-stealth'}]
        \draw (S1) --node[above left=-0pt] {\lblsize $A_1$} (S2);
        \draw (S1) --node[below left=-0pt] {\lblsize $A_1$} (S3);
      \end{scope}
    \end{tikzpicture}
  };

  \node[dashed, below left=.75 and 3.25 of u, circle, draw, inner sep=0pt, outer sep=0pt] (u1) {
    \begin{tikzpicture}[solid, remember picture]
      \node (1A1) {$A_1 \to$};
      \begin{scope}[every node/.style={circle, draw, inner sep=0pt, minimum size=12pt, node distance=15pt}]
        \node[right=0pt of 1A1, fill=black, text=white] (1A12) {\idsize 2};
        \node[above right=of 1A12, fill=black, text=white] (1A13) {\idsize 3};
        \node[below right=of 1A12] (1A11) {\idsize 1};
      \end{scope}
      \begin{scope}[every path/.style={-stealth'}]
        \draw (1A12) --node[above left=-0pt] {\lblsize $A_2$} (1A13);
        \draw (1A12) --node[below left=-0pt] {\lblsize $A_2$} (1A11);
      \end{scope}
    \end{tikzpicture}
  };

  \node[below right=.75 and 1 of u, shape=isosceles triangle, shape border rotate=90, draw,
  align=center, anchor=north] (Atree) {$\mathsf{dt}(A_1)$};

  \node[dashed, below left=.75 and 2 of u1, circle, draw, inner sep=0pt, outer sep=0pt] (u11) {
    \begin{tikzpicture}[solid, remember picture]
      \node (11A2) {$A_2 \to$};
      \begin{scope}[every node/.style={circle, draw, inner sep=0pt, minimum size=12pt, node distance=15pt}]
        \node[right=0pt of 11A2, fill=black, text=white] (11A22) {\idsize 2};
        \node[above right=of 11A22, fill=black, text=white] (11A23) {\idsize 3};
        \node[below right=of 11A22] (11A21) {\idsize 1};
      \end{scope}
      \begin{scope}[every path/.style={-stealth'}]
        \draw (11A22) --node[above left=-0pt] {\lblsize $A_3$} (11A23);
        \draw (11A22) --node[below left=-0pt] {\lblsize $A_3$} (11A21);
      \end{scope}
    \end{tikzpicture}
  };

  \node[dashed, below right=.75 and 2 of u1, circle, draw, inner sep=0pt, outer sep=0pt] (u12) {
    \begin{tikzpicture}[solid, remember picture]
      \node (12A2) {$A_2 \to$};
      \begin{scope}[every node/.style={circle, draw, inner sep=0pt, minimum size=12pt, node distance=15pt}]
        \node[right=0pt of 12A2, fill=black, text=white] (12A22) {\idsize 2};
        \node[above right=of 12A22, fill=black, text=white] (12A23) {\idsize 3};
        \node[below right=of 12A22] (12A21) {\idsize 1};
      \end{scope}
      \begin{scope}[every path/.style={-stealth'}]
        \draw (12A22) --node[above left=-0pt] {\lblsize $A_3$} (12A23);
        \draw (12A22) --node[below left=-0pt] {\lblsize $A_3$} (12A21);
      \end{scope}
    \end{tikzpicture}
  };

  \node[dashed, below left=.75 and .15 of u11, circle, draw, inner sep=0pt, outer sep=0pt] (u111) {
    \begin{tikzpicture}[solid, remember picture]
      \node (111A3) {$A_3 \to$};
      \begin{scope}[every node/.style={circle, draw, inner sep=0pt, minimum size=12pt, node distance=15pt}]
        \node[right=0pt of 111A3, fill=black, text=white] (111A32) {\idsize 2};
        \node[above right=of 111A32, fill=black, text=white] (111A33) {\idsize 3};
        \node[below right=of 111A32] (111A31) {\idsize 1};
      \end{scope}
      \begin{scope}[every path/.style={-stealth'}]
        \draw (111A32) --node[above left=-0pt] {\lblsize $a$} (111A33);
        \draw (111A32) --node[below left=-0pt] {\lblsize $a$} (111A31);
      \end{scope}
    \end{tikzpicture}
  };

  \node[dashed, below right=.75 and .15 of u11, circle, draw, inner sep=0pt, outer sep=0pt] (u112) {
    \begin{tikzpicture}[solid, remember picture]
      \node (112A3) {$A_3 \to$};
      \begin{scope}[every node/.style={circle, draw, inner sep=0pt, minimum size=12pt, node distance=15pt}]
        \node[right=0pt of 112A3, fill=black, text=white] (112A32) {\idsize 2};
        \node[above right=of 112A32, fill=black, text=white] (112A33) {\idsize 3};
        \node[below right=of 112A32] (112A31) {\idsize 1};
      \end{scope}
      \begin{scope}[every path/.style={-stealth'}]
        \draw (112A32) --node[above left=-0pt] {\lblsize $a$} (112A33);
        \draw (112A32) --node[below left=-0pt] {\lblsize $a$} (112A31);
      \end{scope}
    \end{tikzpicture}
  };

  \node[dashed, below left=.75 and .15 of u12, circle, draw, inner sep=0pt, outer sep=0pt] (u121) {
    \begin{tikzpicture}[solid, remember picture]
      \node (121A3) {$A_3 \to$};
      \begin{scope}[every node/.style={circle, draw, inner sep=0pt, minimum size=12pt, node distance=15pt}]
        \node[right=0pt of 121A3, fill=black, text=white] (121A32) {\idsize 2};
        \node[above right=of 121A32, fill=black, text=white] (121A33) {\idsize 3};
        \node[below right=of 121A32] (121A31) {\idsize 1};
      \end{scope}
      \begin{scope}[every path/.style={-stealth'}]
        \draw (121A32) --node[above left=-0pt] {\lblsize $a$} (121A33);
        \draw (121A32) --node[below left=-0pt] {\lblsize $a$} (121A31);
      \end{scope}
    \end{tikzpicture}
  };

  \node[dashed, below right=.75 and .15 of u12, circle, draw, inner sep=0pt, outer sep=0pt] (u122) {
    \begin{tikzpicture}[solid, remember picture]
      \node (122A3) {$A_3 \to$};
      \begin{scope}[every node/.style={circle, draw, inner sep=0pt, minimum size=12pt, node distance=15pt}]
        \node[right=0pt of 122A3, fill=black, text=white] (122A32) {\idsize 2};
        \node[above right=of 122A32, fill=black, text=white] (122A33) {\idsize 3};
        \node[below right=of 122A32] (122A31) {\idsize 1};
      \end{scope}
      \begin{scope}[every path/.style={-stealth'}]
        \draw (122A32) --node[above left=-0pt] {\lblsize $a$} (122A33);
        \draw (122A32) --node[below left=-0pt] {\lblsize $a$} (122A31);
      \end{scope}
    \end{tikzpicture}
  };

  \tabBox{below left=2 and 1 of u111.south west}{11100}{(0,0)}{1}
  \tabBox{at=(11100o4)}{11101}{(0,1)}{2}
  \tabBox{below=35pt of 11100o1}{11110}{(1,0)}{}
  \tabBox{at=(11110o4)}{11111}{(1,1)}{}
  \tabBox{below=35pt of 11110o1}{11120}{(2,0)}{}
  \tabBox{at=(11120o4)}{11121}{(2,1)}{}
  \tabBox{below=35pt of 11120o1}{11130}{(3,0)}{}
  \tabBox{at=(11130o4)}{11131}{(3,1)}{}
  \begin{scope}[every path/.style={{Circle[length=3.5pt, width=3.5pt]}-{Stealth[length=4pt, width=3pt]}, shorten <=-3pt}]
    \draw (11110p) .. controls +(0.4,-0.5) and +(0.4,0.5) .. (11130pr);
    \draw (11120pl) .. controls +(0.4,-0.5) and +(0.4,0.5) .. (11130p);
    \draw (11130pl) .. controls +(-0.4,0.5) and +(-0.4,-0.5) .. (11100p);
    \draw (11111p) .. controls +(0.4,-0.5) and +(0.4,0.5) .. (11131pr);
    \draw (11121pl) .. controls +(0.4,-0.5) and +(0.4,0.5) .. (11131p);
    \draw (11131pl) .. controls +(-0.4,0.5) and +(-0.4,-0.5) .. (11101p);
  \end{scope}

  \tabBox{below left=2 and 1 of u112.south west}{11200}{(0,0)}{1}
  \tabBox{at=(11200o4)}{11201}{(0,1)}{2}
  \tabBox{below=35pt of 11200o1}{11210}{(1,0)}{}
  \tabBox{at=(11210o4)}{11211}{(1,1)}{}
  \tabBox{below=35pt of 11210o1}{11220}{(2,0)}{}
  \tabBox{at=(11220o4)}{11221}{(2,1)}{1}
  \tabBox{below=35pt of 11220o1}{11230}{(3,0)}{}
  \tabBox{at=(11230o4)}{11231}{(3,1)}{}
  \begin{scope}[every path/.style={{Circle[length=3.5pt, width=3.5pt]}-{Stealth[length=4pt, width=3pt]}, shorten <=-3pt}]
    \draw (11210p) .. controls +(0.4,-0.5) and +(0.4,0.5) .. (11230pr);
    \draw (11220pl) .. controls +(0.4,-0.5) and +(0.4,0.5) .. (11230p);
    \draw (11230pl) .. controls +(-0.4,0.5) and +(-0.4,-0.5) .. (11200p);
    \draw (11211p) .. controls +(0.4,-0.5) and +(0.4,0.5) .. (11221pr);
    \draw (11221pl) .. controls +(-0.4,0.5) and +(-0.4,-0.5) .. (11201p);
    \draw (11231pl) .. controls +(-0.4,0.5) and +(-0.4,-0.5) .. (11221p);
  \end{scope}

  \tabBox{below left=2 and 1 of u121.south west}{12100}{(0,0)}{1}
  \tabBox{at=(12100o4)}{12101}{(0,1)}{2}
  \tabBox{below=35pt of 12100o1}{12110}{(1,0)}{}
  \tabBox{at=(12110o4)}{12111}{(1,1)}{1}
  \tabBox{below=35pt of 12110o1}{12120}{(2,0)}{}
  \tabBox{at=(12120o4)}{12121}{(2,1)}{}
  \tabBox{below=35pt of 12120o1}{12130}{(3,0)}{}
  \tabBox{at=(12130o4)}{12131}{(3,1)}{}
  \begin{scope}[every path/.style={{Circle[length=3.5pt, width=3.5pt]}-{Stealth[length=4pt, width=3pt]}, shorten <=-3pt}]
    \draw (12110p) .. controls +(0.4,-0.5) and +(0.4,0.5) .. (12130pr);
    \draw (12120pl) .. controls +(0.4,-0.5) and +(0.4,0.5) .. (12130p);
    \draw (12130pl) .. controls +(-0.4,0.5) and +(-0.4,-0.5) .. (12100p);
    \draw (12111p) .. controls +(-0.4,0.5) and +(-0.4,-0.5) .. (12101p);
    \draw (12121pr) .. controls +(0.4,-0.5) and +(0.4,0.5) .. (12131pr);
    \draw (12131pl) .. controls +(-0.4,0.5) and +(-0.4,-0.5) .. (12111pl);
  \end{scope}

  \tabBox{below left=2 and 1 of u122.south west}{12200}{(0,0)}{1}
  \tabBox{at=(12200o4)}{12201}{(0,1)}{2}
  \tabBox{below=35pt of 12200o1}{12210}{(1,0)}{}
  \tabBox{at=(12210o4)}{12211}{(1,1)}{1}
  \tabBox{below=35pt of 12210o1}{12220}{(2,0)}{}
  \tabBox{at=(12220o4)}{12221}{(2,1)}{1}
  \tabBox{below=35pt of 12220o1}{12230}{(3,0)}{}
  \tabBox{at=(12230o4)}{12231}{(3,1)}{}
  \begin{scope}[every path/.style={{Circle[length=3.5pt, width=3.5pt]}-{Stealth[length=4pt, width=3pt]}, shorten <=-3pt}]
    \draw (12210p) .. controls +(0.4,-0.5) and +(0.4,0.5) .. (12230pr);
    \draw (12220pl) .. controls +(0.4,-0.5) and +(0.4,0.5) .. (12230p);
    \draw (12230pl) .. controls +(-0.4,0.5) and +(-0.4,-0.5) .. (12200p);
    \draw (12211pl) .. controls +(-0.4,0.5) and +(-0.4,-0.5) .. (12201pl);
    \draw (12221p) .. controls +(-0.4,0.5) and +(-0.4,-0.5) .. (12211p);
    \draw (12231pr) .. controls +(-0.4,0.5) and +(-0.4,-0.5) .. (12221pr);
  \end{scope}

  \draw (u) -- (u1);
  \draw (u) -- (Atree.north);
  \draw (u1) -- (u11);
  \draw (u1) -- (u12);
  \draw (u11) -- (u111);
  \draw (u11) -- (u112);
  \draw (u12) -- (u121);
  \draw (u12) -- (u122);
\end{tikzpicture}
}

%% file: figures/traversal_star_works.tex
\begin{tikzpicture}
  \node (S) {$S:$};
  \begin{scope}[every node/.style={circle, draw, inner sep=0pt, minimum size=12pt, node distance=15pt}]
    \node[right=0pt of S] (S1) {\idsize 1};
    \node[above=of S1, rectangle] (S2) {\lblsize $A_1$};
    \node[below=of S1, rectangle] (S3) {\lblsize $A_1$};
  \end{scope}
  \begin{scope}
    \draw (S1) --node[left=-2pt, pos=.3] {\scriptsize $1$} (S2);
    \draw (S1) --node[left=-2pt, pos=.3] {\scriptsize $1$} (S3);
  \end{scope}

  \node[right=1.5 of S] (A1) {$A_1 \to$};
  \begin{scope}[every node/.style={circle, draw, inner sep=0pt, minimum size=12pt, node distance=15pt}]
    \node[right=0pt of A1, fill=black, text=white] (A11) {\idsize 1};
    \node[above=of A11, rectangle] (A12) {\lblsize $A_2$};
    \node[below=of A11, rectangle] (A13) {\lblsize $A_2$};
  \end{scope}
  \begin{scope}
    \draw (A11) --node[left=-2pt, pos=.3] {\scriptsize $1$} (A12);
    \draw (A11) --node[left=-2pt, pos=.3] {\scriptsize $1$} (A13);
  \end{scope}

  \node[right=1.5 of A1] (A2) {$A_2 \to$};
  \begin{scope}[every node/.style={circle, draw, inner sep=0pt, minimum size=12pt, node distance=15pt}]
    \node[right=0pt of A2, fill=black, text=white] (A21) {\idsize 1};
    \node[above=of A21, rectangle] (A22) {\lblsize $A_3$};
    \node[below=of A21, rectangle] (A23) {\lblsize $A_3$};
  \end{scope}
  \begin{scope}
    \draw (A21) --node[left=-2pt, pos=.3] {\scriptsize $1$} (A22);
    \draw (A21) --node[left=-2pt, pos=.3] {\scriptsize $1$} (A23);
  \end{scope}
  
  \node[right=1.5 of A2] (dots) {$\cdots$};

  \node[right=2.5 of A2] (An) {$A_n \to$};
  \begin{scope}[every node/.style={circle, draw, inner sep=0pt, minimum size=12pt, node distance=15pt}]
    \node[right=0pt of An, fill=black, text=white] (An3) {\idsize 3};
    \node[above=of An3] (An2) {\idsize 2};
    \node[below=of An3] (An1) {\idsize 1};
  \end{scope}
  \begin{scope}[every path/.style={-stealth'}]
    \draw (An3) --node[left=-2pt] {\lblsize $a$} (An2);
    \draw (An3) --node[left=-2pt] {\lblsize $a$} (An1);
  \end{scope}
\end{tikzpicture}

%% file: conclusions.tex
\section{Conclusions}
We have presented a data structure which supports the traversal of hypergraphs given as
straight-line hyperedge replacement (SL HR) graph grammars.  A single traversal step along one
hyperedge is carried out in constant time.  Our method requires two restrictions: 
\begin{enumerate}
  \item each node can be the $k$-th incident node of at most one $\sigma$-labeled hyperedge for
    every edge label $\sigma$, and 
  \item the number of nodes incident with any nonterminal edge is bounded by a constant.
\end{enumerate}
In the future we would like to investigate whether Restriction~1 can be removed or at least
weakened. We would like to implement our data structure.  It will be interesting to experimentally
compare the traversal of trees represented via our data structure, against existing solutions for
the constant delay traversal of grammar-compressed trees.  A direct (more or less naive)
implementation for grammar-compressed trees is investigated in a technical
report~\cite{DBLP:journals/corr/abs-1012-5696}; it was rudimentally compared against a constant-time
solution (using an implementation by Sadakane of succinct trees
\cite{DBLP:conf/alenex/ArroyueloCNS10,DBLP:journals/talg/NavarroS14} for solving the ``next link
problem'', cf. Section~\ref{sss:naive_delay}), giving slower times (on average) for the constant-time
solution~\cite{cla10}. A comparison with the top-tree method~\cite{DBLP:journals/iandc/BilleGLW15}
would also be very interesting, since this method allows for traversal with logarithmic delay
without any additional structures.